\documentclass[aps,prx,reprint,groupedaddress,superscriptaddress,longbibliography]{revtex4-2}
\pdfoutput=1
\usepackage{caption}
\usepackage{physics}
\usepackage{amsmath,amsthm}
\usepackage{graphicx}
\usepackage{bm}
\usepackage{bbm}
\usepackage{hyperref}
\usepackage{IEEEtrantools}

\DeclareMathOperator*{\minimize}{minimize:}
\DeclareMathOperator*{\maximize}{maximize:}

\newcommand{\densityH}[1][\mathcal{H}]{\tilde{D}(#1)}
\newcommand{\normdensityH}[1][\mathcal{H}]{D(#1)}
\newcommand{\HH}{\ensuremath{\mathcal{H}}}
\newcommand{\Hinf}{\ensuremath{\mathcal{H}_\infty}}
\newcommand{\Hfin}{\ensuremath{\mathcal{H}_N}}
\newcommand{\reduced}{\tau_A}

\newcommand{\meas}{\Phi_M^{AB}}
\newcommand{\nondest}{\Phi_T^{AB}}
\newcommand{\sapub}{S_a} 
\newcommand{\sbpub}{S_b}
\newcommand{\sapriv}{S_\alpha}
\newcommand{\sbpriv}{S_\beta}
\newcommand{\rhoo}{\tilde{\rho}^\infty}
\newcommand{\rhon}{\ensuremath{\tilde{\rho}^\Pi}} 
 
\newcommand{\nrho}[1][N]{\ensuremath{\tilde{\rho}^{#1}}}
\newcommand{\rhot}{\tilde{\rho}}
\newcommand{\sigmat}{\tilde{\sigma}}
\newcommand{\taut}{\tilde{\tau}}
\newcommand{\sigman}{\ensuremath{\tilde{\sigma}^\Pi}} 
\newcommand{\sinf}{\ensuremath{\mathbf{S}_\infty}}
\newcommand{\sfin}{\ensuremath{\mathbf{S}_N}}
\newcommand{\pos}[1]{\textrm{Pos}\left(#1\right)} 
\newcommand{\fqkd}{f^{QKD}}

\newcommand{\e}{\epsilon}
\newcommand{\ee}{\epsilon'}
\newcommand{\eee}{\epsilon''}
\newcommand{\disp}[1]{\hat{D}\left(#1\right)}
\newcommand{\dispd}[1]{\hat{D}^\dagger\left(#1\right)}
\newcommand{\obsn}{\hat{n}_{\beta_i}} 
\newcommand{\obsnk}{\hat{n}_{\beta_k}} 
\newcommand{\obsnsq}{\hat{n}^2_{\beta_i}} 
\newcommand{\obsnsqk}{\hat{n}^2_{\beta_k}} 
\newcommand{\expn}{\left\langle\obsn\right\rangle}
\newcommand{\expnsq}{\left\langle\obsnsq\right\rangle}
\newcommand{\expneff}{\left\langle\obsn\right\rangle^\text{eff}}
\newcommand{\expnsqeff}{\left\langle\obsnsq\right\rangle^\text{eff}}
\newcommand{\PiN}{\Pi^N_{B_{\beta_i}}}
\newcommand{\PiNn}{\bar{\Pi}^N_{B_{\beta_i}}}
\newcommand{\noisy}[1]{\left[#1 \right]'}
\newcommand{\nsyobsn}{\noisy{\hat{n}_{\sqrt{\eta_d} \beta_i}}}
\newcommand{\nsyobsnsq}{\noisy{\hat{n}^2_{\sqrt{\eta_d} \beta_i}}}
\newcommand{\nsyexpn}{\left\langle\nsyobsn\right\rangle}
\newcommand{\nsyexpnsq}{\left\langle\nsyobsnsq\right\rangle}
\newcommand{\vel}{\nu_{el}} 
\newcommand{\raising}{\hat{a}^\dagger}
\newcommand{\lowering}{\hat{a}}
\newcommand{\Hflag}{\mathcal{H}_F}

\newtheorem{thm}{Theorem}
\newtheorem{cor}{Corollary}
\newtheorem{lemma}{Lemma}

\begin{document}

\title{Dimension Reduction in Quantum Key Distribution for Continuous- and Discrete-Variable Protocols}

\author{Twesh Upadhyaya}
\email{twesh.upadhyaya@uwaterloo.ca}
\affiliation{Institute for Quantum Computing and Department of Physics and Astronomy\\University of Waterloo, Waterloo, Ontario, Canada N2L 3G1}
\author{Thomas van Himbeeck}
\affiliation{Institute for Quantum Computing and Department of Physics and Astronomy\\University of Waterloo, Waterloo, Ontario, Canada N2L 3G1}
\affiliation{Department of Electrical \& Computer Engineering, University of Toronto, Toronto, Ontario, Canada M5S 3G4}
\author{Jie Lin}
\affiliation{Institute for Quantum Computing and Department of Physics and Astronomy\\University of Waterloo, Waterloo, Ontario, Canada N2L 3G1}
\author{Norbert L\"{u}tkenhaus}
\affiliation{Institute for Quantum Computing and Department of Physics and Astronomy\\University of Waterloo, Waterloo, Ontario, Canada N2L 3G1}

\date{\today}

\begin{abstract}
We develop a method to connect the infinite-dimensional description of optical continuous-variable quantum key distribution (QKD) protocols to a finite-dimensional formulation. The secure key rates of the optical QKD protocols can then be evaluated using recently-developed reliable numerical methods for key rate calculations. We apply this method to obtain asymptotic key rates for discrete-modulated continuous-variable QKD protocols, which are of practical significance due to their experimental simplicity and potential for large-scale deployment in quantum-secured networks. Importantly, our security proof does not require the photon-number cutoff assumption relied upon in previous works. We also demonstrate that our method can provide practical advantages over the flag-state squasher when applied to discrete-variable protocols.
\end{abstract}

\maketitle

\section{Introduction} \label{Intro}
Quantum key distribution (QKD) \cite{Bennett1984,Ekert1991} enables two remote parties, Alice and Bob, to establish information-theoretically secure keys even in the presence of an eavesdropper, Eve. These keys can then be used in many other cryptographic applications, such as the one-time pad. To prove QKD to be secure with a specified key rate,
we need to assume a quantum-mechanical model for Alice and Bob's devices, but do not have to assume anything about the processing power available to the eavesdropper \cite{Colbeck2011}. Reviews of QKD protocols can be found in \cite{Scarani2009,Xu2020,Pirandola2020}. 

For a given QKD protocol, the goal of a security proof is to find a lower bound on the secure key rate. Analytical methods for this task can be very involved, tend to be restricted to symmetric protocols, and can introduce looseness in the lower bounds. These issues are ameliorated by the recent development of tight, reliable numerical methods for finding secure key rates \cite{Coles2016,Winick2018}. At a high level, these methods determine the key rate by solving a particular convex optimization over the set of quantum states that could be held by Alice and Bob. Thus, when the bipartite Hilbert space is infinite-dimensional, the numerical methods cannot be used directly. Fortunately, for many discrete-variable (DV) protocols, the numerical methods can be applied by using the squashing map \cite{Beaudry2008,Tsurumaru2010,Gittsovich2014} or the more general flag-state squasher \cite{Zhang2021} to reduce the problem to finite dimensions. However, these squashing approaches do not seem applicable to continuous-variable (CV) protocols. Additionally, even for DV protocols, the flag-state squasher can have challenging runtimes \cite{Li2020}.

Discrete-modulated continuous-variable QKD (DMCVQKD) is a family of protocols that utilize existing telecommunication infrastructure, including homodyne or conjugate homodyne detection \cite{Ralph1999,Hillery2000,Silberhorn2002}. They are thus promising candidates for deployment in large scale quantum-secured networks. In comparison to Gaussian-modulated CVQKD \cite{Grosshans2002,Grosshans2003,Weedbrook2004}, discrete modulation is less demanding on the source modulator and on the error-correction protocols, yet is expected to achieve similar key rates. It is thus of interest to establish security proofs for DMCVQKD. Of particular interest is DMCVQKD with four or more modulated states, which is expected to outperform protocols with a smaller constellation. While there are analytic asymptotic security proofs of DMCVQKD with two \cite{Zhao2009} or three \cite{Bradler2018} modulated states, they are difficult to generalize to more states. Recent works have numerically studied asymptotic security proofs for DMCVQKD with any number of modulated states \cite{Ghorai2019,Lin2019}. However, these approaches assume the state is finite-dimensional, known as the photon-number cutoff assumption. Thus, while these results seem numerically plausible, they do not constitute a rigorous asymptotic security proof, as the cutoff assumption cannot be justified. A full finite-key analysis of binary-modulated DMCVQKD has also been recently completed in {\cite{Matsuura2021}}. 

The main contribution of this paper is a method to tightly lower bound the key rate of an infinite-dimensional QKD protocol in terms of a finite-dimensional convex optimization. In combination with existing numerical tools for solving finite-dimensional convex optimizations, this enables us to find tight, reliable key rates for general device-dependent QKD protocols in infinite-dimensional Hilbert spaces. Our dimension reduction method can also be applied to study other quantum information tasks, such as entanglement verification \cite{Killoran2011}.

As a result, our method can provide a complete asymptotic security proof for discrete-modulated continuous-variable protocols with any number of modulated states, with tight key rates and without relying on the photon-number cutoff assumption. While our focus in this work is on calculating asymptotic key rates, we expect key elements of our method to lift to a finite-key analysis. 

Our dimension reduction method also provides an alternative approach to study protocols admitting a flag-state squasher. We consider unbalanced phase-encoded BB84 as an example, and show that our method can have an improved runtime compared to the flag-state squasher, while providing similar results.

The remainder of this paper is structured as follows. In Sec. \ref{keyrateoptim}, we review the basic steps of a QKD protocol and how the key rate can be formulated as a convex optimization. In Sec. \ref{framework}, we develop our framework for dimension reduction, in more generality than is needed for QKD. In Sec. \ref{frameworkapp}, we then specialize our general method to asymptotic key rate calculations. In Sec. \ref{numerical}, we show how to implement the relevant optimizations numerically. In Sec. \ref{dmcvqkdW}, we calculate key rates for DMCVQKD, including modelling postselection and trusted noise. In Sec. \ref{fss} we compare our method to the flag-state squasher. Finally, we provide concluding remarks and avenues for future work in Sec. \ref{conclusion}. Certain technical details are relegated to the Appendices.

\section{Background: QKD Protocols and Security Analysis} \label{keyrateoptim}
\subsection{Generic QKD Protocol Steps}
We first review the basic steps of a generic QKD protocol. Alice and Bob have access to an uncharacterized quantum channel and an authenticated, public classical channel. By the source-replacement scheme \cite{Bennett1992,Grosshans2003SR,Curty2004,Ferenczi2012}, any prepare-and-measure (P\&M) protocol can be equivalently viewed as an entanglement-based (EB) one. Thus, without loss of generality we consider entanglement-based protocols.
\begin{enumerate}
\item Alice and Bob establish a quantum state $\rho_{AB}$. 
\item Alice and Bob measure their subsystems with positive operator-valued measures (POVMs) $\{P_A^i\}$ and $\{P_B^j\}$. To each outcome $i,j$, they associate two pieces of classical data: a public announcement $a_i$, $b_j$ and a private measurement result $\alpha_i$, $\beta_j$. The respective alphabets from which the values are drawn are $\sapub, \sbpub, \sapriv,\sbpriv$. We can think of the classical announcements as partitioning the data.
\end{enumerate}
After repeating the previous two steps for a large number of rounds, Alice and Bob proceed to the classical phase. 
\begin{enumerate}
\setcounter{enumi}{2}
\item Alice and Bob choose a random subset of the rounds to use for parameter estimation. For these testing rounds, they announce their public and private results. This allows them to determine the expectations $\gamma_i$ of some testing observables $\Gamma_i$.
\item Alice and Bob announce their public data. Based on the joint announcements, they may decide to discard some rounds. This is represented by a binary function $d: \sapub \times \sbpub \rightarrow \{0,1\} $; where $d=0$ if the signal is kept and $d=1$ if discarded. 
\item Based on the public announcements and their private data, one party performs the \emph{key map}. When Alice (Bob) performs the key map, it is conventionally referred to as direct (reverse) reconciliation. In the following discussion, we consider the case where Bob performs the key map. For a $k$-ary key, the key map is a function $g: \sapub \times \sbpub \times \sbpriv \rightarrow \{0,1,...k-1,\perp\}$. The $\perp$ symbol is only used to flag the discarded or \emph{sifted} signals, so $g(a,b,\beta)=\ \perp \iff d(a,b)=1$. 
\item Alice and Bob then perform error correction over the classical channel to get Alice's data to agree with the sifted key established by Bob. 
\item Alice and Bob perform privacy amplification using a two-universal hash function to obtain the final shared secret key. 
\end{enumerate}
In practice, the discarded signals are simply removed before performing privacy amplification. We include them with the discard flag $\perp$ only to formulate the protocol in a trace-preserving manner.

This description of an EB protocol is general. When modelling the EB version of a P\&M protocol, some steps can be made more specific. Suppose Alice prepares signal states $\ket{\psi_i}$ with probability $p(i)$. In the source-replacement scheme, this is modelled as Alice preparing the state $\tau_{AA'}=\sum_{ij} \sqrt{p(i)p(j)}\dyad{i}{j}_A\otimes\dyad{\psi_i}{\psi_j}_{A'}$, sending system $A'$ to Bob, and measuring with the POVM $\{\dyad{i}_A\}$. As Eve cannot access Alice's lab, the reduced density matrix is known. Thus, for P\&M protocols there is an additional constraint in parameter estimation, namely $\rho_A=\reduced=\sum_{ij} \sqrt{p(i) p(j)} \braket{\psi_j}{\psi_i}\dyad{i}{j}$. Note that $\tau_A$ is closely related to the Gram matrix of signal states.

\subsection{Post-processing Channel}\label{postprocesschannel}
In order to evaluate the key rate, we formally define the post-processing steps as a quantum channel. We introduce registers $\tilde{A}$ and $\tilde{B}$, which hold the public announcements, $\overline{A}$ and $\overline{B}$, which hold the private measurement data, and $Z$, which holds the result of the key map. Without loss of generality, Eve has access to the register $E$ purifying $\rho_{AB}$ and the public information. 

Alice and Bob's measurement can be described by a channel $\meas$ that is simply given by $\Tr_{AB}(\nondest)$, where
\begin{equation}
\begin{split}
\label{nondestdefn}
\nondest(\rho_{ABE})=\sum_{i, j}  \dyad{a_i}_{\tilde{A}} \otimes \dyad{\alpha_i}_{\overline{A}}  \otimes  \dyad{b_j}_{\tilde{B}}\\ \otimes \dyad{\beta_j}_{\overline{B}} \otimes  \left[\left(\sqrt{P_A^i} \otimes \sqrt{P_B^j}\right) \rho_{ABE} \left(\sqrt{P_A^i} \otimes \sqrt{P_B^j}\right)\right].
\end{split}
\end{equation}
The action of the key map can be represented by the isometry
\begin{equation}
\begin{split}
V=\sum_{\sapub, \sbpub, \sbpriv} \ket{g(a,b,\beta)}_Z \otimes \dyad{a}_{\tilde{A}} \otimes \dyad{b}_{\tilde{B}} \otimes  \dyad{\beta}_{\overline{B}}.
\end{split}
\end{equation}
The final state between all parties is then $V\meas\left(\rho_{ABE}\right)V^\dagger$. We can pull the partial trace in the measurement channel $\meas$ through the isometry $V$. Then, the final classical-quantum state between the register holding the result of the key map and Eve is
\begin{align}
\sigma_{Z[E]}&=\Tr_{A\overline{A}B\overline{B}}( V \nondest(\rho_{ABE}) V^\dagger),\\
\label{postmap}
&\equiv\Phi(\rho_{ABE}).
\end{align}
The channel $\Phi$ characterizes the post-processing steps and $[E]$ denotes the composite register $E\tilde{A}\tilde{B}$. Note that discarded signals will not contribute to the key rate, as $\sigma_{Z[E]}$ is block-diagonal in the classical announcements. To compute $\Phi$, it is thus simpler and equivalent to not apply the POVM elements leading to discarded outcomes, rather than use a discard symbol \cite{Lin2019}. In this case, the postprocessing map is completely positive and trace non-increasing.
\subsection{Asymptotic Key Rate Formula}
\label{qkdapp}
The secure key rate of a QKD protocol, in the asymptotic limit and assuming independent and identically distributed (IID) rounds (collective attack), is given by the Devetak-Winter formula \cite{Devetak2005}. One can then aim to lift this to the key rate under coherent attacks, for example, using a quantum de Finetti theorem \cite{Renner2009}, and ultimately to a fully-composable security proof incorporating finite-size effects. Under collective attacks, for a given $\rho_{AB}$, the asymptotic key rate is
\begin{equation}
R^\infty=I(Z:X)-\chi(Z|[E]),
\end{equation}
where $I$ denotes the mutual information and $\chi$ denotes the Holevo information. The quantities are evaluated on the post-processed state after a single round, and $X$ refers to the party who does not perform the key map. In general, the state $\rho_{AB}$ is unknown. However, it is constrained by Alice and Bob through testing. The Devetak-Winter formula should be evaluated on the worst-case state compatible with these constraints. 

Determining the key rate of a protocol can then be reformulated as a convex optimization problem \cite{Coles2016}. The Devetak-Winter formula can be rearranged as, 
\begin{align}
R^\infty&=H(Z|[E])-H(Z)+I(Z:X),\\
\label{keyrateformula}
&= H(Z|[E]) -H(Z|X).
\end{align}
Only the first term needs to be optimized over. The second term is replaced by the actual error-correction cost $\delta^{leak}_{EC}$ which is the number of bits leaked per round. For realistic error correction, $\delta^{leak}_{EC}$ will be larger than the Shannon limit $H(Z|X)$. The key rate is thus
\begin{equation}
\label{convexkeyrate}
R^\infty=\min_{\rho_{AB} \in \mathbf{S}^{QKD}} [\fqkd(\rho_{AB}) ]-\delta^{leak}_{EC},
\end{equation}
where the objective function is the conditional entropy evaluated on the purified state after post-processing 
\begin{equation}
\label{objfunc}
\fqkd(\rho_{AB})=H(Z|[E])_{\Phi(\rho_{ABE})},
\end{equation}
and the set $\mathbf{S}^{QKD}$ is defined by the parameter estimation Alice and Bob perform, as well as the reduced density matrix constraint for P\&M protocols \cite{Winick2018}. That is,
\begin{equation}
\begin{split}
\label{gammaset}
\mathbf{S}^{QKD}=\{\rho\in \pos{\mathcal{H}_{AB}}:&\Tr(\rho)=1,\\
&\Tr_B(\rho)=\reduced, \\
&\Tr(\rho\Gamma_i)=\gamma_i \}.
\end{split}
\end{equation}
As both the objective function and feasible set are convex \cite{Winick2018}, Eq. \eqref{convexkeyrate} is a convex optimization.

When the Hilbert space $\mathcal{H}_{AB}$ is finite-dimensional, this problem can be reliably solved numerically \cite{Coles2016,Winick2018}. However, when the Hilbert space is infinite-dimensional, it is clearly not possible to directly solve this optimization numerically. We will develop a general method that, for all QKD protocols, allows us to compute tight lower bounds on the asymptotic key rate by relating the infinite-dimensional optimization to a finite-dimensional one. We can then numerically solve the finite-dimensional problem using the methods of \cite{Winick2018} to get tight lower bounds on the key rate for protocols where the state lives in an infinite-dimensional Hilbert space.

\section{General Framework \label{framework}}
In this section, we develop a general method to lower bound infinite-dimensional optimizations of the sort seen in the previous section. In fact, we present our results in more generality than used for finding QKD key rates. Our general method may be applied to other quantum information scenarios, such as entanglement verification \cite{Killoran2011}. 

Our notational conventions are as follows. For a Hilbert space $\HH$, let $\normdensityH$ ($\densityH$) be the set of (sub)normalized density operators on $\HH$. That is, $\densityH=\pos{\mathcal{H}} \cap \mathcal{T}_1$, where $\pos{\mathcal{H}}$ is the set of positive semidefinite operators on $\HH$, and $\mathcal{T}_1$ is the set of trace-class operators with trace no greater than 1. We use tildes to denote operators that are subnormalized.

\subsection{Problem Setup and Definitions}
We begin with some definitions. Let $\Hinf$ be a separable Hilbert space, which may be infinite-dimensional. Let $\sinf$ be a convex subset of $\densityH[\Hinf]$. Finally, let $f$ be a convex function from $\densityH[\Hinf]$ to $\mathbbm{R}$. The infinite-dimensional optimization we will consider is
\begin{equation}
\label{inf}
\min_{\rhot \in\sinf} f(\rhot).
\end{equation}
We assume that there exists a feasible operator $\rhoo$ achieving the optimum. This assumption holds in the practical case where the objective function $f$ is continuous and $\sinf$ is compact. Our goal then is to find a lower bound on $f(\rhoo)$. Our strategy to do this is to relate the infinite-dimensional optimization to a suitably chosen finite-dimensional optimization, which can then be solved numerically. Note that in the process we choose certain mathematical objects freely; how to choose them effectively is the focus of the following sections. 

Choose $\Hfin$ to be any finite-dimensional subspace of $\Hinf$. $\Pi$ is defined as the projector onto this subspace, and $\bar{\Pi}\equiv\mathbbm{1}-\Pi$. Choose $\sfin$ to be a nonempty convex subset of $\densityH[\Hfin]$ satisfying
\begin{equation}
\label{containment}
\Pi \sinf \Pi \subseteq \sfin.
\end{equation}
It is always possible to choose such a set $\sfin$, as $ \Pi \densityH[\Hinf] \Pi \subseteq \densityH[\Hfin]$. We can now define the finite-dimensional optimization
\begin{equation}
\label{fin}
\min_{\rhot \in\sfin} f(\rhot).
\end{equation}
We assume there exists a feasible operator $\nrho$ achieving this optimum. Again, this assumption holds in the practical case when $f$ is continuous and $\sfin$ is compact.

\subsection{Main Theorem}
At a high level, our proof method is illustrated in Fig. \ref{thmfig}. In order to relate the objective function $f$ evaluated at the two optimal operators, $\rhoo$ and $\nrho$, we introduce an auxiliary variable, $\rhon\equiv\Pi \rhoo \Pi$. This variable is the projection of the infinite-dimensional optimum $\rhoo$ onto the chosen finite subspace. We relate the optima to this auxiliary variable separately, and then to each other. For positive operators $P$ and $Q$, our convention for the fidelity function is $F(P,Q)=\Tr(\sqrt{\sqrt{Q} P \sqrt{Q} })$.
\begin{figure}
\includegraphics[scale=0.6]{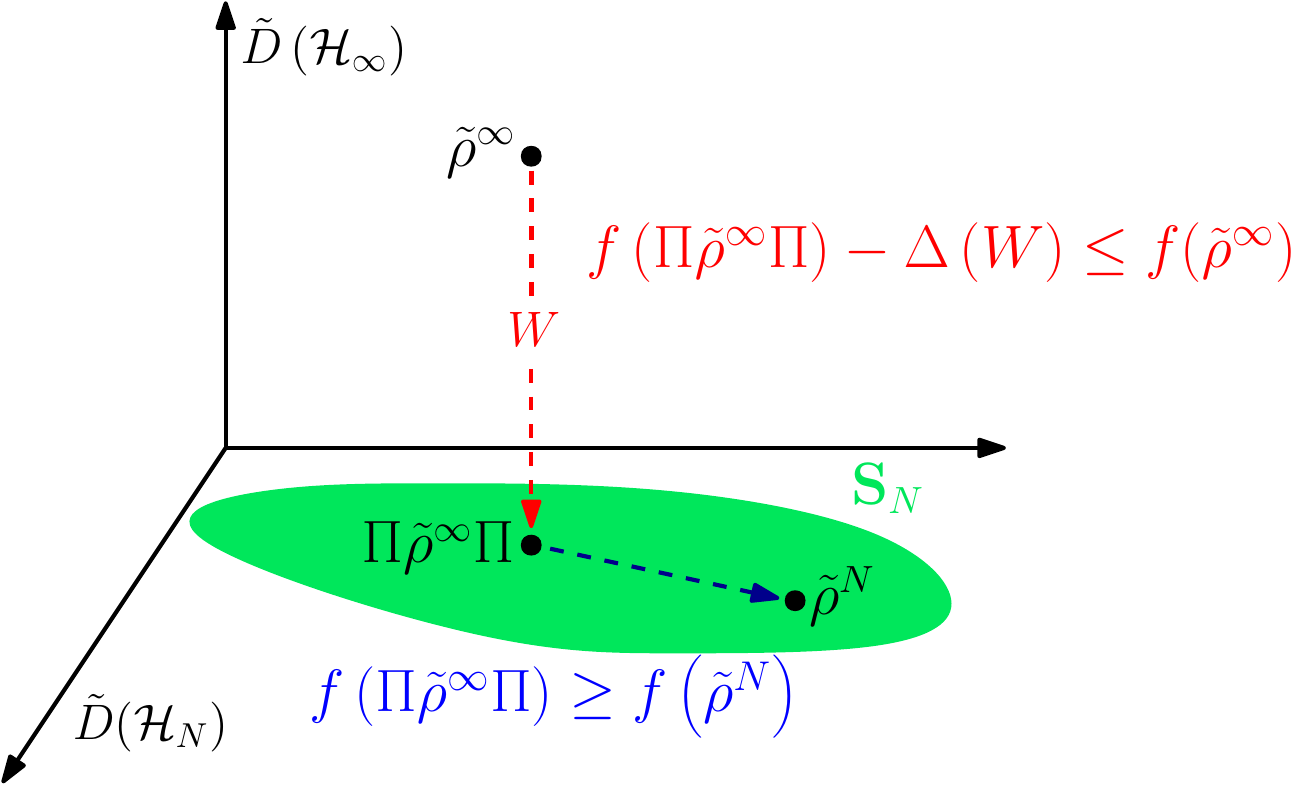}
\caption{Pictorial representation of Theorem \eqref{theorem_general}. The set $\sfin$ is chosen to contain the projection $\Pi\rhoo\Pi$. This auxiliary variable is used to relate $f(\rhoo)$ and $f(\nrho)$ }
\label{thmfig}
\end{figure}

\textbf{Inequality 1, Finite Set:} We first relate $f(\nrho)$ and $f(\rhon)$. By definition, $\rhon \in \Pi \sinf \Pi$. By the containment property introduced in Eq. \eqref{containment}, it follows that $\rhon \in \sfin$. Thus, $\rhon$ is feasible for the minimization in Eq. \eqref{fin}. Since $\nrho$ achieves the minimum, it follows that
\begin{equation}
	\label{finiteset}
	f(\nrho)\leq f(\rhon).
\end{equation}

\textbf{Inequality 2, Projection:} We next relate $f(\rhoo)$ and $f(\rhon)$. To do this, we introduce a property of $f$ relating the fidelity of input states to changes in the function. Define $f$ as being \emph{uniformly close to decreasing under projection (UCDUP) on $\mathbf{S}\subseteq \densityH[\Hinf]$ with correction term $\Delta$} if, for every $\sigmat\in\mathbf{S}$, it satisfies
\begin{equation}
\begin{split}
\label{ucdupdefn}
F(\sigmat,\Pi\sigmat\Pi) \geq \Tr(\sigmat)-W \implies f(\Pi\sigmat\Pi) - f(\sigmat) \leq \Delta(W),
\end{split}
\end{equation}
where $\Delta$ is a nonnegative, increasing function satisfying $\Delta(0)=0$. This property has some similarities to uniform-continuity with respect to trace distance. In fact, it is implied by the latter. However, because being UCDUP is a weaker condition, we may be able to find smaller $\Delta$ than implied by a uniform-continuity bound (see Sec. \ref{ucdup}). In some cases, it may also be possible to set $\Delta=0$. This is the case for bounding the QKD objective function when the key map POVM elements commute with $\Pi$ (see Sec. \ref{dup}). It is also the case for some entanglement measures \cite{Killoran2011}.

We can re-express the condition on $W$ in a more useful way. By Winter's Gentle Measurement Lemma \cite{Winter1999}, which holds with equality when considering projections, the fidelity can be written as
\begin{equation}
F(\sigmat,\Pi\sigmat\Pi)=\Tr(\sigmat\Pi).
\end{equation}
Rearranging, it follows that
\begin{equation}
F(\sigmat,\Pi\sigmat\Pi) \geq \Tr(\sigmat)-W \iff \Tr(\sigmat \bar{\Pi})\leq W.
\end{equation}
Thus, when $f$ is UCDUP on $\sinf$ and $\Tr(\rhoo \bar{\Pi})\leq W$,
\begin{equation}
\label{projection_inequc}
f(\rhon)-\Delta(W)\leq f(\rhoo).
\end{equation}

We now have all the pieces in place to prove our main theorem.
\begin{thm}[Relating Finite- and Infinite-Dimensional Optimizations]
\label{theorem_general}
If $f$ is UCDUP on $\sinf$ with correction term $\Delta$ (Eq. \eqref{ucdupdefn}) and $\Tr(\rhoo\bar{\Pi})\leq W$, then 
\begin{equation}
f(\nrho) -\Delta(W)\leq f(\rhoo).
\end{equation}
\end{thm}
\begin{proof}
This theorem follows from chaining the inequality in Eq. \eqref{finiteset}, which relates $f(\nrho)$ and $f(\rhon)$, with the one in Eq. \eqref{projection_inequc}, which relates $f(\rhoo)$ and $f(\rhon)$.
\end{proof}

In order to apply this theorem, in addition to choosing $\Hfin$ and $\sfin$, we need to determine an expression for the correction term $\Delta$ and a suitable value for $W$ (see Sec. {\ref{frameworkapp}}). For the latter quantity, as $\rhoo$ is unknown, we choose the smallest value of $W$ satisfying
\begin{equation}
\label{Wsdp}
W\geq\max_{\rhot\in\sinf} \Tr(\rhot \bar{\Pi}).
\end{equation}
We thus see that the determination of $W$ itself involves an infinite-dimensional optimization. In practice, this optimization tends to be considerably easier to solve than the original one (Eq. {\eqref{inf}}). In particular, when $\sinf$ is the feasible set of a semidefinite program (SDP), known as a spectrahedron, then any solution to the dual problem provides an upper bound. For the QKD protocols we study, we can obtain tight values of $W$ by analytically solving the dual or a relaxed version of the primal problem (see Sec. {\ref{qkdweight}}). 

Note that $W$ will not only be used to determine the correction term $\Delta(W)$, but also to parametrize the set $\sfin$ we choose (see Sec. \ref{finitesetchoice}).

\section{Application of Framework to QKD}
\label{frameworkapp}
Having developed a general method to lower-bound convex optimizations, we will now apply it to asymptotic QKD key rate optimizations. Our initial infinite-dimensional optimization is as described in Sec. \ref{qkdapp}, so that $\Hinf=\mathcal{H}_{AB}$, $f=\fqkd$, and $\sinf=\mathbf{S}^{QKD}$. Unless otherwise stated, we will focus on the case where $\mathcal{H}_A$ is finite-dimensional. In order to apply Theorem \eqref{theorem_general}, we need four inputs: the subspace $\Hfin$ onto which $\Pi$ projects, the bound on weight outside the subspace $W$, the correction term $\Delta$, and the finite set $\sfin$. The first two of these depend on the specific QKD protocol, so we only give some general remarks on them in this section. The second two are generic for all protocols, so we derive them in detail.

\subsection{Choose Subspace \texorpdfstring{$\mathcal{H}_N$} {}}
This step is protocol-specific, and the choice of finite subspace $\Hfin$ can be made freely. However, for a fixed finite dimension, some subspaces will give a better lower bound than others. Intuitively, we want to choose a subspace which contains most of the weight of $\rhoo$. Since this state is unknown, a good heuristic is to choose the subspace containing the most weight of the expected state under a representative channel model (see Sec. \ref{dmcvqkdfiniteset}). This is conceptually similar to assuming a particular channel for the purpose of designing a QKD protocol or error-correcting code. Another consideration is to choose the projection to commute with the objective function POVM elements (see Sec. \ref{correctiontermderivation}) or the constraint operators (see Sec. \ref{finitesetchoice}). The choice of the dimension of the finite subspace will be influenced by the run time of the numerical method used to solve the finite-dimensional optimization.

\subsection{Bound Weight \texorpdfstring{$W$}{} Outside Subspace}\label{qkdweight} 
This step is also protocol-specific. Note that for $\sinf$ as given in Eq. \eqref{gammaset}, the bound on $W$ shown in Eq. \eqref{Wsdp} is a semidefinite program. The range of approaches to calculate $W$ is fairly wide, so in place of general strategies, we instead survey some examples. For DMCVQKD, we analytically find a solution to the dual SDP (see Sec. \ref{dmcvqkdweight}). For unbalanced phase-encoded BB84, the monotonicity of cross-clicks or double-clicks with increasing photon number, along with Markov's inequality, is used to bound $W$ \cite{Li2020}.

\subsection{Determine Correction Term \texorpdfstring{$\Delta$}{} }\label{correctiontermderivation}
Recall our objective function $f$ is given in Eq. \eqref{objfunc}. We will show that it is UCDUP and determine the correction term $\Delta$ as a function of $W$. We will find a general correction term which does not depend on the details of the postprocessing map (see Eq. \eqref{postmap}) and is thus applicable to all protocols. We will also show that when the postprocessing map satisfies a certain property, which holds for some protocols, we can omit the correction term entirely, i.e. $\Delta=0$. 
\subsubsection{General Case}\label{ucdup}
In this case, we do not make any assumptions about the post-processing channel $\Phi$ in Eq. \eqref{postmap}, other than the fact that it is indeed a channel. The crux of the proof is in the following lemma.
\begin{lemma}
\label{ucduplemma}
Let $\mathcal{H}_A$ and $\mathcal{H}_B$ be two Hilbert spaces, where the dimension of $\mathcal{H}_A$ is $\abs{A}$ while $\mathcal{H}_B$ can be infinite-dimensional. Let $\rhot_{AB}, \sigmat_{AB}\in\tilde{D}( \mathcal{H}_A \otimes  \mathcal{H}_B)$ be two subnormalized, classical-quantum states with $\Tr(\rhot_{AB})\geq\Tr(\sigmat_{AB})$. If $\frac{1}{2}\norm{\rhot_{AB}-\sigmat_{AB}}_1\leq\epsilon$, then
\begin{equation}
H(A|B)_{\sigmat_{AB}} - H(A|B)_{\rhot_{AB}} \leq  \e \log_2\abs{A} +(1+\e)h\left(\frac{\e}{1+\e}\right),
\end{equation}
where $h(x)$ is the binary entropy function.
\end{lemma}

\begin{proof}
This result is a generalization of Lemma 2 in \cite{Winter2016}, which was derived for normalized states. The proof of the extension to subnormalized states is given in Appendix \ref{ctybd}.
\end{proof}
With this lemma in hand, we will prove that $f$ is UCDUP on $\normdensityH[\Hinf]$ in the following theorem. Note this is sufficient as for any QKD protocol, $\sinf \subseteq \normdensityH[\Hinf]$.
\begin{thm}
\label{ucdupthm}
The QKD objective function $f$ is UCDUP on $\normdensityH[\Hinf]$ with correction term
\begin{equation}
\begin{split}
\label{qkducdup}
\Delta(W)= \sqrt{2W-W^2} \log_2 \abs{Z} \\+  \left(1+\sqrt{2W-W^2}\right) h\left(\frac{\sqrt{2W-W^2}}{1+\sqrt{2W-W^2}}\right),
\end{split}
\end{equation} 
where $\abs{Z}$ is the dimension of the key map register.
\end{thm}

\begin{proof}
As per the definition of UCDUP, let $\sigma_{AB}\in\normdensityH[\Hinf]$ be a state satisfying $F(\sigma_{AB},\Pi\sigma_{AB}\Pi) \geq 1-W$. Let $\sigma_{ABE}$ and $\sigmat^{\Pi}_{ABE}$ be purifications of $\sigma_{AB}$ and $\Pi\sigma_{AB}\Pi$ respectively (we consider a purification of an unnormalized state to not change the trace). By Uhlmann's theorem \cite{Uhlmann1976}, we can choose the purifications to satisfy
\begin{align}
F(\sigma_{ABE},\sigmat^{\Pi}_{ABE})&=F(\sigma_{AB},\Pi\sigma_{AB}\Pi)\\
&\geq 1 -W.
\end{align}
By the monotonicity of fidelity under channels, we have
\begin{equation}
F(\tau_{Z[E]},\taut^\Pi_{Z[E]})\geq 1-W
\end{equation}
where $\tau_{Z[E]}=\Phi(\sigma_{ABE})$ and $ \taut^\Pi_{Z[E]}=\Phi(\sigmat^{\Pi}_{ABE})$. The Fuchs-van de Graaf inequalities \cite{Fuchs99}, which are valid in infinite dimensions and for subnormalized states, relate fidelity and trace distance as follows, $1-F(\rhot,\sigmat)\leq\frac{1}{2}\norm{\rhot-\sigmat}_1\leq\sqrt{1-F(\rhot,\sigmat)^2}$. By the second of these,
\begin{equation}
\frac{1}{2}\norm{\tau_{Z[E]}-\taut^\Pi_{Z[E]}}_1 \leq \sqrt{2W-W^2}.
\end{equation}
Since $\Phi$ is trace-preserving, $\Tr(\tau_{Z[E]})\geq \Tr(\taut^\Pi_{Z[E]})$. Thus, we can apply Lemma \ref{ucduplemma} to obtain
\begin{align}
f(\Pi\sigma_{AB} \Pi) - f(\sigma_{AB}) &= H(Z|[E])_{\taut^\Pi_{Z[E]}} - H(Z|[E])_{\tau_{Z[E]}} \\
\begin{split}
&\leq  \sqrt{2W-W^2} \log_2 \abs{Z} \\
+  (1+\sqrt{2W-W^2}) &h\left(\frac{\sqrt{2W-W^2}}{1+\sqrt{2W-W^2}}\right).
\end{split}
\end{align}
This is precisely the condition we require for $f$ to be UCDUP, so we identify the right-hand side as $\Delta(W)$.
\end{proof}
For a key map with $k$ outcomes, the dimension of the key map register is $k$. For the purpose of this argument, the discard symbol $\perp$ does not count towards a key outcome. The reason for this is that instead of $\perp$, one could use any pre-existing key symbol to flag discarded signals. Since the classical-quantum state $\sigma_{Z[E]}$ between the key map register and Eve is block-diagonal in the classical announcements (Eq. \eqref{postmap}), Eve could identify the discarded signals from those public announcements alone. This would leave the value of the objective function unchanged. Indeed, the $\perp$ symbol is only used for clarity in our presentation.
\subsubsection{Special Case: Block-Diagonal Measurements}\label{dup}
If the key map POVM elements are block-diagonal with respect to $\Pi$ and $\bar{\Pi}$, the correction term is zero.
\begin{thm}
\label{nodelta}
Let $\Phi$ be defined by the POVMs $\{P_A^i\}$ and $\{P_B^j\}$ and a key map isometry $V$. If all POVM elements are block-diagonal, so that $[P_A^i\otimes P_B^j,\Pi]=0 \ \forall i,j$, then
\begin{equation}
\Delta(W)=0.
\end{equation}
\end{thm}

\begin{proof}
Recall that $\Pi$ only acts on the $AB$ subsystem. Then, $\Pi\rho_{ABE}\Pi$ is a purification of $\Pi\rho_{AB}\Pi$. $\Pi$ commutes with $V$ as they act on different subsystems, and commutes with all elements of the POVMs $\{P_A^i\}$ and $\{P_B^j\}$ by assumption. By the definition of $\nondest$ (Eq. \eqref{nondestdefn}), it then follows that
\begin{equation}
\Pi \ V \ \nondest\left(\rho_{ABE}\right) \ V^\dagger \ \Pi =  V \ \nondest\left(\Pi\rho_{ABE} \Pi \right) \ V^\dagger
\end{equation}
and analogously for $\bar{\Pi}$.

Defining the channel $\Xi_{AB}(\rho)=\Pi\rho\Pi+\bar{\Pi}\rho\bar{\Pi}$, we have 
\begin{align}
f(\rho_{AB})&=H(Z|[E])_{ \Phi( \rho_{ABE} ) }\\ 
&=H(Z|[E])_{ \Tr_{A\overline{A}B\overline{B}}\left( V \nondest\left(\rho_{ABE}\right) V^\dagger \right) }\\ 
&=H(Z|[E])_{ \Tr_{A\overline{A}B\overline{B}}\left( \Xi_{AB}\left(V \nondest\left(\rho_{ABE}\right) V^\dagger \right)\right)} \label{channeltrace}\\ 
&=H(Z|[E])_{ \Tr_{A\overline{A}B\overline{B}}\left( V \nondest\left(\Xi_{AB}\left(\rho_{ABE}\right)\right) V^\dagger \right)}\label{channelcommute}\\ 
&\geq H(Z|[E])_{ \Phi( \Pi\rho_{ABE}\Pi)} + H(Z|[E])_{ \Phi( \bar{\Pi}\rho_{ABE}\bar{\Pi})}\label{concavity}\\ 
&\geq H(Z|[E])_{ \Phi( \Pi\rho_{ABE}\Pi)} \label{condentrpos} \\ 
&= f(\Pi\rho_{AB}\Pi),
\end{align}
where we have in Eq. \eqref{channeltrace} freely introduced a channel as it will be traced out,  in Eq. \eqref{concavity} used the concavity of conditional entropy, and in Eq. \eqref{condentrpos} used the nonnegativity of conditional entropy for classical-quantum states (see Eq. \eqref{postmap}). 
\end{proof}

\subsection{Choose Finite Set \texorpdfstring{$\sfin$}{}}\label{finitesetchoice}
Recall that the feasible set $\sinf$ for the infinite-dimensional optimization is given in Eq. \eqref{gammaset}. The form of this set is common to all protocols. We consider one particular way of choosing $\sfin$ which is to individually expand the constraints of $\sinf$, using $W$ as a parameter. Note that in order to have the highest key rates, we want $\sfin$ to be as small as possible, while still satisfying the containment in Eq. \eqref{containment}. Assuming $\sigma\in\sinf$, we will now list constraints that $\sigman\equiv\Pi\sigma\Pi$ must satisfy. 

We begin with the trace constraint, $\Tr(\sigma)=1$, which can be easily expanded. By the definition of $W$ in Eq. \eqref{Wsdp}, $1-W\leq\Tr(\sigman)\leq1$. 

We next consider the expectation constraints, $\Tr(\sigma \Gamma_i)=\gamma_i$. We choose to define the loosened constraints as $\gamma_i^{min}\leq\Tr(\sigman\Gamma_i)\leq\gamma_i^{max}$, where
\begin{align}
\label{gammamax}
\gamma_i^{max}&=\max_{\rho\in\sinf} \Tr(\Gamma_i \ \Pi \rho \Pi),\\
\label{gammamin}
\gamma_i^{min}&=\min_{\rho\in\sinf} \Tr(\Gamma_i \ \Pi \rho \Pi).
\end{align}
For concreteness, we now specialize our discussion to the case where $[\Pi,\Gamma_i]=0$ and $\Gamma_i\geq0 \ \forall i$. This condition is satisfied for all the protocols we study in this article. We emphasize that this is not a particularly strong assumption. With a judicious choice of $\Pi$, this condition can be achieved for many protocols. In the alternate, determining the bounds in Eqs. \eqref{gammamax} and \eqref{gammamin} without this assumption is also tractable. For example, tight bounds are derived even for non-commuting, non-positive, and unbounded observables $\Gamma_i$ in \cite{Killoran2011}.

In the following theorem, we derive the desired bounds on expectations in the finite subspace. 
\begin{thm}
	Let $\Gamma_i\geq0$ and $[\Pi,\Gamma_i]=0$. If $\Tr(\rho\bar{\Pi})\leq W$ and $\Tr(\rho\Gamma_i)=\gamma_i$, then $ \gamma_i-W\norm{\Gamma_i}_\infty\leq \Tr(\Pi\rho\Pi \ \Gamma_i)\leq\gamma_i$.
\end{thm}

\begin{proof} By the commutation relation, it follows that 
	\begin{align}
	\Tr(\Pi\rho\Pi\Gamma_i)&=\Tr(\sqrt{\Gamma_i}\Pi\rho\Pi\sqrt{\Gamma_i})\\
	&=\Tr(\Pi\sqrt{\Gamma_i}\rho\sqrt{\Gamma_i}\Pi).\label{commute}
	\end{align}
We now find the upper and lower bounds separately. 

For the upper bound, we simply note that the trace of a positive operator can only decrease under projection. Then,
	\begin{align}
	\Tr(\Pi\sqrt{\Gamma_i}\rho\sqrt{\Gamma_i}\Pi)&\leq\Tr(\sqrt{\Gamma_i}\rho\sqrt{\Gamma_i})\\
	&=\gamma_i.
	\end{align}

For the lower bound, recall H\"{o}lder's inequality, which states that for any two operators $A$ and $B$, $\Tr(A^\dagger B)\leq\norm{A}_p\norm{B}_q$, where $\frac{1}{p}+\frac{1}{q}=1$ so that $\norm{\cdot}_p$ and $\norm{\cdot}_q$ are dual Schatten norms. Then,
	\begin{align}
	\begin{split}
	\Tr(\Pi\sqrt{\Gamma_i}\rho\sqrt{\Gamma_i}\Pi)&=\Tr(\sqrt{\Gamma_i}\rho\sqrt{\Gamma_i}) \\&\quad - \Tr(\bar{\Pi}\sqrt{\Gamma_i}\rho\sqrt{\Gamma_i}\bar{\Pi})
	\end{split}\\
	&=\gamma_i-\Tr(\bar{\Pi}\rho\bar{\Pi} \ \Gamma_i)\\
	&\geq\gamma_i-\norm{\bar{\Pi}\rho\bar{\Pi}}_1 \norm{\Gamma_i}_\infty \label{holder}\\
	&\geq \gamma_i-W\norm{\Gamma_i}_\infty,
	\end{align}
where in Eq. \eqref{holder} we used H\"{o}lder's inequality. Note that this lower bound is trivial for unbounded observables.
\end{proof}

Finally we consider the reduced state constraint $\Tr_B(\sigma)=\reduced$. We could, using a complete Hermitian basis on system $A$, write the reduced state constraint as a set of expectations and expand it in the same manner as above. However, there are two better ways to perform the expansion. One approach works in general, while the other is tighter but only works in a specific case. In general, by the Fuchs-van de Graaf inequalities, $\frac{1}{2}\norm{\sigma-\sigman}_1\leq\sqrt{2W-W^2}$. By the monotonicity of trace distance under channels, and by the fact that taking the partial trace is a channel, this implies the constraint $\frac{1}{2}\norm{\reduced-\Tr_B(\sigman)}_1\leq\sqrt{2W-W^2}$. Alternatively, in the case where the projection has the form $\Pi=\Pi_A \otimes \Pi_B$, a simple positivity argument implies the constraint $\Tr_B(\sigman)\leq \Pi_A \reduced \Pi_A$. This constraint can be used even when $\mathcal{H}_A$ is infinite-dimensional.

Having expanded all the constraints in the case where the observables are positive operators and commute with the projection, we summarize the definitions of the original infinite-dimensional optimization:
\begin{IEEEeqnarray*}{uC'L}
$\displaystyle{\minimize_{\rho}}$& &f(\rho) \\
subject to:& 	&\Tr(\rho)=1\\
\IEEEyesnumber&	&\Tr_B(\rho)=\reduced \\
&	&\Tr(\rho\Gamma_i)=\gamma_i \\
& &\rho\in \pos{{\mathcal{H}_\infty}},
\end{IEEEeqnarray*}
and the expanded finite-dimensional optimization:
\begin{IEEEeqnarray*}{uC'L}
 \label{sfinoptim}
$\displaystyle{\minimize_{\rhot}}$& &f(\rhot) \\
subject to:& 	&1-W\leq\Tr(\rhot)\leq1\\
\IEEEyesnumber &	&\frac{1}{2}\norm{\Tr_B(\rhot)-\reduced}_1\leq \sqrt{2W-W^2} \\
&	&\gamma_i-W\norm{\Gamma_i}_\infty\leq\Tr(\rhot\Gamma_i)\leq\gamma_i\\
& &\rhot\in \pos{\mathcal{H}_N}.
\end{IEEEeqnarray*}
When $\Pi=\Pi_A\otimes\Pi_B$, the constraint on the reduced density matrix is instead $\Tr_B(\rhon)\leq \Pi_A \reduced \Pi_A$. By construction, $\sfin$ is a convex subset of $\densityH[\Hfin]$ containing $\Pi \sinf \Pi$. The tightness of the feasible set is guaranteed by the fact that for $W=0$, $\sfin=\sinf$. For nonzero $W$, our numerical results in Sec. \ref{fss} provide strong evidence that our definition of $\sfin$ is close to the optimal choice $\Pi \sinf \Pi$.

\section{Numerical Implementation}\label{numerical}
\subsection{Review of Reliable Numerical Method}
Having formulated the finite-dimensional optimization of interest, we now wish to solve it numerically. In order to obtain a reliable lower bound in spite of floating-point imprecision and the imperfection of convex solvers, we use the numerical framework developed in \cite{Winick2018}. We briefly summarize the two steps of the method here. 

The first step is to approximately solve the optimization as written, using a method such as the Frank-Wolfe algorithm. This returns an approximate minimum, $\rho_{opt}$. The gradient of the objective function at this point $\nabla f(\rho_{opt})$ is computed. In the second step an expanded, linearized SDP is constructed from the first step optimization. To linearize, the objective function is replaced with $\Tr[\rho \nabla f(\rho_{opt})]$. Intuitively, this represents lower bounding the original convex function $f$ by a tangent hyperplane. To expand, the feasible set is enlarged by a small amount to account for numerical imprecision. Finally, the dual of this expanded, linearized SDP is derived. This dual SDP is solved numerically and any feasible point is a reliable lower bound on the original convex minimization. As discussed in Appendix \ref{epsilonchange}, we make use of one small improvement in how we expand the set in the second step, to account for numerical imprecisions, compared to \cite{Winick2018}. This change is important because it allows us to improve our numerical results at long distances (see Sec. \ref{dmcvqkdresults}). 

Note that the constraint operators $\Gamma_i$ and the POVM elements $P_k$ defining $f$ are, in general, infinite-dimensional. However, as $\Tr(\Pi \rho \Pi \ X)=\Tr[(\Pi \rho \Pi) ( \Pi X \Pi )]$, we can equivalently set $\Gamma_i \rightarrow \Pi \Gamma_i \Pi$ and $P_k\rightarrow \Pi P_k \Pi$. In the following discussion, we tacitly assume this substitution has been made in order to represent the optimizations numerically.

\subsection{SDP Formulation for Numerics}
To apply the numerical framework, we need the feasible set of the convex minimization to be that of an SDP. To show that $\sfin$ is such a set, we rewrite the trace distance constraint for the reduced density matrix. The trace norm can be expressed as an SDP \cite{Watrous2018}, which allows us to rewrite the constraint using slack variables. This equivalent reformulation of the finite-dimensional optimization is given by:
\begin{IEEEeqnarray*}{uC'L}
$\displaystyle{\minimize_{\rhot, R, S}}$& &f(\rhot) \\
subject to:& 	&1-W\leq\Tr(\rhot)\leq1\\
&	&\Tr(R)+\Tr(S)\leq 2\sqrt{2W-W^2} \\
&   &\Tr_B(\rhot)-R\leq\reduced\\
\IEEEyesnumber &   &-S-\Tr_B(\rhot)\leq-\reduced\\
&	&\gamma_i-W\norm{\Gamma_i}_\infty\leq\Tr(\rhot\Gamma_i)\leq\gamma_i\\
& &\rhot\in \pos{\mathcal{H}_N}\\
& & R,S\in \pos{\mathcal{H}_A}.
\end{IEEEeqnarray*}

Let $\xi$ denote the adjoint of the partial trace map $\Tr_B$ restricted to operators on $\Hfin$. Following the numerical framework, the corresponding dual linearized SDP is:
\begin{IEEEeqnarray*}{uC'L}
$\displaystyle{\maximize_{\vec{y},y_s,Y_1,Y_2}}$ & & -\vec{y}\cdot(\vec{\gamma}+\vec{\epsilon}_{rep})- y_s(2 \sqrt{2W-W^2}+\epsilon'_{rep}) \\
& & \quad-\Tr(\reduced Y_1)+\Tr(\reduced Y_2) \\
subject to:& 	&\sum_{i=1}^{2m} y_i\Gamma_i + \xi(Y_1) - \xi(Y_2) \geq - \nabla f(\rho_{opt})\\
&	&y_s \mathbbm{1}_A \geq Y_1\\
\IEEEyesnumber&	&y_s \mathbbm{1}_A \geq Y_2\\
&&\vec{y}\in \mathbbm{R}_{\geq 0}^{2m}\\
&&y_s\in\mathbbm{R}_{\geq 0}\\
&&Y_1, Y_2 \in \pos{\mathcal{H}_{A}}
\end{IEEEeqnarray*}
where $\vec{\gamma}=(\{\gamma_i\}_{i=1}^m, \{W\norm{\Gamma_i}- \gamma_i\}_{i=1}^m)$, $\vec{\Gamma}=(\{\Gamma_i\}_{i=1}^m, \{-\Gamma_i\}_{i=1}^m) $, and $\vec{\e}_{rep}$ and $\e_{rep}'$ are the expansion parameters which account for finite numerical precision (see Appendix \ref{epsilonchange} for further discussion).

As noted in Sec. \ref{finitesetchoice}, the finite-dimensional optimization has a slightly different form when the projection only acts on Bob's system, that is $\Pi=\mathbbm{1}_A \otimes \Pi_B$. In this case, the dual SDP is the same as above, but the term in the objective function and both constraints involving $y_s$ are removed, $Y_2$ is set to zero, and the term $-\e_{rep}' \Tr(Y_1 )$ is added to the objective function.

\section{Discrete-Modulated CVQKD \label{dmcvqkdW}}
We apply our approach for infinite-dimensional key rate calculations to DMCVQKD. We begin by reviewing the protocol and setting up the infinite-dimensional optimization. We then apply the steps of our method from Sec. \ref{frameworkapp} to relate the key rate to a finite-dimensional optimization. Finally, we solve this finite problem numerically to obtain key rates. Unless otherwise stated, all discussion is in the ideal detector scenario. We note that although continuity bounds have been considered in the context of DMCVQKD in \cite{Kaur2021}, this is very different from our work, as the bounds in \cite{Kaur2021} are used to quantify how well Gaussian modulation is approximated.
\subsection{Protocol Description}
We briefly review the DMCVQKD protocol. More details can be found in \cite{Lin2019}. In each round of the quantum phase, Alice prepares one of $d$ coherent signal states $\{\ket{\alpha_i}\}_{i=0}^{d-1}$ with probability $p(i)$. Our security proof works for any constellation of signal states, but we focus on the protocol with four symmetrically modulated states, also known as \emph{quadrature phase-shift keying}. In this case, Alice prepares the states $\{\ket{\alpha},\ket{i \alpha},\ket{-\alpha}, \ket{-i \alpha}\}$ uniformly at random, for a fixed signal state amplitude $\alpha$. Bob then performs either a homodyne or heterodyne measurement. We focus on the case where he does a heterodyne measurement. Bob's POVM is then $\{\frac{1}{\pi} \dyad{\zeta} \}_{\zeta\in\mathbbm{C}}$, while Alice's POVM is $\{\dyad{i} \}_{i=0}^{d-1}$ as we work in the source-replacement picture.

We focus on reverse reconciliation as it is known to outperform direct reconciliation in terms of transmission distance. Again, our security proof can be easily adapted to direct reconciliation. In reverse reconciliation, Bob maps his heterodyne measurement outcome $\zeta$ to a key symbol based on which region in the complex plane, or phase space, the outcome falls in. The key map is visualized in Fig. \ref{keymap}. Bob can also perform postselection (sifting) to improve the key rate of the protocol \cite{Silberhorn2002}. 

\begin{figure}
\centering
\includegraphics[width=\linewidth]{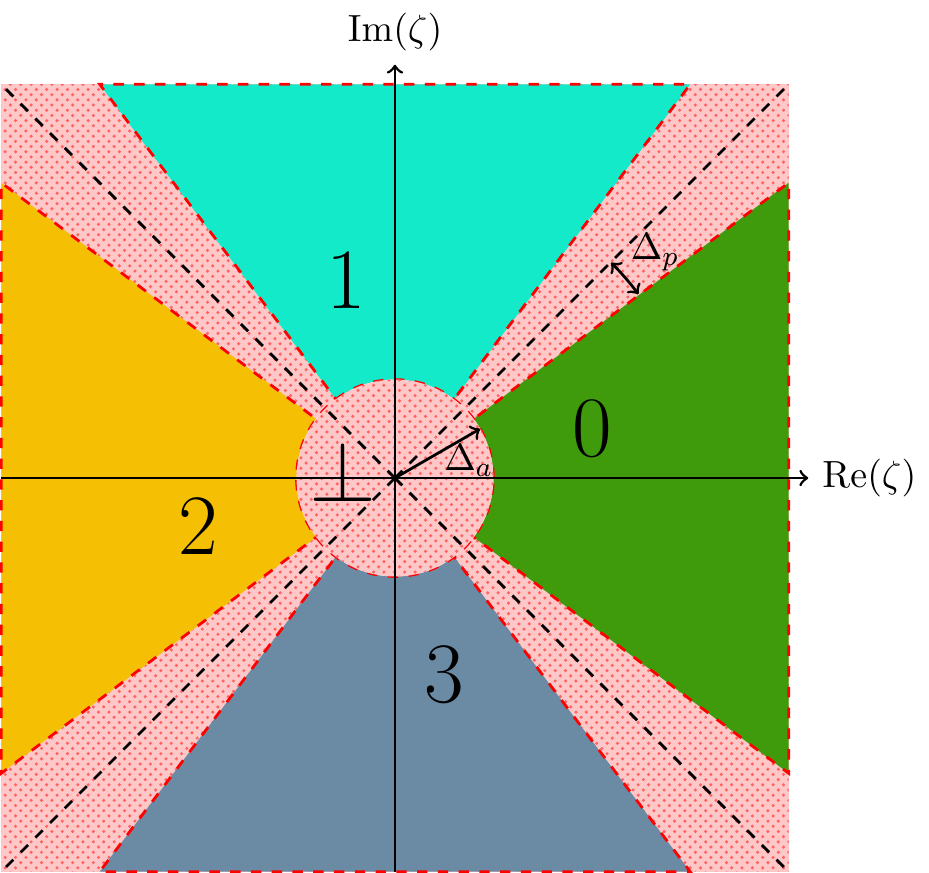}
\caption{Phase-space regions for the key map, with postselection, for reverse reconciliation in QPSK DMCVQKD. Bob obtains the measurement result $\zeta$ from the heterodyne detector. He maps the outcome to the symbol of the region containing $\zeta$. $\Delta_a$ and $\Delta_p$ are amplitude and phase postselection parameters, respectively. The $\perp$ region corresponds to signals that are discarded.}
\label{keymap}
\end{figure}

After a large number of rounds, Alice and Bob move on to the classical phase. For parameter estimation, Alice announces which signal state was sent. Since Bob's measurement determines the Husimi Q-function, he could in principle perform complete tomography to uniquely determine the received conditional states. This would not extend reasonably to a finite-key analysis. We thus choose to only use certain coarse-grained observables as constraints. (It is then unnecessary for Bob to announce his fine-grained measurement results, he only needs to announce these coarse-grained expectations.) In Appendix \ref{expectationapp}, we outline how to calculate the coarse-grained expectations from the fine-grained probability distribution. For reverse reconciliation with postselection, Bob also announces which signals are discarded. Finally, Alice and Bob perform error correction and privacy amplification on the sifted key.

\subsection{Infinite-Dimensional Optimization}
To formalize this description of the protocol, we define the objective function $f$ and the observables $\Gamma_i$.

\subsubsection{Objective Function}
Recall the definition of the postprocessing map in Eq. \eqref{postmap}. Outside of the testing rounds, the only announcements in this protocol are for sifting. Since Bob defines the key map and makes the sifting decision, this means that Alice's data is irrelevant for the postprocessing map. Thus, for the purpose of the postprocessing map, we can take Alice's POVM to trivially be $\{\mathbbm{1}_A\}$. For Bob's POVM, we do not need to work with the fine-grained set and can instead deal directly with the coarse-grained region operators. These region operators can be expressed as linear combinations of the fine-grained POVM elements. The region operators corresponding to non-discarded signals are
\begin{align}
R^j_B&=\frac{1}{\pi}\int_{\Delta_a}^{\infty} \int_{\frac{(2j-1)\pi}{4}+\Delta_p}^{\frac{(2j+1)\pi}{4}-\Delta_p} r \dyad{r e^{i\theta}} d\theta dr,
\end{align}
for $j\in\{0,1,2,3\}$.  Here $\Delta_a$ and $\Delta_p$ are postselection parameters for the amplitude and phase. Not performing postselection corresponds to setting $\Delta_a,\Delta_p=0$. 
The key map $g$ trivially maps Bob's private measurement result to the key register. Thus, we can omit the key map isometry altogether, and simply relabel the register $\overline{B}$ to $Z$. The simplified completely positive and trace non-increasing form of the postprocessing map $\Phi$ in Eq. \eqref{postmap} is then
\begin{equation}
\label{dmcvqkdmap}
\Phi(\rho_{ABE})= \sum_{z=0}^{3} \dyad{z}_Z \otimes \Tr_{AB}[\rho_{ABE} \ (\mathbbm{1}_A \otimes R^z_B )],
\end{equation}
where we have omitted trivial or redundant registers \cite{Lin2019}.

\subsubsection{Observables}
Alice and Bob's POVMs are as given above. As discussed, we need to choose a set of coarse-grained observables for Bob's side of parameter estimation. Typically, these observables are taken to be the quadratures and their higher moments \cite{Lin2019}. However, we introduce a different set of observables which will simplify our development, as they will commute with the projection we later define (Eq. {\eqref{dmcvqkdprojection}}). They will be parametrized by a list of complex numbers $\{\beta_i\}_{i=0}^{d-1}$. For any operator $X$, we use the shorthand notation $X_{\gamma}\equiv \disp{\gamma} X \dispd{\gamma}$ where $\disp{\gamma}$ is the displacement operator with complex parameter $\gamma$. Denoting the photon number operator by $\hat{n}\equiv \hat{a}^\dagger \hat{a}$, we choose as constraint observables for our protocol the set
\begin{equation}
\label{dmcvqkdgammas}
\{\Gamma_i\}=\{\dyad{i}\otimes \obsn, \dyad{i}\otimes \obsnsq\}_{i=0}^{d-1}.
\end{equation}
We consider second order constraints in $\hat{n}$ as this will be necessary to make the weight $W$ outside the subspace sufficiently small. We choose the displacements to be
\begin{equation}
\beta_i=\sqrt{\eta}\alpha_i,
\end{equation}
which are the amplitudes of the coherent signal states after passing through a Gaussian channel with loss $\eta$. As will be discussed in Sec. \ref{dmcvqkdresults}, our choice here is based on the expected channel behavior in an honest implementation of the protocol. We emphasize however that our security proof method works for any choice of $\{\beta_i\}$.

This choice of coarse-grained observables is of interest as it elucidates some of the essential working principles of the DMCVQKD protocol. The observables measure how spread out a state is compared to the coherent state $\ket{\beta_i}$. Intuitively, this characterizes the deviation from a generalized beamsplitting attack \cite{Heid2006}. These constraints also dovetail with a natural choice for the finite subspace.

Note that although Bob's observables have a dependence on the signal state, Bob physically performs the same heterodyne measurement each round. For parameter estimation, he simply holds all fine-grained data and coarse grains only after Alice announces which signal state was sent.

\subsubsection{Optimization Formulation}
We are now able to write down the infinite-dimensional optimization for DMCVQKD. Let $\mathcal{H}_A$ be the Hilbert space with dimension equal to the number of signal states $d$. Let $\mathcal{H}_B$ be the Hilbert space of a single optical mode, spanned by Fock states $\{\ket{n}\}_{n=0}^\infty$. The minimization is then:
\begin{IEEEeqnarray*}{uC'L}
$\displaystyle{\minimize_{\rho}}$& &f(\rho) \\
subject to:& 	&\Tr(\rho)=1\\
&&	\Tr_B(\rho)=\sum_{i,j} \sqrt{p(i) p(j)} \braket{\alpha_j}{\alpha_i}\dyad{i}{j}  \\
\IEEEyesnumber&	&\Tr[ \rho \left(\tfrac{1}{p(i)}\dyad{i}{i} \otimes \obsn \right) ]=\expn \label{dmcvqkdinf}\\
 &   &\Tr[\rho \left(\tfrac{1}{p(i)}\dyad{i}{i} \otimes \obsnsq \right) ]=\expnsq \\
& &\rho\in \pos{\mathcal{H}_A \otimes \mathcal{H}_B}.
\end{IEEEeqnarray*}

\subsection{Finite-Dimensional Optimization and Correction Term}
We apply the steps of our method to convert the preceding infinite-dimensional optimization into a tractable finite-dimensional one, and determine the associated correction term.

\subsubsection{Choose Subspace \texorpdfstring{$\mathcal{H}_N$}{}} \label{dmcvqkdfiniteset}
Recall our general principle is to choose the subspace containing the most weight of the state under a typical channel model. As discussed in Sec. \ref{dmcvqkdresults}, the channel model for a fiber-based implementation of this protocol is a lossy and noisy Gaussian channel. As noted above, for a pure-loss channel the expected coherent state is $\ket{0_{\beta_i}}$, where $\ket{n_{\gamma}}\equiv\disp{\gamma}\ket{n}$. For the specific noise model we consider, the expected state will be a displaced thermal state. This means some weight leaks into displaced Fock states with $n>0$. For the $i^{th}$ signal state, the best projection on $\mathcal{H}_B$ is thus
\begin{equation}
\PiN=\sum_{n=0}^N  \dyad{n_{\beta_i}}.
\end{equation}
We will refer to $N$ as the subspace dimension parameter. The projection operator on the total Hilbert space is
\begin{equation}
\label{dmcvqkdprojection}
\Pi^N\equiv\sum_{i=0}^{d-1} \dyad{i}_A \otimes \PiN.
\end{equation}
Note that this projection commutes with the observables (Eq. \eqref{dmcvqkdgammas}) because $\PiN$ commutes with $\obsn$. We use this fact to define the finite set(see Sec. {\ref{dmcvqkdfinitesetdefinition}}).

Unlike the truncated Fock basis considered in previous work \cite{Ghorai2019,Lin2019}, our finite subspace contains the full weight of the state when the channel is purely lossy. This is important as it ensures our numerics exactly reproduces the analytically solvable loss-only case. 

\subsubsection{Bound Weight \texorpdfstring{$W$}{} Outside Subspace}\label{dmcvqkdweight}
To bound the weight outside the subspace, we analytically solve the dual of the SDP in Eq. \eqref{Wsdp}. Our result is stated in the following theorem.
\begin{thm}[Bound on W for DMCVQKD]
\label{weightbound}
For the DMCVQKD protocol, with $\Pi^N$ as defined in Eq. \eqref{dmcvqkdprojection}, the weight outside the subspace is bounded by
\begin{equation}
W= \sum_{i=0}^{d-1} p(i) \frac{\expnsq-\expn}{N(N+1)}.
\end{equation}
\end{thm}

\begin{proof}
To prove this theorem, we will consider Bob's conditional states $\rho_B^i = \frac{1}{p(i)} \Tr_A\left[\rho_{AB} \left(\dyad{i}_A \otimes \mathbbm{1}_B\right) \right]$. Let $\PiNn\equiv\mathbbm{1}_B-\PiN$ and let $W_i\equiv\Tr(\rho_B^i \PiNn)$ be the weight of the $i^{th}$ conditional state. We first show that $W=\sum_i p(i) W_i$.
\begin{align}
\Tr(\rho \bar{\Pi}^N)&=\Tr[\rho \left( \sum_{i=0}^{d-1} \dyad{i}_A \otimes \PiNn\right) ]\\
&= \sum_{i=0}^{d-1} \Tr(\tilde{\rho}_B^i \PiNn)\\
&= \sum_{i=0}^{d-1} p(i) W_i. \label{conditionalweight}
\end{align}
Now, we only need to bound the weight of each conditional state. Using the constraints from Eq. \eqref{dmcvqkdinf}, each of these bounds can be expressed as a primal SDP.
\begin{IEEEeqnarray*}{uC'L}
$\displaystyle{\maximize_{\rho}}$& &\Tr(\PiNn \rho) \\
subject to:& 	&\Tr(\rho)=1\\
\IEEEyesnumber&	&\Tr ( \obsn \rho )=\expn\\
&	&\Tr ( \obsnsq \rho )=\expnsq\\
& &\rho\in \pos{\mathcal{H}_B}.
\end{IEEEeqnarray*}
In order to find an upper bound on this primal SDP we consider its dual. By weak duality, it holds that a feasible solution to the dual SDP upper bounds the primal. In fact, strong duality holds for this SDP, so this upper bound can be made tight.
\begin{IEEEeqnarray*}{uC'L}
$\displaystyle{\minimize_{\vec{y}}}$& &y_1+\expn y_2+\expnsq y_3 \\
\IEEEyesnumber subject to:& &	y_1 \mathbbm{1}_B + y_2 \obsn +y_3 \obsnsq - \PiNn \geq0\\
& &\vec{y} \in \mathbbm{R}^3\
\end{IEEEeqnarray*}
A feasible solution for the dual is $y_1=0$, $y_2=\frac{-1}{N(N+1)}$, $y_3=\frac{1}{N(N+1)}$. (This is in fact the optimal solution.) This can be easily verified as all operators in the constraint are diagonal in the $\ket{n_{\beta_i}}$ basis, so positivity is implied if and only if the diagonal entries are nonnegative. This solution leads to the objective value $W_i=\frac{\expnsq-\expn}{N(N+1)}$. Substituting into Eq. \eqref{conditionalweight}, the proof is complete.
\end{proof}

\subsubsection{Determine Correction Term \texorpdfstring{$\Delta$}{}} 
The correction term, as a function of the weight $W$ and key map register dimension $\abs{Z}$, is given in Theorem \ref{ucdupthm}. We use the value of $W$ determined in Sec. \ref{dmcvqkdweight}. Regardless of whether postselection is performed, $\abs{Z}=4$ as there are four non-discarded key outcomes (see Sec. \ref{correctiontermderivation}.

\subsubsection{Choose Finite Set \texorpdfstring{$\sfin$}{}}\label{dmcvqkdfinitesetdefinition}
By design, the projection $\Pi^N$ (Eq. \eqref{dmcvqkdprojection}) commutes with the positive observables $\Gamma_i$ (Eq. \eqref{dmcvqkdgammas}). We can thus use the form of $\sfin$ in Eq. \eqref{sfinoptim}. The finite-dimensional optimization is:
\begin{IEEEeqnarray*}{uC'L}
$\displaystyle{\minimize_{\rhot}}$& &f(\rhot) \\
subject to:& 	&1-W\leq\Tr(\rhot)\leq1\\
&	&\frac{1}{2}\norm{\Tr_B(\rhot)-\reduced}_1\leq \sqrt{2W-W^2} \\
\IEEEyesnumber&   &\Tr[\rhot\left(\tfrac{1}{p(i)}\dyad{i}{i} \otimes \obsn \right) ]\leq\expn \\
   & &\Tr[\rhot\left(\tfrac{1}{p(i)}\dyad{i}{i} \otimes \obsnsq \right) ]\leq\expnsq \\
& &\rhot\in \pos{\mathcal{H}_N}
\end{IEEEeqnarray*}
where $\reduced=\sum_{ij} \sqrt{p(i) p(j)} \braket{\alpha_j}{\alpha_i}\dyad{i}{j}$ and $W= \sum_i p(i) \frac{\expnsq-\expn}{N(N+1)}$. Note that the lower bounds in Eq. \eqref{sfinoptim} are not useful as the observables $\Gamma_i$ are unbounded.

In order to implement the optimization numerically, we need to choose a basis in which to represent our operators. The natural choice is $\{\ket{i}_A\otimes\ket{n_{\beta_i}}_B\}$. We compute the matrix elements of the observables and the objective function POVM in this basis in Appendix \ref{matrixops}. As this basis is not a standard construction for a bipartite Hilbert space, calculating the partial trace and its adjoint is slightly involved. We present these technical details in Appendix \ref{matrixops}.

\subsection{DMCVQKD with Trusted Detector Noise}
In the previous section, we have given a security proof for DMCVQKD assuming Bob's detector is ideal. We can extend this result to the scenario where Bob has imperfect but characterized detectors. We consider the specific model for detector imperfections given in \cite{Lin2020}. Namely, the two homodyne detectors comprising the heterodyne measurement have an associated efficiency and electronic noise. To illustrate our approach, we focus on the case where the two detectors have the same efficiency $\eta_d$ and electronic noise $\nu_{el}$. In this case, Bob's POVM elements are displaced thermal states, as opposed to coherent states \cite{Lin2020}. 

The protocol description is exactly the same as for the ideal detector case, except Bob's heterodyne detection is noisy. This noisy POVM enters into the optimization in two ways: changing the objective function $f$ and the observables $\Gamma_i$. To evaluate the objective function defined by a noisy POVM $f^{noisy}$, we need only calculate the new matrix elements. This calculation is presented in Appendix \ref{tnqo}. Critically, due to our coarse-graining, there is a simple relation between the noisy and ideal observables. Thus, our results on the bound of $W$ carry over from the ideal detector scenario. For any operator $A$, we denote its noisy counterpart by $\noisy{A}$. As shown in Appendix \ref{tnqo}, the ideal and noisy observables are related by linear combinations,
\begin{gather}
\label{nsyobs1}
\nsyobsn=\eta_d \obsn + \nu_{el} \mathbbm{1},\\
\label{nsyobs2}
\nsyobsnsq=\eta_d^2  \obsnsq+\eta_d (4\nu_{el}+1-\eta_d) \obsn +(2\nu_{el}^2+\nu_{el})\mathbbm{1}.
\end{gather} 
Bob measures the observables displaced by $\sqrt{\eta_d}\beta_i$. With these noisy expectations, the ideal ones can effectively be recreated by inverting the relationships in Eqs. \eqref{nsyobs1}, \eqref{nsyobs2}. Explicitly,
\begin{gather}
\expneff=\frac{\nsyexpn-\nu_{el}}{\eta_d},\\
\begin{split}
\expnsqeff=\frac{1}{\eta_d^2} \bigg(\nsyexpnsq-2\vel^2-\vel\\-(4\vel+1-\eta_d)\left(\nsyexpn-\vel\right)\bigg).
\end{split}
\end{gather}
The finite-dimensional optimization is then:
\begin{IEEEeqnarray*}{uC'L}
$\displaystyle{\minimize_{\rhot}}$& &f^{noisy}(\rhot) \\
subject to:& 	&1-W\leq\Tr(\rhot)\leq1\\
&	&\frac{1}{2}\norm{\Tr_B(\rhot)-\reduced}_1\leq \sqrt{2W-W^2} \\
\IEEEyesnumber &   &\Tr[\rhot\left(\tfrac{1}{p(i)}\dyad{i}{i} \otimes \obsn \right) ]\leq\expneff \\
   & &\Tr[\rhot\left(\tfrac{1}{p(i)}\dyad{i}{i} \otimes \obsnsq \right) ]\leq\expnsqeff \\
& &\rhot\in \pos{\mathcal{H}_N}
\end{IEEEeqnarray*}
where $\reduced=\sum_{ij} \sqrt{p(i) p(j)} \braket{\alpha_j}{\alpha_i}\dyad{i}{j}$ and $W= \sum_i p(i) \frac{\expnsqeff-\expneff}{N(N+1)}$.

\subsection{Simulation Results}\label{dmcvqkdresults}
\subsubsection{Simulation Parameters}
To understand the performance of this protocol and demonstrate our security proof approach, we simulate expectation values obtained from a typical experiment. In particular, we model the signal states as passing through a noisy and lossy Gaussian channel. The transmittance $\eta$ is modelled as a function of distance $d$ according to $\eta=10^{-k \cdot d/10}$, where $k$ is the attenuation factor of the channel. We use a typical value for commerical-grade fiber $k=0.2$ \textrm{dB/km}. The excess noise $\xi$ is taken to be fixed at the channel input, for example as preparation noise, so that Bob sees the effective noise $\delta=\eta \xi$. The expectation values for this simulation are $\expn=\delta/2$ and $\expnsq=\delta(1+\delta)/2$, as derived in Appendix \ref{expectationapp}. This implies $W=\delta^2/[2N(N+1)]$. We emphasize that our security proof does not depend on these parameter choices and simulation model, which are only used to illustrate the performance of the protocol in a typical implementation.

To account for realistic error-correction costs, $\delta^{leak}_{EC}$ in Eq. \eqref{convexkeyrate} is taken to be $H(Z)-\beta_{EC}I(Z:X)$, where $\beta_{EC}$ characterizes the error-correction efficiency. We use $\beta_{EC}=0.95$ as a representative value. The error-correction cost is calculated by simulating the joint probability distribution obtained by Alice and Bob (see Appendix \ref{expectationapp}).

All our algorithms are implemented in \textsc{Matlab} R2019B, using the convex optimization package CVX 2.1 \cite{Grant2014,Grant2008} with the \textsc{Mosek} 8.0.0.60 solver \cite{Mosek2016}. The Frank-Wolfe algorithm, with a maximum of 30 iterations, is used to solve the first step of the numerical method. All parameter optimizations use the fminbnd algorithm included in the \textsc{Matlab} distribution, which uses a combination of parabolic interpolation and golden-section search.

We emphasize that all key rate plots include the correction term unless stated otherwise. That is, $R^\infty=C_{num}-\delta^{leak}_{EC}-\Delta(W)$, where $C_{num}$ denotes the reliable numerical lower bound on $f(\nrho)$. In order to evaluate the effect of the correction term, and to compare with previous work using the photon number cutoff assumption, we will find it useful to consider the uncorrected values, defined as $C_{num}-\delta^{leak}_{EC}$. 

\subsubsection{Key Rate Plots}
We present key rate plots for different choices of channel parameters $\eta$ and $\xi$, protocol parameters $\alpha$, $\Delta_a$, and $\Delta_p$, and the subspace dimension parameter $N$. For the trusted noise scenario, we also consider $\eta_d$ and $\nu_{el}$.

In Fig. \ref{rawkeyrates}, we compare the key rates and uncorrected values from our dimension reduction method to the key rates under the photon number cutoff assumption obtained in \cite{Lin2019}. In order to enable a meaningful comparison, we use the protocol parameters from \cite{Lin2019}. The uncorrected values, which are equal to the key rate before subtracting the correction term $\Delta$, are essentially identical to those in \cite{Lin2019}. As the results from \cite{Lin2019} are an upper bound on the key rate, this indicates our choice of $\sfin$ is tight. Further, our corrected key rates are very close to the uncorrected values. This illustrates our correction term is small for reasonable values of the subspace dimension parameter $N$, at low channel excess noise (see Fig. \ref{correctionfrac}). 

\begin{figure}
\includegraphics[width=\linewidth]{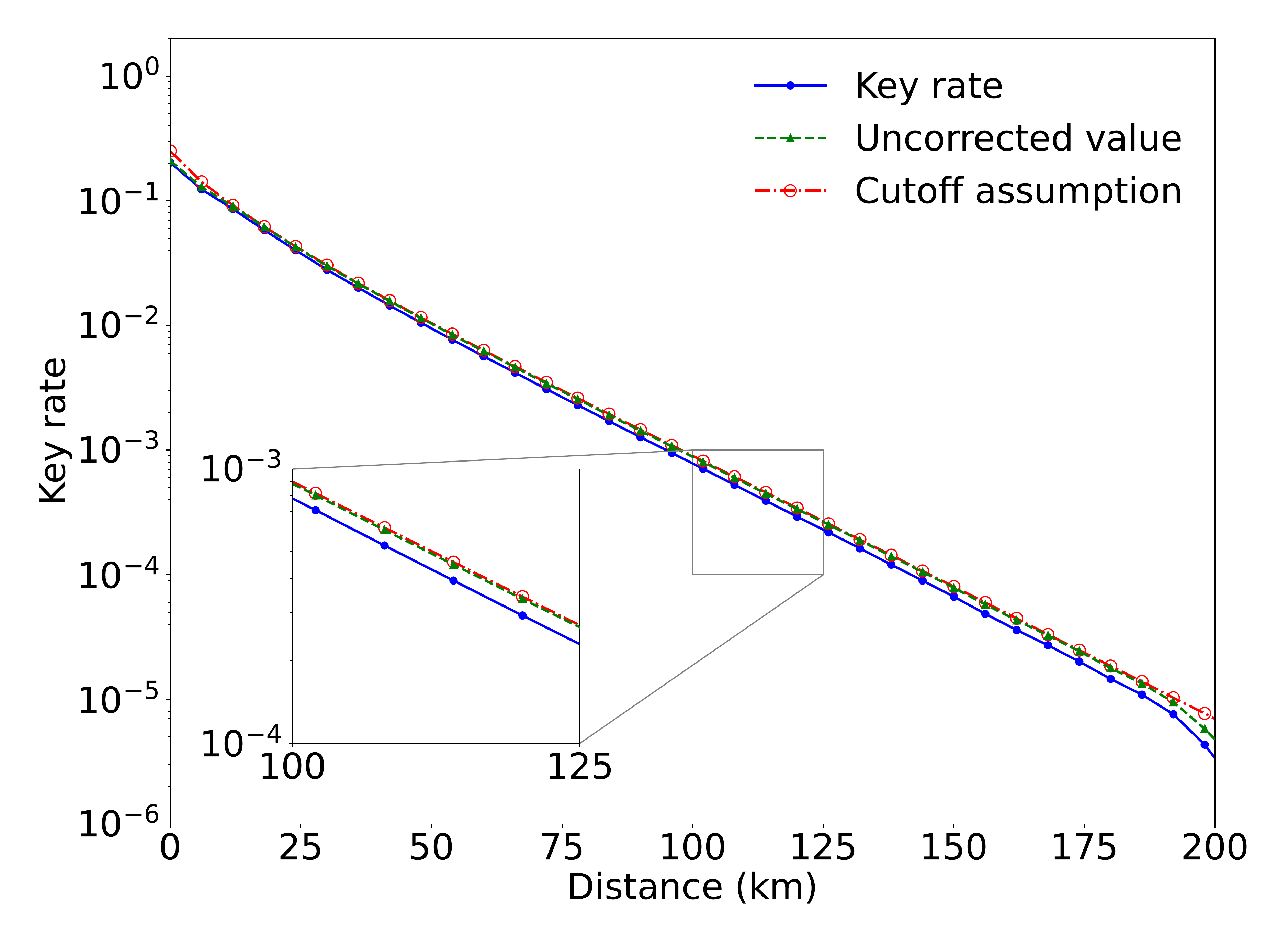}
\caption{Comparison of key rates and uncorrected values from our dimension reduction method with key rates under the photon number cutoff assumption from \cite{Lin2019}. Results are plotted versus distance with excess noise $\xi=0.01$, and are in the ideal detector scenario. Postselection parameters and signal state intensities from \cite{Lin2019} are used: $\alpha=0.6$, $\Delta_p=0$, and $\Delta_a$ is optimized with a coarse-grained search over $[0.5,0.65]$. The subspace dimension parameter is $N=20$.}
\label{rawkeyrates}
\end{figure}

As our key rates are very similar to the ones under the cutoff assumption, the qualitative conclusions of previous work \cite{Lin2019,Lin2020} are confirmed to hold under our precise treatment, without the previous working assumptions. We thus defer to \cite{Lin2019,Lin2020} for more extensive parameter exploration, and only focus on some important results in this section.

In Fig. \ref{optKRPSN3040max}, we plot the ideal detector key rates for different channel parameters. The signal state intensity and postselection parameters are numerically optimized for each distance and value of excess noise, using $N=10$. The key rates are calculated using $N=40$, except for a small number of points where we use $N=30$ to ameliorate numerics issues, as discussed in Appendix \ref{epsilonchange}. Using postselection extends the range of the protocol for high excess noise while also reducing the amount of data processing for error correction, which can be a bottleneck in actual implementations. For example, for $2\%$ excess noise, postselection increases the maximum distance by around $50$ km, while discarding $40$\% of the signals. The small number of outlying points that deviate from the trend are due to numerical issues inherent to convex solvers. We emphasize that these key rates are still rigorous, and can be improved by using higher numerical precision. 

\begin{figure}
\includegraphics[width=\linewidth]{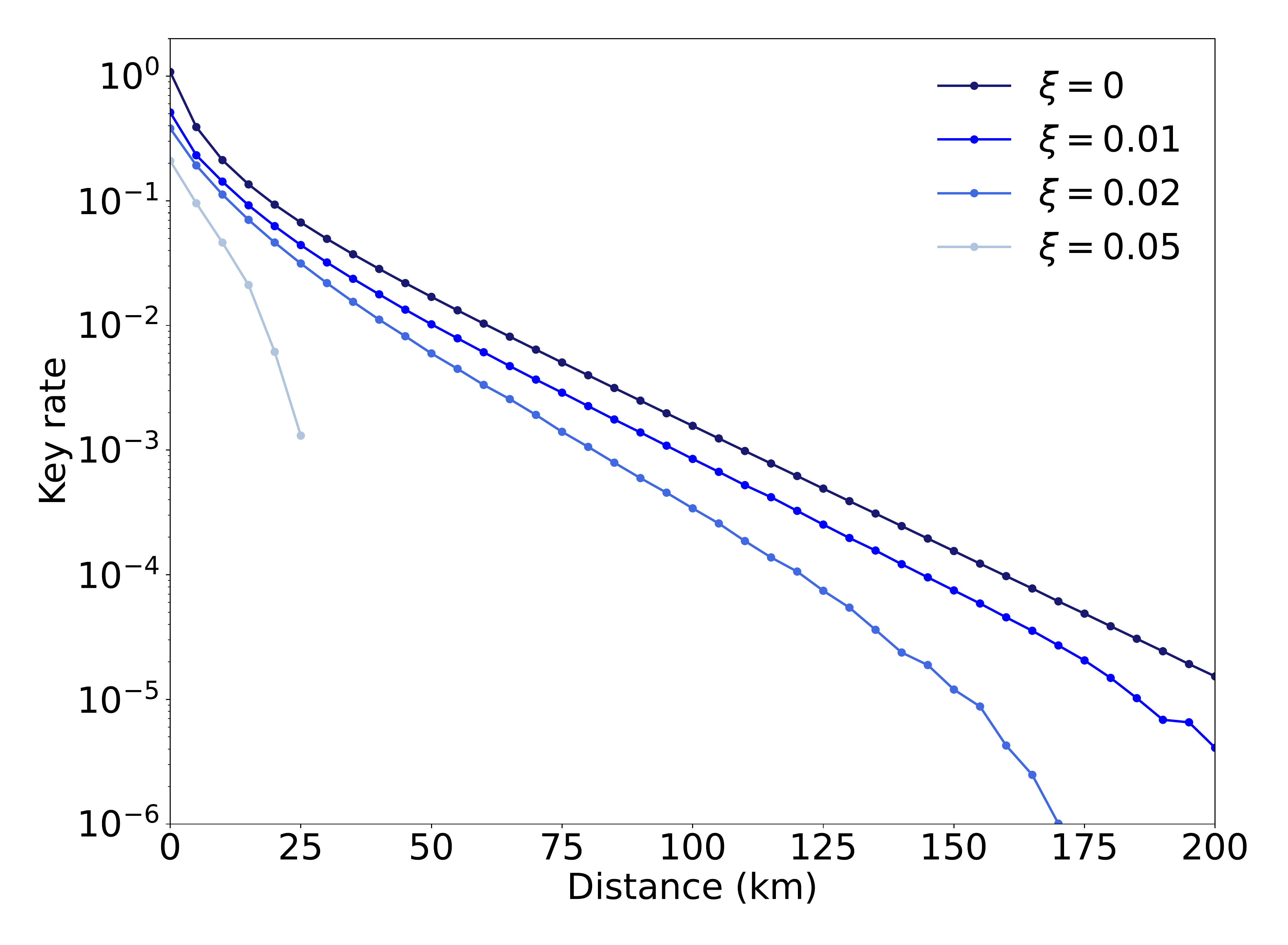}
\caption{Ideal detector secure key rates versus transmission distance, for different values of excess noise $\xi$, with optimized postselection parameters and signal state intensity. For each point, the better of the two results from $N=30$ and $N=40$ is used, with the majority of points from $N=40$. Postselection can improve the noise tolerance and range of the protocol.}
\label{optKRPSN3040max}
\end{figure}

To illustrate the relative size of the correction term, we plot it as a fraction of the uncorrected value in Fig. \ref{correctionfrac}, at a fixed distance of $15$ km. More precisely, for each value of excess noise $\xi$ and subspace dimension parameter $N$, we plot the fractional correction term $\Delta(W)/C_{num}$, where $W=(\eta\xi)^2/[2N(N+1)]$, while $C_{num}$ is computed at a fixed value of $N=40$ for each $\xi$. (In the range considered here, $C_{num}$ has a negligible dependence on $N$.) The protocol parameters are the same as in Fig. \ref{optKRPSN3040max}. We see that for a pure-loss channel, the correction term is zero since our subspace fully contains the simulation state, as discussed in Sec. \ref{dmcvqkdfiniteset}. For small values of excess noise, the correction term is negligible even for small $N$. For larger values of excess noise, and especially in the high loss regime, $N$ must be increased to obtain reasonable results. This is because the correction term scales like the loss $\eta$ for all values of excess noise $\xi$, while, for nonzero $\xi$, the key rate scales worse than $\eta$. It is an interesting avenue for future research to determine if a smaller correction term can be found.

\begin{figure}
\includegraphics[width=\linewidth]{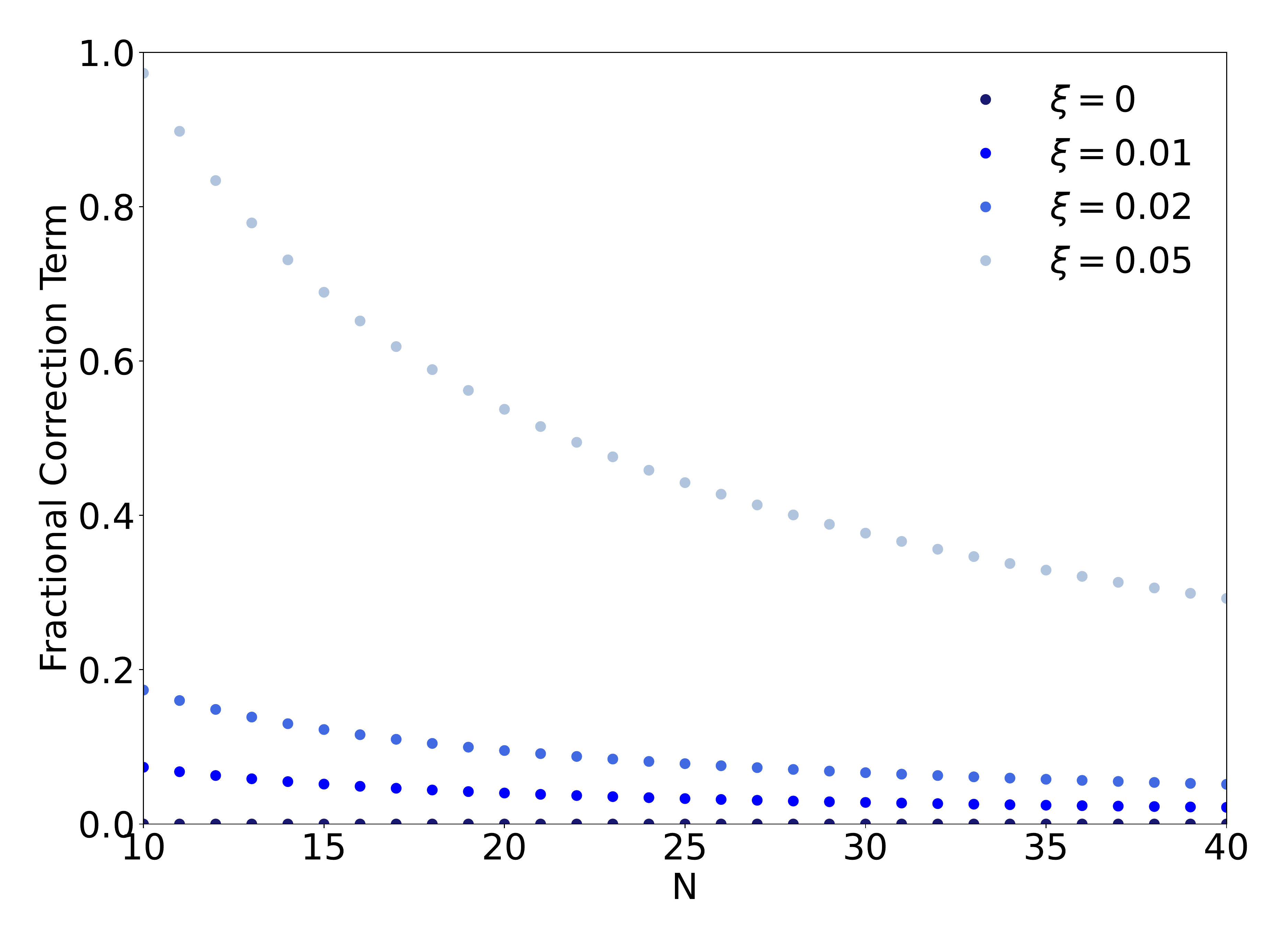}
\caption{The fractional correction term versus the subspace dimension parameter $N$ for different values of excess noise, in the ideal detector scenario. The distance is $15$ km and protocol parameters are optimized.}
\label{correctionfrac}
\end{figure}

The optimal signal state amplitudes $\alpha_{opt}$ are shown in Fig. \ref{optalphaN10}. The optimization range $[0.5,2]$ is sufficient for almost all parameter choices, though $\alpha_{opt}$ tends to infinity as distance and excess noise tend to zero. The general trend is that $\alpha_{opt}$ decreases as distance and excess noise increase. With $\alpha_{opt}$ fixed, $\Delta_a$ is optimized over $[0,1]$. (One could jointly optimize all protocol parameters, but we do not expect this to noticeably improve the key rates.) We find that the optimal value for $\Delta_p$ seems to always be zero, so phase postselection is omitted altogether. Thus, while the postselection pattern used here is a simple and intuitive one, it is an interesting future research topic to investigate other postselection patterns.

\begin{figure}
\centering
\includegraphics[width=\linewidth]{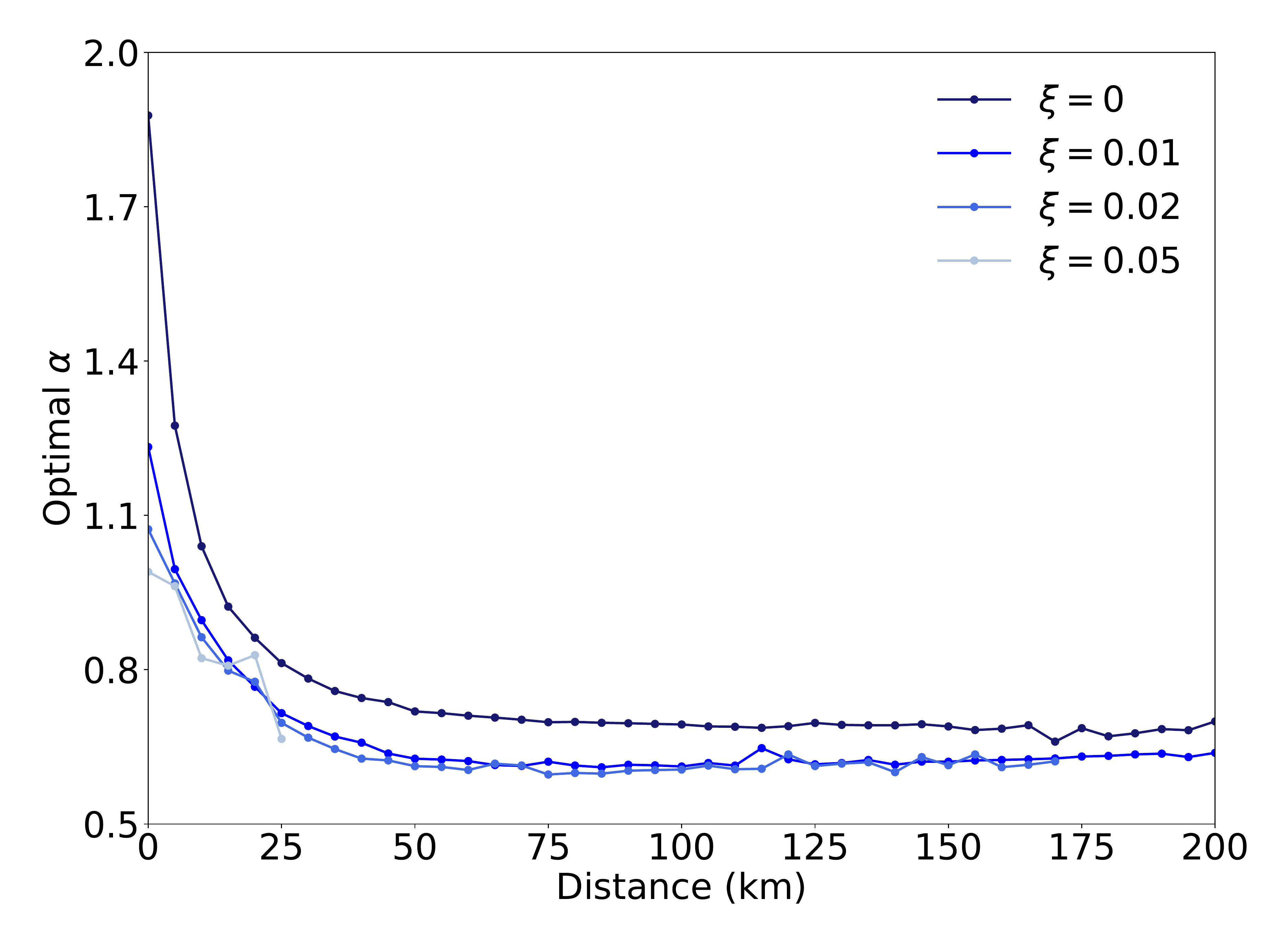}
\caption{Optimal signal state amplitude $\alpha_{opt}$ versus transmission distance, for different values of excess noise $\xi$ and in the ideal detector scenario. The amplitude is optimized in the range $[0.5,2]$, with $\Delta_a=\Delta_p=0$ and $N=10$.}
\label{optalphaN10}
\end{figure}

In Fig. \ref{TNfig}, we consider the key rates in the trusted detector noise scenario. We take the channel to have $1\%$ excess noise, and consider different values of detector efficiency and electronic noise. The protocol parameters are optimized for the ideal detector scenario, and the same parameters are used for each of the different trusted noise cases. We observe that even with large detector imperfections and $1$\% channel excess noise, it is possible to generate secure key at $200$ km. As expected, trusted detector noise does not significantly alter the scaling of the key rates. This is markedly different from the effect of channel excess noise (see Fig. \ref{optKRPSN3040max}).

\begin{figure}
\includegraphics[width=\linewidth]{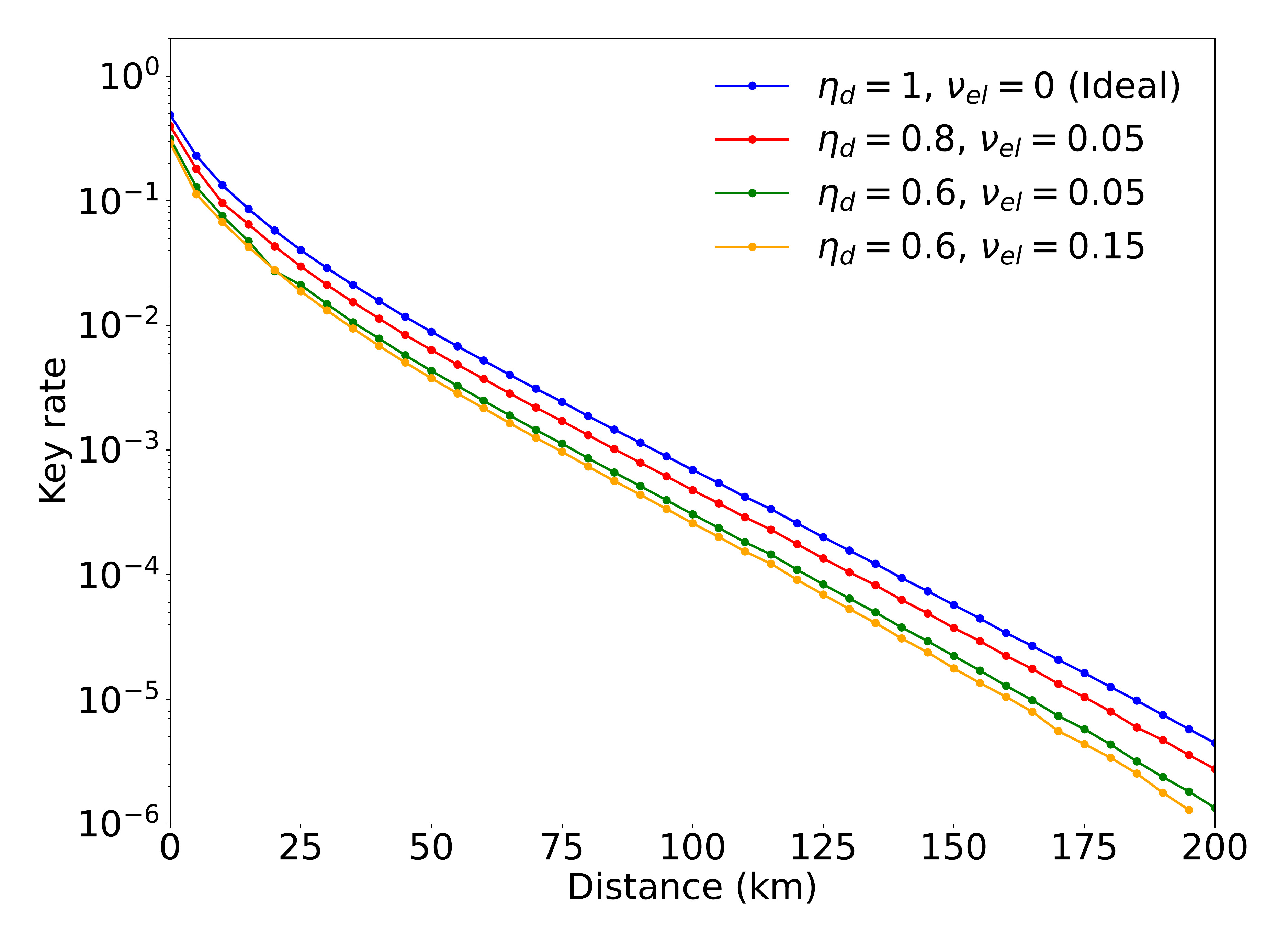}
\caption{Key rates versus distance for different trusted detector imperfections, with excess noise $\xi = 0.01$. Protocols are evaluated with the same optimized protocol parameters, including postselection, as in Fig. \ref{optKRPSN3040max}. The subspace dimension parameter is $N=10$.}
\label{TNfig}
\end{figure}

Discrete-modulated (DM) CVQKD is intended to be a more experimentally feasible alternative to Gaussian-modulated (GM) CVQKD \cite{Weedbrook2004,Fossier2009}. It is thus of interest to compare the performance of the two protocols. We perform a very basic comparison in Fig. \ref{GMDMcomparison}, using the same error-correction efficiency $\beta_{EC}=0.95$, detector loss $\eta_d=0.6$, and electronic noise $\nu_{el}=0.05$. The signal variance is optimized for GM, while the optimized protocol parameters for DM are the same as in Fig. \ref{TNfig}. We note that a complete and in-depth comparison of the two protocols would have to account for many more implementation details. For a pure-loss channel, the GM key rates are around an order of magnitude higher. At higher excess noise, the gap is larger as GM is more robust to channel noise. However, both protocols are largely unaffected by trusted detector imperfections. At $\xi=0.01$, a typical value for channel excess noise, the key rates scale similarly. Further, we expect the tolerance of DM to channel noise can be improved by using a larger constellation of signal states.

\begin{figure}
\includegraphics[width=\linewidth]{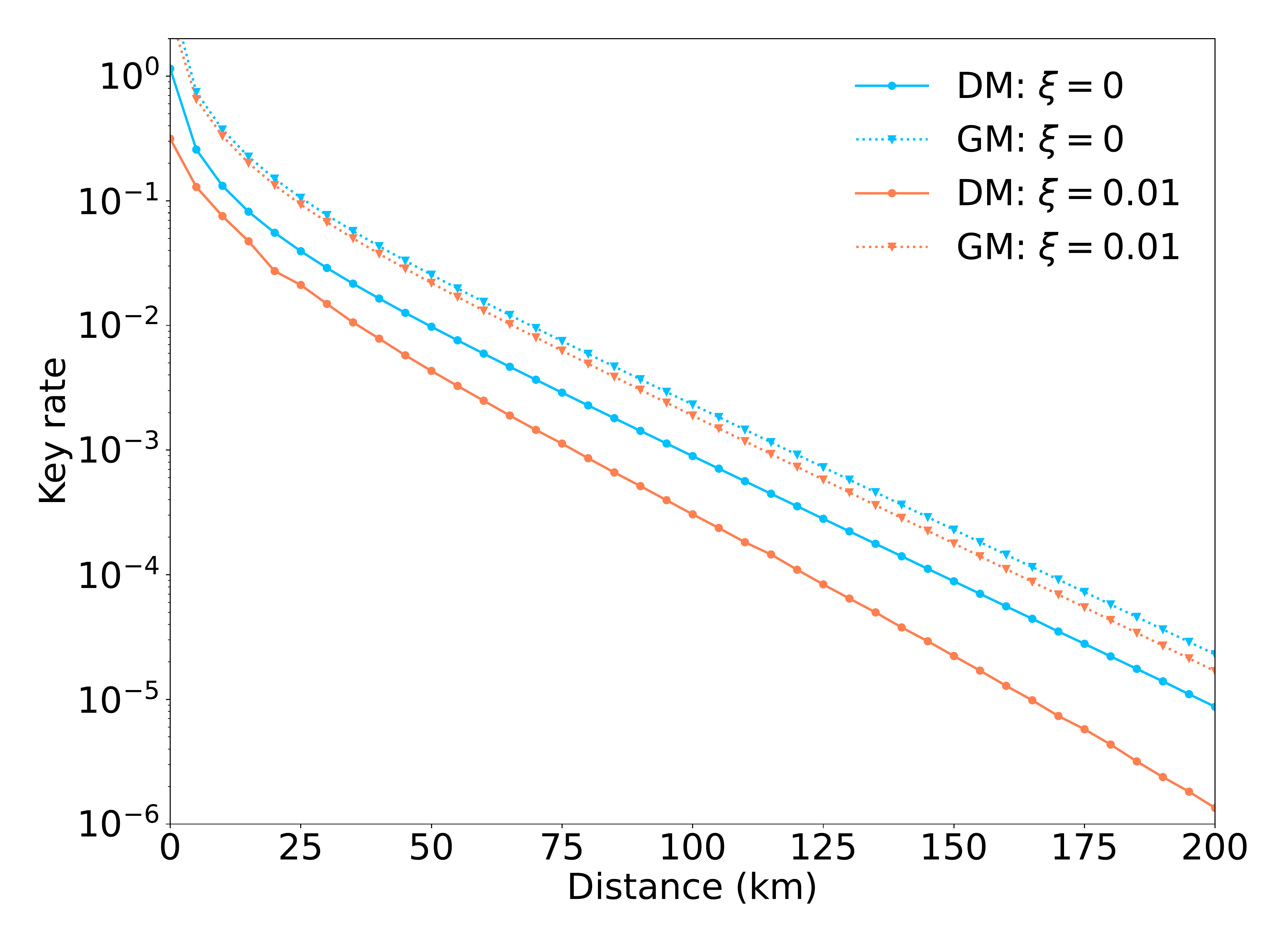}
\caption{Comparison of Gaussian \cite{Fossier2009} and discrete-modulation keyrates for different values of excess noise $\xi$, in the trusted noise scenario with $\eta_d=0.6$, $\nu_{el}=0.05$. Parameters for both protocols are optimized. The subspace dimension parameter is $N=10$.}
\label{GMDMcomparison}
\end{figure}

\section{Comparison to Flag-State Squasher}
\label{fss}
Our dimension reduction (DR) approach encompasses another method known as the flag-state squasher (FSS) \cite{Zhang2021}. The FSS also obtains lower bounds on the key rate by solving a finite-dimensional optimization. However, the FSS is restricted to protocols where both the key map POVM elements and constraint observables commute with the projection. Notably, this is not the case for DMCVQKD. In this section, we compare our method to the FSS both analytically and numerically. This demonstrates the advantages of our method and offers further insight into the FSS approach.

\subsection{Analytical Comparison}
We briefly summarize the FSS, deferring a complete description to \cite{Zhang2021}. As usual, Alice's POVM is given by $\{\dyad{j}_A\} $, while Bob's POVM is $\{\Gamma^i_B\}$. The corresponding probabilities are $\{\gamma_{ij}\}$. The FSS also requires choosing a projection $\Pi=\mathbbm{1}_A\otimes \Pi_B$ onto a finite subspace and upper-bounding the weight $W$ outside that subspace. It is assumed that $[\Gamma^i_B,\Pi_B]=0$.

Define a flag Hilbert space $\Hflag$, with dimension equal to the number of elements in Bob's POVM. The finite-dimensional optimization is over density matrices in $\Hfin\oplus\Hflag$. Bob's new POVM is $\tilde{\Gamma}^i=\Pi \ \Gamma^i\ \Pi \oplus \dyad{i}_F$. Alice's POVM and the expectation values are unchanged. The objective function is $f_{FSS}(\rhot_N\oplus\sigmat_F)=f(\rhot_N)$, i.e. it simply discards the flag portion and evaluates the usual key rate function on the remaining portion. This completes the formulation of the finite-dimensional optimization for the FSS.

The squashing map $\Lambda_B$ is a channel from $\Hinf$ to $\Hfin\oplus\Hflag$ defined as $\Lambda_B(\rho_\infty)=\Pi \rho_\infty \Pi \oplus \sum_i \Tr_B[\rho_\infty (\Gamma^i - \Pi \Gamma^i \Pi)] \dyad{i}_F$. Note that $\rho_\infty$ is feasible for the original infinite-dimensional optimization if and only if $\Lambda(\rho_\infty)$ is feasible for the finite-dimensional one. In this sense, we can think of the flag-state squasher as implicitly solving over the tightest possible choice of $\sfin$, namely $\Pi\sinf\Pi$. 

As a special case, our method can be applied to any protocol admitting a flag-state squasher. The FSS requires choosing a projection $\Pi$ and getting a bound on weight $W$, which establishes the first two steps of our method. Since the key map POVM elements commute with the projection, we can set $\Delta=0$ (Theorem \ref{nodelta}). All observables are POVM elements, so we can use the explicit form of $\sfin$ in Eq. \eqref{sfinoptim}. This establishes the last two steps of our method. We can now compare both approaches in the following theorem.

\begin{thm}
\label{fssdr}
For a fixed projection $\Pi$ and weight $W$, our dimension reduction (DR) method gives the same key rate as the flag-state squasher (FSS) when $\sfin=\Pi \sinf \Pi$.
\end{thm}
\begin{proof}
Let $\Lambda(\rho_\infty)=\rhot_N \oplus \sigmat_F$ be a state reaching the minimum in the FSS optimization. By definition, $f_{FSS}(\rhot_N \oplus \sigmat_F)=f_{DR}(\rhot_N)$. By the definition of $\Lambda$, $\rhot_N=\Pi\rho_\infty\Pi$. Thus, $\rhot_N\in\sfin$, so is feasible for the dimension reduction optimization. Since this optimization is a minimization, $R^\infty_{FSS}\geq R^\infty_{DR}$. 

Conversely, let $\rhot_N$ be a subnormalized state reaching the minimum in the dimension reduction optimization. By the definition of $\sfin$, $\rhot_N=\Pi \rho_\infty \Pi$ for some state $\rho_\infty\in\sinf$. By definition, $f_{DR}(\rhot_N)=f_{FSS}(\Lambda(\rho_\infty))$. By the property of the squashing map, $\Lambda(\rho_\infty)$ is feasible for the FSS optimization. Since the FSS optimization is a minimization, $R^\infty_{DR} \geq R^\infty_{FSS}$.
\end{proof}

If $\sfin$ is not chosen optimally, then our dimension reduction method gives a lower key rate. In practice, our explicit prescription for choosing $\sfin$ in Eq. \eqref{sfinoptim} gives very similar key rates to the flag state squasher (see Sec. \ref{fssnumerics}), suggesting this choice is essentially optimal.

In addition to being more general, our method has an important advantage compared to the flag-state squasher. Our finite-dimensional optimization is over a smaller Hilbert space, since we do not require flag-state dimensions. Therefore, if we compare fixed \emph{total} dimension, which roughly determines the runtime, our method can give higher key rates than the flag-state squasher. Some protocols can have a very large number of POVM constraints. For the flag-state squasher, using all the constraints would make the runtime prohibitive, as the dimension of the problem depends on the number of constraints. Thus, a smaller set of coarse-grained POVM elements is typically used. Our dimension reduction method is not limited by the number of constraints and can thus handle the fine-grained POVM directly, potentially giving better key rates.

\subsection{Numerical Comparison: Unbalanced Phase-Encoded BB84}\label{fssnumerics}
Having provided an analytical comparison between our method and the flag-state squasher, we now perform a sample numerical comparison of the two methods. We will consider the unbalanced phase-encoded BB84 protocol, to which the flag state squasher has recently been applied \cite{Li2020}. We defer a complete description of the protocol to \cite{Li2020}. Briefly, this is a phase-encoded BB84 protocol where Alice and Bob's interferometers each have a loss $1-\kappa$ only in the arm with the phase modulator; hence the term unbalanced. The projection is $\Pi_B=\sum_{\substack{0\leq n_1, n_2 \\ n_1+n_2\leq N}} \dyad{n_1,n_2}$, where $\ket{n_1,n_2}$ are two-mode Fock basis states. The weight $W$ is bounded using the fact that the frequency of cross-clicks increases with photon number \cite{Li2020}. 

In Fig. \ref{bb84fig} we compare the key rates from our method and the flag-state squasher, for a channel with transmittance $\eta$ and for different interferometer asymmetric transmittance $\kappa$. All parameters are the same as in Figure 3(a) of \cite{Li2020}, and the signal state intensity is optimized separately for each method and parameter choice. We see that our method gives essentially identical key rates. In conjunction with Theorem \ref{fssdr}, this provides strong numerical evidence that our heuristic choice of $\sfin$ in Eq. \eqref{sfinoptim} is tight. While a more thorough benchmarking would be in order, we remark that for generating the data in Fig. \ref{bb84fig}, our method was approximately five times faster than the flag-state squasher as implemented in \cite{Li2020} (using the same SDPT3 solver \cite{Toh1999,Tutuncu2003}).

\begin{figure}
\includegraphics[width=\linewidth]{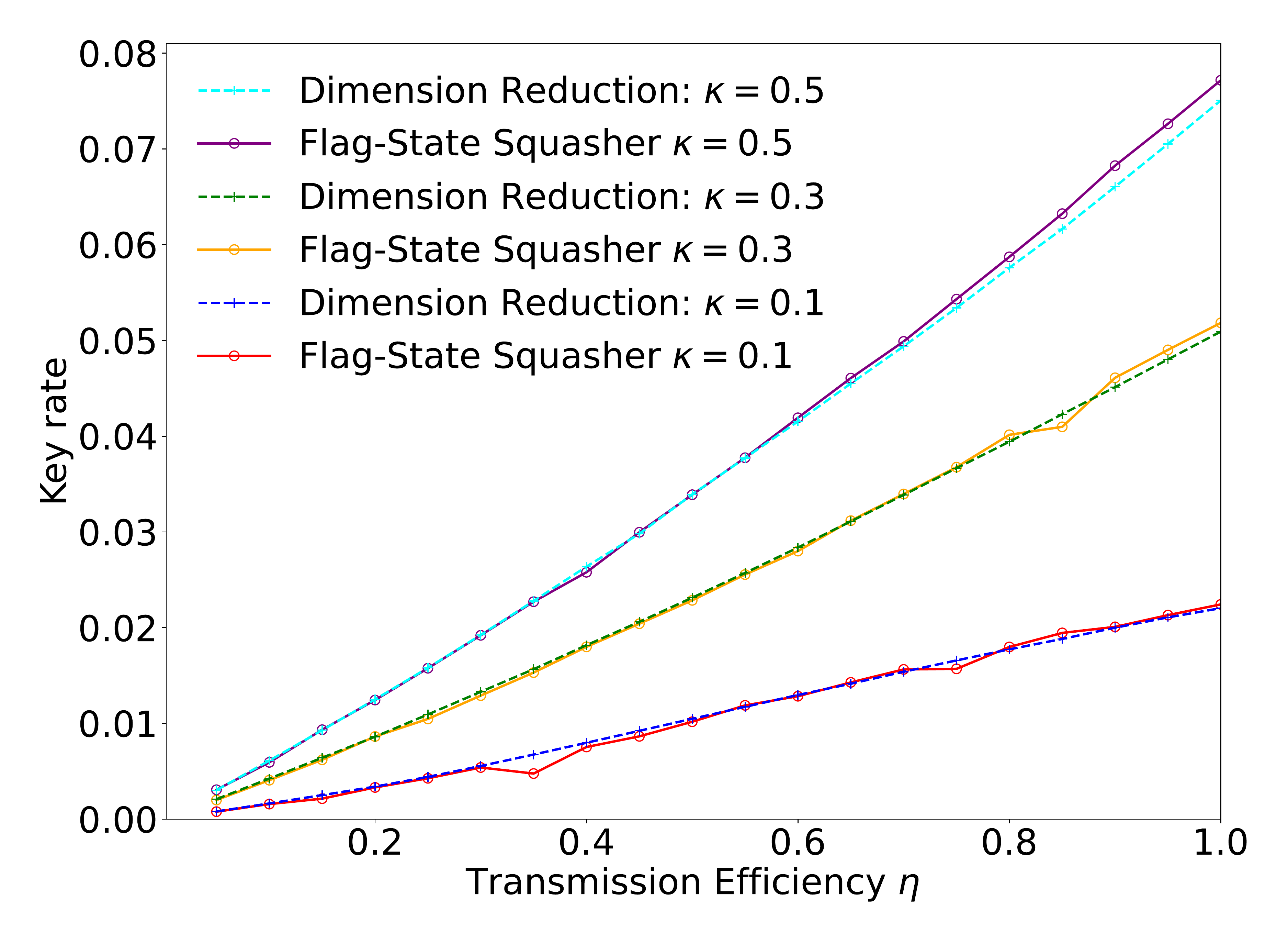}
\caption{Key rates for unbalanced phase-encoded BB84, versus transmission efficiency $\eta$, for different values of asymmetric interferometer loss $1-\kappa$. It is clear that the key rates from our dimension reduction method are nearly identical to those from the flag-state squasher, indicating the tightness of our method in practice. In generating the data for this graph, our dimension reduction method was approximately five times faster than the flag-state squasher as implemented in \cite{Li2020}.}
\label{bb84fig} 
\end{figure}

\section{Conclusion}\label{conclusion}
In summary, we establish a framework to lower bound a large dimensional convex optimization using a judiciously chosen smaller dimensional one. We show how this framework can be used to reduce the dimension of QKD key rate calculations. This allows existing numerical tools for finite-dimensional key rate calculations to be applied to protocols in infinite-dimensional Hilbert spaces. This allows us to do more detailed modelling of imperfections in devices. An important application of our method is to prove the asymptotic security of DMCVQKD with an arbitrary number of modulated states. As a concrete example, we apply this method to the quadrature phase-shift keying scheme in both ideal and trusted detector noise scenarios. We show that discrete modulation key rates can scale similarly to Gaussian modulation. Moreover, we rigorously demonstrate that postselection of data can improve the key rates for DMCVQKD. Using unbalanced phase-encoded BB84 as an example, we show that our approach can achieve key rates nearly identical to those from the flag-state squasher, while having an improved runtime.  

Some directions for future work are as follows. One may be able to make the correction term for the QKD objective function smaller by using the commutation relations of the POVM elements and the projection. Qualitatively, this would interpolate between the two cases we considered. There may also be tighter ways to construct the finite set. This could involve using additional properties of the constraint operators in specific cases. For DMCVQKD in particular, one may consider the effects of using more modulated states and different postselection regions in order to increase the key rate. Given the recent development of a finite-key numerical framework \cite{George2021}, we hope to extend our dimension reduction method to finite-key analysis of protocols in infinite-dimensional Hilbert spaces. We expect that key elements of the method, including bounding the weight outside the subspace and expanding the feasible set, will lift to the finite-key analysis. 

\begin{acknowledgments}
The Institute for Quantum Computing is supported in part by Innovation, Science, and Economic Development Canada. T.U. acknowledges the support of NSERC through the Alexander Graham Bell Canada Graduate Scholarship. T.V.H. acknowledges funding from Mitacs through the Mitacs Accelerate program grant. This research has been supported by NSERC under the Discovery Grants Program, Grant No. 341495, and under the Collaborative Research and Development Program, Grant No. CRDP J 522308-17. Financial support for this work has been partially provided by Huawei Technologies Canada Co., Ltd.
\end{acknowledgments}

\bibliographystyle{../modrevtex}
\bibliography{../mainbib.bib}

\begin{thebibliography}{48}%
\makeatletter
\providecommand \@ifxundefined [1]{%
 \@ifx{#1\undefined}
}%
\providecommand \@ifnum [1]{%
 \ifnum #1\expandafter \@firstoftwo
 \else \expandafter \@secondoftwo
 \fi
}%
\providecommand \@ifx [1]{%
 \ifx #1\expandafter \@firstoftwo
 \else \expandafter \@secondoftwo
 \fi
}%
\providecommand \natexlab [1]{#1}%
\providecommand \enquote  [1]{``#1''}%
\providecommand \bibnamefont  [1]{#1}%
\providecommand \bibfnamefont [1]{#1}%
\providecommand \citenamefont [1]{#1}%
\providecommand \href@noop [0]{\@secondoftwo}%
\providecommand \href [0]{\begingroup \@sanitize@url \@href}%
\providecommand \@href[1]{\@@startlink{#1}\@@href}%
\providecommand \@@href[1]{\endgroup#1\@@endlink}%
\providecommand \@sanitize@url [0]{\catcode `\\12\catcode `\$12\catcode
  `\&12\catcode `\#12\catcode `\^12\catcode `\_12\catcode `\%12\relax}%
\providecommand \@@startlink[1]{}%
\providecommand \@@endlink[0]{}%
\providecommand \url  [0]{\begingroup\@sanitize@url \@url }%
\providecommand \@url [1]{\endgroup\@href {#1}{\urlprefix }}%
\providecommand \urlprefix  [0]{URL }%
\providecommand \Eprint [0]{\href }%
\providecommand \doibase [0]{https://doi.org/}%
\providecommand \selectlanguage [0]{\@gobble}%
\providecommand \bibinfo  [0]{\@secondoftwo}%
\providecommand \bibfield  [0]{\@secondoftwo}%
\providecommand \translation [1]{[#1]}%
\providecommand \BibitemOpen [0]{}%
\providecommand \bibitemStop [0]{}%
\providecommand \bibitemNoStop [0]{.\EOS\space}%
\providecommand \EOS [0]{\spacefactor3000\relax}%
\providecommand \BibitemShut  [1]{\csname bibitem#1\endcsname}%
\let\auto@bib@innerbib\@empty
\bibitem [{\citenamefont {Bennett}\ and\ \citenamefont
  {Brassard}(1984)}]{Bennett1984}%
  \BibitemOpen
  \bibfield  {author} {\bibinfo {author} {\bibfnamefont {C.~H.}\ \bibnamefont
  {Bennett}}\ and\ \bibinfo {author} {\bibfnamefont {G.}~\bibnamefont
  {Brassard}},\ }\bibfield  {title} {\emph {\bibinfo {title} {Quantum
  Cryptography: Public Key Distribution and Coin Tossing}},\ }in\ \href@noop {}
  {\emph {\bibinfo {booktitle} {Proceedings of IEEE International Conference on
  Computers, Systems and Signal Processing}}}\ (\bibinfo  {publisher} {IEEE},\
  \bibinfo {address} {New York, USA},\ \bibinfo {year} {1984})\ pp.\ \bibinfo
  {pages} {175--179}\BibitemShut {NoStop}%
\bibitem [{\citenamefont {Ekert}(1991)}]{Ekert1991}%
  \BibitemOpen
  \bibfield  {author} {\bibinfo {author} {\bibfnamefont {A.~K.}\ \bibnamefont
  {Ekert}},\ }\bibfield  {title} {\emph {\bibinfo {title} {Quantum cryptography
  based on Bell's theorem}},\ }\href
  {https://doi.org/10.1103/PhysRevLett.67.661} {\bibfield  {journal} {\bibinfo
  {journal} {Phys. Rev. Lett.}\ }\textbf {\bibinfo {volume} {67}},\ \bibinfo
  {pages} {661} (\bibinfo {year} {1991})}\BibitemShut {NoStop}%
\bibitem [{\citenamefont {Colbeck}\ and\ \citenamefont
  {Renner}(2011)}]{Colbeck2011}%
  \BibitemOpen
  \bibfield  {author} {\bibinfo {author} {\bibfnamefont {R.}~\bibnamefont
  {Colbeck}}\ and\ \bibinfo {author} {\bibfnamefont {R.}~\bibnamefont
  {Renner}},\ }\bibfield  {title} {\emph {\bibinfo {title} {No extension of
  quantum theory can have improved predictive power}},\ }\href
  {https://doi.org/10.1038/ncomms1416} {\bibfield  {journal} {\bibinfo
  {journal} {Nature Communications}\ }\textbf {\bibinfo {volume} {2}},\
  \bibinfo {pages} {411} (\bibinfo {year} {2011})}\BibitemShut {NoStop}%
\bibitem [{\citenamefont {Scarani}\ \emph {et~al.}(2009)\citenamefont
  {Scarani}, \citenamefont {Bechmann-Pasquinucci}, \citenamefont {Cerf},
  \citenamefont {Du\ifmmode~\check{s}\else \v{s}\fi{}ek}, \citenamefont
  {L\"utkenhaus},\ and\ \citenamefont {Peev}}]{Scarani2009}%
  \BibitemOpen
  \bibfield  {author} {\bibinfo {author} {\bibfnamefont {V.}~\bibnamefont
  {Scarani}}, \bibinfo {author} {\bibfnamefont {H.}~\bibnamefont
  {Bechmann-Pasquinucci}}, \bibinfo {author} {\bibfnamefont {N.~J.}\
  \bibnamefont {Cerf}}, \bibinfo {author} {\bibfnamefont {M.}~\bibnamefont
  {Du\ifmmode~\check{s}\else \v{s}\fi{}ek}}, \bibinfo {author} {\bibfnamefont
  {N.}~\bibnamefont {L\"utkenhaus}},\ and\ \bibinfo {author} {\bibfnamefont
  {M.}~\bibnamefont {Peev}},\ }\bibfield  {title} {\emph {\bibinfo {title} {The
  security of practical quantum key distribution}},\ }\href
  {https://doi.org/10.1103/RevModPhys.81.1301} {\bibfield  {journal} {\bibinfo
  {journal} {Rev. Mod. Phys.}\ }\textbf {\bibinfo {volume} {81}},\ \bibinfo
  {pages} {1301} (\bibinfo {year} {2009})}\BibitemShut {NoStop}%
\bibitem [{\citenamefont {Xu}\ \emph {et~al.}(2020)\citenamefont {Xu},
  \citenamefont {Ma}, \citenamefont {Zhang}, \citenamefont {Lo},\ and\
  \citenamefont {Pan}}]{Xu2020}%
  \BibitemOpen
  \bibfield  {author} {\bibinfo {author} {\bibfnamefont {F.}~\bibnamefont
  {Xu}}, \bibinfo {author} {\bibfnamefont {X.}~\bibnamefont {Ma}}, \bibinfo
  {author} {\bibfnamefont {Q.}~\bibnamefont {Zhang}}, \bibinfo {author}
  {\bibfnamefont {H.-K.}\ \bibnamefont {Lo}},\ and\ \bibinfo {author}
  {\bibfnamefont {J.-W.}\ \bibnamefont {Pan}},\ }\bibfield  {title} {\emph
  {\bibinfo {title} {Secure quantum key distribution with realistic devices}},\
  }\href {https://doi.org/10.1103/RevModPhys.92.025002} {\bibfield  {journal}
  {\bibinfo  {journal} {Rev. Mod. Phys.}\ }\textbf {\bibinfo {volume} {92}},\
  \bibinfo {pages} {025002} (\bibinfo {year} {2020})}\BibitemShut {NoStop}%
\bibitem [{\citenamefont {Pirandola}\ \emph {et~al.}(2020)\citenamefont
  {Pirandola}, \citenamefont {Andersen}, \citenamefont {Banchi}, \citenamefont
  {Berta}, \citenamefont {Bunandar}, \citenamefont {Colbeck}, \citenamefont
  {Englund}, \citenamefont {Gehring}, \citenamefont {Lupo}, \citenamefont
  {Ottaviani}, \citenamefont {Pereira}, \citenamefont {Razavi}, \citenamefont
  {Shaari}, \citenamefont {Tomamichel}, \citenamefont {Usenko}, \citenamefont
  {Vallone}, \citenamefont {Villoresi},\ and\ \citenamefont
  {Wallden}}]{Pirandola2020}%
  \BibitemOpen
  \bibfield  {author} {\bibinfo {author} {\bibfnamefont {S.}~\bibnamefont
  {Pirandola}}, \bibinfo {author} {\bibfnamefont {U.~L.}\ \bibnamefont
  {Andersen}}, \bibinfo {author} {\bibfnamefont {L.}~\bibnamefont {Banchi}},
  \bibinfo {author} {\bibfnamefont {M.}~\bibnamefont {Berta}}, \bibinfo
  {author} {\bibfnamefont {D.}~\bibnamefont {Bunandar}}, \bibinfo {author}
  {\bibfnamefont {R.}~\bibnamefont {Colbeck}}, \bibinfo {author} {\bibfnamefont
  {D.}~\bibnamefont {Englund}}, \bibinfo {author} {\bibfnamefont
  {T.}~\bibnamefont {Gehring}}, \bibinfo {author} {\bibfnamefont
  {C.}~\bibnamefont {Lupo}}, \bibinfo {author} {\bibfnamefont {C.}~\bibnamefont
  {Ottaviani}}, \bibinfo {author} {\bibfnamefont {J.~L.}\ \bibnamefont
  {Pereira}}, \bibinfo {author} {\bibfnamefont {M.}~\bibnamefont {Razavi}},
  \bibinfo {author} {\bibfnamefont {J.~S.}\ \bibnamefont {Shaari}}, \bibinfo
  {author} {\bibfnamefont {M.}~\bibnamefont {Tomamichel}}, \bibinfo {author}
  {\bibfnamefont {V.~C.}\ \bibnamefont {Usenko}}, \bibinfo {author}
  {\bibfnamefont {G.}~\bibnamefont {Vallone}}, \bibinfo {author} {\bibfnamefont
  {P.}~\bibnamefont {Villoresi}},\ and\ \bibinfo {author} {\bibfnamefont
  {P.}~\bibnamefont {Wallden}},\ }\bibfield  {title} {\emph {\bibinfo {title}
  {Advances in quantum cryptography}},\ }\href
  {https://doi.org/10.1364/AOP.361502} {\bibfield  {journal} {\bibinfo
  {journal} {Adv. Opt. Photon.}\ }\textbf {\bibinfo {volume} {12}},\ \bibinfo
  {pages} {1012} (\bibinfo {year} {2020})}\BibitemShut {NoStop}%
\bibitem [{\citenamefont {Coles}\ \emph {et~al.}(2016)\citenamefont {Coles},
  \citenamefont {Metodiev},\ and\ \citenamefont {L\"{u}tkenhaus}}]{Coles2016}%
  \BibitemOpen
  \bibfield  {author} {\bibinfo {author} {\bibfnamefont {P.~J.}\ \bibnamefont
  {Coles}}, \bibinfo {author} {\bibfnamefont {E.~M.}\ \bibnamefont
  {Metodiev}},\ and\ \bibinfo {author} {\bibfnamefont {N.}~\bibnamefont
  {L\"{u}tkenhaus}},\ }\bibfield  {title} {\emph {\bibinfo {title} {Numerical
  approach for unstructured quantum key distribution}},\ }\href
  {https://doi.org/10.1038/ncomms11712} {\bibfield  {journal} {\bibinfo
  {journal} {Nature Communications}\ }\textbf {\bibinfo {volume} {7}},\
  \bibinfo {pages} {11712} (\bibinfo {year} {2016})}\BibitemShut {NoStop}%
\bibitem [{\citenamefont {Winick}\ \emph {et~al.}(2018)\citenamefont {Winick},
  \citenamefont {L{\"{u}}tkenhaus},\ and\ \citenamefont {Coles}}]{Winick2018}%
  \BibitemOpen
  \bibfield  {author} {\bibinfo {author} {\bibfnamefont {A.}~\bibnamefont
  {Winick}}, \bibinfo {author} {\bibfnamefont {N.}~\bibnamefont
  {L{\"{u}}tkenhaus}},\ and\ \bibinfo {author} {\bibfnamefont {P.~J.}\
  \bibnamefont {Coles}},\ }\bibfield  {title} {\emph {\bibinfo {title}
  {Reliable numerical key rates for quantum key distribution}},\ }\href
  {https://doi.org/10.22331/q-2018-07-26-77} {\bibfield  {journal} {\bibinfo
  {journal} {{Quantum}}\ }\textbf {\bibinfo {volume} {2}},\ \bibinfo {pages}
  {77} (\bibinfo {year} {2018})}\BibitemShut {NoStop}%
\bibitem [{\citenamefont {Beaudry}\ \emph {et~al.}(2008)\citenamefont
  {Beaudry}, \citenamefont {Moroder},\ and\ \citenamefont
  {L\"utkenhaus}}]{Beaudry2008}%
  \BibitemOpen
  \bibfield  {author} {\bibinfo {author} {\bibfnamefont {N.~J.}\ \bibnamefont
  {Beaudry}}, \bibinfo {author} {\bibfnamefont {T.}~\bibnamefont {Moroder}},\
  and\ \bibinfo {author} {\bibfnamefont {N.}~\bibnamefont {L\"utkenhaus}},\
  }\bibfield  {title} {\emph {\bibinfo {title} {Squashing Models for Optical
  Measurements in Quantum Communication}},\ }\href
  {https://doi.org/10.1103/PhysRevLett.101.093601} {\bibfield  {journal}
  {\bibinfo  {journal} {Phys. Rev. Lett.}\ }\textbf {\bibinfo {volume} {101}},\
  \bibinfo {pages} {093601} (\bibinfo {year} {2008})}\BibitemShut {NoStop}%
\bibitem [{\citenamefont {Tsurumaru}(2010)}]{Tsurumaru2010}%
  \BibitemOpen
  \bibfield  {author} {\bibinfo {author} {\bibfnamefont {T.}~\bibnamefont
  {Tsurumaru}},\ }\bibfield  {title} {\emph {\bibinfo {title} {Squash operator
  and symmetry}},\ }\href {https://doi.org/10.1103/PhysRevA.81.012328}
  {\bibfield  {journal} {\bibinfo  {journal} {Phys. Rev. A}\ }\textbf {\bibinfo
  {volume} {81}},\ \bibinfo {pages} {012328} (\bibinfo {year}
  {2010})}\BibitemShut {NoStop}%
\bibitem [{\citenamefont {Gittsovich}\ \emph {et~al.}(2014)\citenamefont
  {Gittsovich}, \citenamefont {Beaudry}, \citenamefont {Narasimhachar},
  \citenamefont {Alvarez}, \citenamefont {Moroder},\ and\ \citenamefont
  {L\"utkenhaus}}]{Gittsovich2014}%
  \BibitemOpen
  \bibfield  {author} {\bibinfo {author} {\bibfnamefont {O.}~\bibnamefont
  {Gittsovich}}, \bibinfo {author} {\bibfnamefont {N.~J.}\ \bibnamefont
  {Beaudry}}, \bibinfo {author} {\bibfnamefont {V.}~\bibnamefont
  {Narasimhachar}}, \bibinfo {author} {\bibfnamefont {R.~R.}\ \bibnamefont
  {Alvarez}}, \bibinfo {author} {\bibfnamefont {T.}~\bibnamefont {Moroder}},\
  and\ \bibinfo {author} {\bibfnamefont {N.}~\bibnamefont {L\"utkenhaus}},\
  }\bibfield  {title} {\emph {\bibinfo {title} {Squashing model for detectors
  and applications to quantum-key-distribution protocols}},\ }\href
  {https://doi.org/10.1103/PhysRevA.89.012325} {\bibfield  {journal} {\bibinfo
  {journal} {Phys. Rev. A}\ }\textbf {\bibinfo {volume} {89}},\ \bibinfo
  {pages} {012325} (\bibinfo {year} {2014})}\BibitemShut {NoStop}%
\bibitem [{\citenamefont {Zhang}\ \emph {et~al.}(2021)\citenamefont {Zhang},
  \citenamefont {Coles}, \citenamefont {Winick}, \citenamefont {Lin},\ and\
  \citenamefont {L\"utkenhaus}}]{Zhang2021}%
  \BibitemOpen
  \bibfield  {author} {\bibinfo {author} {\bibfnamefont {Y.}~\bibnamefont
  {Zhang}}, \bibinfo {author} {\bibfnamefont {P.~J.}\ \bibnamefont {Coles}},
  \bibinfo {author} {\bibfnamefont {A.}~\bibnamefont {Winick}}, \bibinfo
  {author} {\bibfnamefont {J.}~\bibnamefont {Lin}},\ and\ \bibinfo {author}
  {\bibfnamefont {N.}~\bibnamefont {L\"utkenhaus}},\ }\bibfield  {title} {\emph
  {\bibinfo {title} {Security proof of practical quantum key distribution with
  detection-efficiency mismatch}},\ }\href
  {https://doi.org/10.1103/PhysRevResearch.3.013076} {\bibfield  {journal}
  {\bibinfo  {journal} {Phys. Rev. Research}\ }\textbf {\bibinfo {volume}
  {3}},\ \bibinfo {pages} {013076} (\bibinfo {year} {2021})}\BibitemShut
  {NoStop}%
\bibitem [{\citenamefont {Li}\ and\ \citenamefont
  {L\"utkenhaus}(2020)}]{Li2020}%
  \BibitemOpen
  \bibfield  {author} {\bibinfo {author} {\bibfnamefont {N.~K.~H.}\
  \bibnamefont {Li}}\ and\ \bibinfo {author} {\bibfnamefont {N.}~\bibnamefont
  {L\"utkenhaus}},\ }\bibfield  {title} {\emph {\bibinfo {title} {Improving key
  rates of the unbalanced phase-encoded BB84 protocol using the flag-state
  squashing model}},\ }\href {https://doi.org/10.1103/PhysRevResearch.2.043172}
  {\bibfield  {journal} {\bibinfo  {journal} {Phys. Rev. Research}\ }\textbf
  {\bibinfo {volume} {2}},\ \bibinfo {pages} {043172} (\bibinfo {year}
  {2020})}\BibitemShut {NoStop}%
\bibitem [{\citenamefont {Ralph}(1999)}]{Ralph1999}%
  \BibitemOpen
  \bibfield  {author} {\bibinfo {author} {\bibfnamefont {T.~C.}\ \bibnamefont
  {Ralph}},\ }\bibfield  {title} {\emph {\bibinfo {title} {Continuous variable
  quantum cryptography}},\ }\href {https://doi.org/10.1103/PhysRevA.61.010303}
  {\bibfield  {journal} {\bibinfo  {journal} {Phys. Rev. A}\ }\textbf {\bibinfo
  {volume} {61}},\ \bibinfo {pages} {010303} (\bibinfo {year}
  {1999})}\BibitemShut {NoStop}%
\bibitem [{\citenamefont {Hillery}(2000)}]{Hillery2000}%
  \BibitemOpen
  \bibfield  {author} {\bibinfo {author} {\bibfnamefont {M.}~\bibnamefont
  {Hillery}},\ }\bibfield  {title} {\emph {\bibinfo {title} {Quantum
  cryptography with squeezed states}},\ }\href
  {https://doi.org/10.1103/PhysRevA.61.022309} {\bibfield  {journal} {\bibinfo
  {journal} {Phys. Rev. A}\ }\textbf {\bibinfo {volume} {61}},\ \bibinfo
  {pages} {022309} (\bibinfo {year} {2000})}\BibitemShut {NoStop}%
\bibitem [{\citenamefont {Silberhorn}\ \emph {et~al.}(2002)\citenamefont
  {Silberhorn}, \citenamefont {Ralph}, \citenamefont {L\"utkenhaus},\ and\
  \citenamefont {Leuchs}}]{Silberhorn2002}%
  \BibitemOpen
  \bibfield  {author} {\bibinfo {author} {\bibfnamefont {C.}~\bibnamefont
  {Silberhorn}}, \bibinfo {author} {\bibfnamefont {T.~C.}\ \bibnamefont
  {Ralph}}, \bibinfo {author} {\bibfnamefont {N.}~\bibnamefont
  {L\"utkenhaus}},\ and\ \bibinfo {author} {\bibfnamefont {G.}~\bibnamefont
  {Leuchs}},\ }\bibfield  {title} {\emph {\bibinfo {title} {Continuous Variable
  Quantum Cryptography: Beating the 3 dB Loss Limit}},\ }\href
  {https://doi.org/10.1103/PhysRevLett.89.167901} {\bibfield  {journal}
  {\bibinfo  {journal} {Phys. Rev. Lett.}\ }\textbf {\bibinfo {volume} {89}},\
  \bibinfo {pages} {167901} (\bibinfo {year} {2002})}\BibitemShut {NoStop}%
\bibitem [{\citenamefont {Grosshans}\ and\ \citenamefont
  {Grangier}(2002)}]{Grosshans2002}%
  \BibitemOpen
  \bibfield  {author} {\bibinfo {author} {\bibfnamefont {F.}~\bibnamefont
  {Grosshans}}\ and\ \bibinfo {author} {\bibfnamefont {P.}~\bibnamefont
  {Grangier}},\ }\bibfield  {title} {\emph {\bibinfo {title} {Continuous
  Variable Quantum Cryptography Using Coherent States}},\ }\href
  {https://doi.org/10.1103/PhysRevLett.88.057902} {\bibfield  {journal}
  {\bibinfo  {journal} {Phys. Rev. Lett.}\ }\textbf {\bibinfo {volume} {88}},\
  \bibinfo {pages} {057902} (\bibinfo {year} {2002})}\BibitemShut {NoStop}%
\bibitem [{\citenamefont {Grosshans}\ \emph
  {et~al.}(2003{\natexlab{a}})\citenamefont {Grosshans}, \citenamefont
  {Van~Assche}, \citenamefont {Wenger}, \citenamefont {Brouri}, \citenamefont
  {Cerf},\ and\ \citenamefont {Grangier}}]{Grosshans2003}%
  \BibitemOpen
  \bibfield  {author} {\bibinfo {author} {\bibfnamefont {F.}~\bibnamefont
  {Grosshans}}, \bibinfo {author} {\bibfnamefont {G.}~\bibnamefont
  {Van~Assche}}, \bibinfo {author} {\bibfnamefont {J.}~\bibnamefont {Wenger}},
  \bibinfo {author} {\bibfnamefont {R.}~\bibnamefont {Brouri}}, \bibinfo
  {author} {\bibfnamefont {N.~J.}\ \bibnamefont {Cerf}},\ and\ \bibinfo
  {author} {\bibfnamefont {P.}~\bibnamefont {Grangier}},\ }\bibfield  {title}
  {\emph {\bibinfo {title} {Quantum key distribution using gaussian-modulated
  coherent states}},\ }\href {https://doi.org/10.1038/nature01289} {\bibfield
  {journal} {\bibinfo  {journal} {Nature}\ }\textbf {\bibinfo {volume} {421}},\
  \bibinfo {pages} {238} (\bibinfo {year} {2003}{\natexlab{a}})}\BibitemShut
  {NoStop}%
\bibitem [{\citenamefont {Weedbrook}\ \emph {et~al.}(2004)\citenamefont
  {Weedbrook}, \citenamefont {Lance}, \citenamefont {Bowen}, \citenamefont
  {Symul}, \citenamefont {Ralph},\ and\ \citenamefont {Lam}}]{Weedbrook2004}%
  \BibitemOpen
  \bibfield  {author} {\bibinfo {author} {\bibfnamefont {C.}~\bibnamefont
  {Weedbrook}}, \bibinfo {author} {\bibfnamefont {A.~M.}\ \bibnamefont
  {Lance}}, \bibinfo {author} {\bibfnamefont {W.~P.}\ \bibnamefont {Bowen}},
  \bibinfo {author} {\bibfnamefont {T.}~\bibnamefont {Symul}}, \bibinfo
  {author} {\bibfnamefont {T.~C.}\ \bibnamefont {Ralph}},\ and\ \bibinfo
  {author} {\bibfnamefont {P.~K.}\ \bibnamefont {Lam}},\ }\bibfield  {title}
  {\emph {\bibinfo {title} {Quantum Cryptography Without Switching}},\ }\href
  {https://doi.org/10.1103/PhysRevLett.93.170504} {\bibfield  {journal}
  {\bibinfo  {journal} {Phys. Rev. Lett.}\ }\textbf {\bibinfo {volume} {93}},\
  \bibinfo {pages} {170504} (\bibinfo {year} {2004})}\BibitemShut {NoStop}%
\bibitem [{\citenamefont {Zhao}\ \emph {et~al.}(2009)\citenamefont {Zhao},
  \citenamefont {Heid}, \citenamefont {Rigas},\ and\ \citenamefont
  {L\"utkenhaus}}]{Zhao2009}%
  \BibitemOpen
  \bibfield  {author} {\bibinfo {author} {\bibfnamefont {Y.-B.}\ \bibnamefont
  {Zhao}}, \bibinfo {author} {\bibfnamefont {M.}~\bibnamefont {Heid}}, \bibinfo
  {author} {\bibfnamefont {J.}~\bibnamefont {Rigas}},\ and\ \bibinfo {author}
  {\bibfnamefont {N.}~\bibnamefont {L\"utkenhaus}},\ }\bibfield  {title} {\emph
  {\bibinfo {title} {Asymptotic security of binary modulated
  continuous-variable quantum key distribution under collective attacks}},\
  }\href {https://doi.org/10.1103/PhysRevA.79.012307} {\bibfield  {journal}
  {\bibinfo  {journal} {Phys. Rev. A}\ }\textbf {\bibinfo {volume} {79}},\
  \bibinfo {pages} {012307} (\bibinfo {year} {2009})}\BibitemShut {NoStop}%
\bibitem [{\citenamefont {Br\'adler}\ and\ \citenamefont
  {Weedbrook}(2018)}]{Bradler2018}%
  \BibitemOpen
  \bibfield  {author} {\bibinfo {author} {\bibfnamefont {K.}~\bibnamefont
  {Br\'adler}}\ and\ \bibinfo {author} {\bibfnamefont {C.}~\bibnamefont
  {Weedbrook}},\ }\bibfield  {title} {\emph {\bibinfo {title} {Security proof
  of continuous-variable quantum key distribution using three coherent
  states}},\ }\href {https://doi.org/10.1103/PhysRevA.97.022310} {\bibfield
  {journal} {\bibinfo  {journal} {Phys. Rev. A}\ }\textbf {\bibinfo {volume}
  {97}},\ \bibinfo {pages} {022310} (\bibinfo {year} {2018})}\BibitemShut
  {NoStop}%
\bibitem [{\citenamefont {Ghorai}\ \emph {et~al.}(2019)\citenamefont {Ghorai},
  \citenamefont {Grangier}, \citenamefont {Diamanti},\ and\ \citenamefont
  {Leverrier}}]{Ghorai2019}%
  \BibitemOpen
  \bibfield  {author} {\bibinfo {author} {\bibfnamefont {S.}~\bibnamefont
  {Ghorai}}, \bibinfo {author} {\bibfnamefont {P.}~\bibnamefont {Grangier}},
  \bibinfo {author} {\bibfnamefont {E.}~\bibnamefont {Diamanti}},\ and\
  \bibinfo {author} {\bibfnamefont {A.}~\bibnamefont {Leverrier}},\ }\bibfield
  {title} {\emph {\bibinfo {title} {Asymptotic Security of Continuous-Variable
  Quantum Key Distribution with a Discrete Modulation}},\ }\href
  {https://doi.org/10.1103/PhysRevX.9.021059} {\bibfield  {journal} {\bibinfo
  {journal} {Phys. Rev. X}\ }\textbf {\bibinfo {volume} {9}},\ \bibinfo {pages}
  {021059} (\bibinfo {year} {2019})}\BibitemShut {NoStop}%
\bibitem [{\citenamefont {Lin}\ \emph {et~al.}(2019)\citenamefont {Lin},
  \citenamefont {Upadhyaya},\ and\ \citenamefont {L\"utkenhaus}}]{Lin2019}%
  \BibitemOpen
  \bibfield  {author} {\bibinfo {author} {\bibfnamefont {J.}~\bibnamefont
  {Lin}}, \bibinfo {author} {\bibfnamefont {T.}~\bibnamefont {Upadhyaya}},\
  and\ \bibinfo {author} {\bibfnamefont {N.}~\bibnamefont {L\"utkenhaus}},\
  }\bibfield  {title} {\emph {\bibinfo {title} {Asymptotic Security Analysis of
  Discrete-Modulated Continuous-Variable Quantum Key Distribution}},\ }\href
  {https://doi.org/10.1103/PhysRevX.9.041064} {\bibfield  {journal} {\bibinfo
  {journal} {Phys. Rev. X}\ }\textbf {\bibinfo {volume} {9}},\ \bibinfo {pages}
  {041064} (\bibinfo {year} {2019})}\BibitemShut {NoStop}%
\bibitem [{\citenamefont {Matsuura}\ \emph {et~al.}(2021)\citenamefont
  {Matsuura}, \citenamefont {Maeda}, \citenamefont {Sasaki},\ and\
  \citenamefont {Koashi}}]{Matsuura2021}%
  \BibitemOpen
  \bibfield  {author} {\bibinfo {author} {\bibfnamefont {T.}~\bibnamefont
  {Matsuura}}, \bibinfo {author} {\bibfnamefont {K.}~\bibnamefont {Maeda}},
  \bibinfo {author} {\bibfnamefont {T.}~\bibnamefont {Sasaki}},\ and\ \bibinfo
  {author} {\bibfnamefont {M.}~\bibnamefont {Koashi}},\ }\bibfield  {title}
  {\emph {\bibinfo {title} {Finite-size security of continuous-variable quantum
  key distribution with digital signal processing}},\ }\href
  {https://doi.org/10.1038/s41467-020-19916-1} {\bibfield  {journal} {\bibinfo
  {journal} {Nature Communications}\ }\textbf {\bibinfo {volume} {12}},\
  \bibinfo {pages} {252} (\bibinfo {year} {2021})}\BibitemShut {NoStop}%
\bibitem [{\citenamefont {Killoran}\ and\ \citenamefont
  {L\"utkenhaus}(2011)}]{Killoran2011}%
  \BibitemOpen
  \bibfield  {author} {\bibinfo {author} {\bibfnamefont {N.}~\bibnamefont
  {Killoran}}\ and\ \bibinfo {author} {\bibfnamefont {N.}~\bibnamefont
  {L\"utkenhaus}},\ }\bibfield  {title} {\emph {\bibinfo {title} {Strong
  quantitative benchmarking of quantum optical devices}},\ }\href
  {https://doi.org/10.1103/PhysRevA.83.052320} {\bibfield  {journal} {\bibinfo
  {journal} {Phys. Rev. A}\ }\textbf {\bibinfo {volume} {83}},\ \bibinfo
  {pages} {052320} (\bibinfo {year} {2011})}\BibitemShut {NoStop}%
\bibitem [{\citenamefont {Bennett}\ \emph {et~al.}(1992)\citenamefont
  {Bennett}, \citenamefont {Brassard},\ and\ \citenamefont
  {Mermin}}]{Bennett1992}%
  \BibitemOpen
  \bibfield  {author} {\bibinfo {author} {\bibfnamefont {C.~H.}\ \bibnamefont
  {Bennett}}, \bibinfo {author} {\bibfnamefont {G.}~\bibnamefont {Brassard}},\
  and\ \bibinfo {author} {\bibfnamefont {N.~D.}\ \bibnamefont {Mermin}},\
  }\bibfield  {title} {\emph {\bibinfo {title} {Quantum cryptography without
  Bell's theorem}},\ }\href {https://doi.org/10.1103/PhysRevLett.68.557}
  {\bibfield  {journal} {\bibinfo  {journal} {Phys. Rev. Lett.}\ }\textbf
  {\bibinfo {volume} {68}},\ \bibinfo {pages} {557} (\bibinfo {year}
  {1992})}\BibitemShut {NoStop}%
\bibitem [{\citenamefont {Grosshans}\ \emph
  {et~al.}(2003{\natexlab{b}})\citenamefont {Grosshans}, \citenamefont {Cerf},
  \citenamefont {Wenger}, \citenamefont {Tualle-Brouri},\ and\ \citenamefont
  {Grangier}}]{Grosshans2003SR}%
  \BibitemOpen
  \bibfield  {author} {\bibinfo {author} {\bibfnamefont {F.}~\bibnamefont
  {Grosshans}}, \bibinfo {author} {\bibfnamefont {N.~J.}\ \bibnamefont {Cerf}},
  \bibinfo {author} {\bibfnamefont {J.}~\bibnamefont {Wenger}}, \bibinfo
  {author} {\bibfnamefont {R.}~\bibnamefont {Tualle-Brouri}},\ and\ \bibinfo
  {author} {\bibfnamefont {P.}~\bibnamefont {Grangier}},\ }\bibfield  {title}
  {\emph {\bibinfo {title} {Virtual Entanglement and Reconciliation Protocols
  for Quantum Cryptography with Continuous Variables}},\ }\href@noop {}
  {\bibfield  {journal} {\bibinfo  {journal} {Quantum Info. Comput.}\ }\textbf
  {\bibinfo {volume} {3}},\ \bibinfo {pages} {535–552} (\bibinfo {year}
  {2003}{\natexlab{b}})}\BibitemShut {NoStop}%
\bibitem [{\citenamefont {Curty}\ \emph {et~al.}(2004)\citenamefont {Curty},
  \citenamefont {Lewenstein},\ and\ \citenamefont {L\"utkenhaus}}]{Curty2004}%
  \BibitemOpen
  \bibfield  {author} {\bibinfo {author} {\bibfnamefont {M.}~\bibnamefont
  {Curty}}, \bibinfo {author} {\bibfnamefont {M.}~\bibnamefont {Lewenstein}},\
  and\ \bibinfo {author} {\bibfnamefont {N.}~\bibnamefont {L\"utkenhaus}},\
  }\bibfield  {title} {\emph {\bibinfo {title} {Entanglement as a Precondition
  for Secure Quantum Key Distribution}},\ }\href
  {https://doi.org/10.1103/PhysRevLett.92.217903} {\bibfield  {journal}
  {\bibinfo  {journal} {Phys. Rev. Lett.}\ }\textbf {\bibinfo {volume} {92}},\
  \bibinfo {pages} {217903} (\bibinfo {year} {2004})}\BibitemShut {NoStop}%
\bibitem [{\citenamefont {Ferenczi}\ and\ \citenamefont
  {L\"utkenhaus}(2012)}]{Ferenczi2012}%
  \BibitemOpen
  \bibfield  {author} {\bibinfo {author} {\bibfnamefont {A.}~\bibnamefont
  {Ferenczi}}\ and\ \bibinfo {author} {\bibfnamefont {N.}~\bibnamefont
  {L\"utkenhaus}},\ }\bibfield  {title} {\emph {\bibinfo {title} {Symmetries in
  quantum key distribution and the connection between optimal attacks and
  optimal cloning}},\ }\href {https://doi.org/10.1103/PhysRevA.85.052310}
  {\bibfield  {journal} {\bibinfo  {journal} {Phys. Rev. A}\ }\textbf {\bibinfo
  {volume} {85}},\ \bibinfo {pages} {052310} (\bibinfo {year}
  {2012})}\BibitemShut {NoStop}%
\bibitem [{\citenamefont {Devetak}\ and\ \citenamefont
  {Winter}(2005)}]{Devetak2005}%
  \BibitemOpen
  \bibfield  {author} {\bibinfo {author} {\bibfnamefont {I.}~\bibnamefont
  {Devetak}}\ and\ \bibinfo {author} {\bibfnamefont {A.}~\bibnamefont
  {Winter}},\ }\bibfield  {title} {\emph {\bibinfo {title} {Distillation of
  secret key and entanglement from quantum states}},\ }\href
  {https://doi.org/10.1098/rspa.2004.1372} {\bibfield  {journal} {\bibinfo
  {journal} {Proceedings of the Royal Society A: Mathematical, Physical and
  Engineering Sciences}\ }\textbf {\bibinfo {volume} {461}},\ \bibinfo {pages}
  {207} (\bibinfo {year} {2005})}\BibitemShut {NoStop}%
\bibitem [{\citenamefont {Renner}\ and\ \citenamefont
  {Cirac}(2009)}]{Renner2009}%
  \BibitemOpen
  \bibfield  {author} {\bibinfo {author} {\bibfnamefont {R.}~\bibnamefont
  {Renner}}\ and\ \bibinfo {author} {\bibfnamefont {J.~I.}\ \bibnamefont
  {Cirac}},\ }\bibfield  {title} {\emph {\bibinfo {title} {de Finetti
  Representation Theorem for Infinite-Dimensional Quantum Systems and
  Applications to Quantum Cryptography}},\ }\href
  {https://doi.org/10.1103/PhysRevLett.102.110504} {\bibfield  {journal}
  {\bibinfo  {journal} {Phys. Rev. Lett.}\ }\textbf {\bibinfo {volume} {102}},\
  \bibinfo {pages} {110504} (\bibinfo {year} {2009})}\BibitemShut {NoStop}%
\bibitem [{\citenamefont {{Winter}}(1999)}]{Winter1999}%
  \BibitemOpen
  \bibfield  {author} {\bibinfo {author} {\bibfnamefont {A.}~\bibnamefont
  {{Winter}}},\ }\bibfield  {title} {\emph {\bibinfo {title} {Coding theorem
  and strong converse for quantum channels}},\ }\href
  {https://doi.org/10.1109/18.796385} {\bibfield  {journal} {\bibinfo
  {journal} {IEEE Transactions on Information Theory}\ }\textbf {\bibinfo
  {volume} {45}},\ \bibinfo {pages} {2481} (\bibinfo {year}
  {1999})}\BibitemShut {NoStop}%
\bibitem [{\citenamefont {Winter}(2016)}]{Winter2016}%
  \BibitemOpen
  \bibfield  {author} {\bibinfo {author} {\bibfnamefont {A.}~\bibnamefont
  {Winter}},\ }\bibfield  {title} {\emph {\bibinfo {title} {Tight Uniform
  Continuity Bounds for Quantum Entropies: Conditional Entropy, Relative
  Entropy Distance and Energy Constraints}},\ }\href
  {https://doi.org/10.1007/s00220-016-2609-8} {\bibfield  {journal} {\bibinfo
  {journal} {Communications in Mathematical Physics}\ }\textbf {\bibinfo
  {volume} {347}},\ \bibinfo {pages} {291} (\bibinfo {year}
  {2016})}\BibitemShut {NoStop}%
\bibitem [{\citenamefont {Uhlmann}(1976)}]{Uhlmann1976}%
  \BibitemOpen
  \bibfield  {author} {\bibinfo {author} {\bibfnamefont {A.}~\bibnamefont
  {Uhlmann}},\ }\bibfield  {title} {\emph {\bibinfo {title} {The “transition
  probability” in the state space of a *-algebra}},\ }\href
  {https://doi.org/https://doi.org/10.1016/0034-4877(76)90060-4} {\bibfield
  {journal} {\bibinfo  {journal} {Reports on Mathematical Physics}\ }\textbf
  {\bibinfo {volume} {9}},\ \bibinfo {pages} {273 } (\bibinfo {year}
  {1976})}\BibitemShut {NoStop}%
\bibitem [{\citenamefont {{Fuchs}}\ and\ \citenamefont {{van de
  Graaf}}(1999)}]{Fuchs99}%
  \BibitemOpen
  \bibfield  {author} {\bibinfo {author} {\bibfnamefont {C.~A.}\ \bibnamefont
  {{Fuchs}}}\ and\ \bibinfo {author} {\bibfnamefont {J.}~\bibnamefont {{van de
  Graaf}}},\ }\bibfield  {title} {\emph {\bibinfo {title} {Cryptographic
  distinguishability measures for quantum-mechanical states}},\ }\href
  {https://doi.org/10.1109/18.761271} {\bibfield  {journal} {\bibinfo
  {journal} {IEEE Transactions on Information Theory}\ }\textbf {\bibinfo
  {volume} {45}},\ \bibinfo {pages} {1216} (\bibinfo {year}
  {1999})}\BibitemShut {NoStop}%
\bibitem [{\citenamefont {Watrous}(2018)}]{Watrous2018}%
  \BibitemOpen
  \bibfield  {author} {\bibinfo {author} {\bibfnamefont {J.}~\bibnamefont
  {Watrous}},\ }\href {https://doi.org/10.1017/9781316848142} {\emph {\bibinfo
  {title} {The Theory of Quantum Information}}},\ \bibinfo {edition} {1st}\
  ed.\ (\bibinfo  {publisher} {Cambridge University Press},\ \bibinfo {address}
  {Cambridge, UK},\ \bibinfo {year} {2018})\BibitemShut {NoStop}%
\bibitem [{\citenamefont {Kaur}\ \emph {et~al.}(2021)\citenamefont {Kaur},
  \citenamefont {Guha},\ and\ \citenamefont {Wilde}}]{Kaur2021}%
  \BibitemOpen
  \bibfield  {author} {\bibinfo {author} {\bibfnamefont {E.}~\bibnamefont
  {Kaur}}, \bibinfo {author} {\bibfnamefont {S.}~\bibnamefont {Guha}},\ and\
  \bibinfo {author} {\bibfnamefont {M.~M.}\ \bibnamefont {Wilde}},\ }\bibfield
  {title} {\emph {\bibinfo {title} {Asymptotic security of discrete-modulation
  protocols for continuous-variable quantum key distribution}},\ }\href
  {https://doi.org/10.1103/PhysRevA.103.012412} {\bibfield  {journal} {\bibinfo
   {journal} {Phys. Rev. A}\ }\textbf {\bibinfo {volume} {103}},\ \bibinfo
  {pages} {012412} (\bibinfo {year} {2021})}\BibitemShut {NoStop}%
\bibitem [{\citenamefont {Heid}\ and\ \citenamefont
  {L\"utkenhaus}(2006)}]{Heid2006}%
  \BibitemOpen
  \bibfield  {author} {\bibinfo {author} {\bibfnamefont {M.}~\bibnamefont
  {Heid}}\ and\ \bibinfo {author} {\bibfnamefont {N.}~\bibnamefont
  {L\"utkenhaus}},\ }\bibfield  {title} {\emph {\bibinfo {title} {Efficiency of
  coherent-state quantum cryptography in the presence of loss: Influence of
  realistic error correction}},\ }\href
  {https://doi.org/10.1103/PhysRevA.73.052316} {\bibfield  {journal} {\bibinfo
  {journal} {Phys. Rev. A}\ }\textbf {\bibinfo {volume} {73}},\ \bibinfo
  {pages} {052316} (\bibinfo {year} {2006})}\BibitemShut {NoStop}%
\bibitem [{\citenamefont {Lin}\ and\ \citenamefont
  {L\"utkenhaus}(2020)}]{Lin2020}%
  \BibitemOpen
  \bibfield  {author} {\bibinfo {author} {\bibfnamefont {J.}~\bibnamefont
  {Lin}}\ and\ \bibinfo {author} {\bibfnamefont {N.}~\bibnamefont
  {L\"utkenhaus}},\ }\bibfield  {title} {\emph {\bibinfo {title} {Trusted
  Detector Noise Analysis for Discrete Modulation Schemes of
  Continuous-Variable Quantum Key Distribution}},\ }\href
  {https://doi.org/10.1103/PhysRevApplied.14.064030} {\bibfield  {journal}
  {\bibinfo  {journal} {Phys. Rev. Applied}\ }\textbf {\bibinfo {volume}
  {14}},\ \bibinfo {pages} {064030} (\bibinfo {year} {2020})}\BibitemShut
  {NoStop}%
\bibitem [{\citenamefont {Grant}\ and\ \citenamefont {Boyd}(2014)}]{Grant2014}%
  \BibitemOpen
  \bibfield  {author} {\bibinfo {author} {\bibfnamefont {M.}~\bibnamefont
  {Grant}}\ and\ \bibinfo {author} {\bibfnamefont {S.}~\bibnamefont {Boyd}},\
  }\href@noop {} {\bibinfo {title} {{CVX}: Matlab Software for Disciplined
  Convex Programming, version 2.1}},\ \bibinfo {howpublished}
  {\url{http://cvxr.com/cvx}} (\bibinfo {year} {2014})\BibitemShut {NoStop}%
\bibitem [{\citenamefont {Grant}\ and\ \citenamefont {Boyd}(2008)}]{Grant2008}%
  \BibitemOpen
  \bibfield  {author} {\bibinfo {author} {\bibfnamefont {M.~C.}\ \bibnamefont
  {Grant}}\ and\ \bibinfo {author} {\bibfnamefont {S.~P.}\ \bibnamefont
  {Boyd}},\ }\bibfield  {title} {\emph {\bibinfo {title} {Graph Implementations
  for Nonsmooth Convex Programs}},\ }in\ \href@noop {} {\emph {\bibinfo
  {booktitle} {Recent Advances in Learning and Control}}},\ \bibinfo {editor}
  {edited by\ \bibinfo {editor} {\bibfnamefont {V.~D.}\ \bibnamefont
  {Blondel}}, \bibinfo {editor} {\bibfnamefont {S.~P.}\ \bibnamefont {Boyd}},\
  and\ \bibinfo {editor} {\bibfnamefont {H.}~\bibnamefont {Kimura}}}\ (\bibinfo
   {publisher} {Springer London},\ \bibinfo {address} {London, UK},\ \bibinfo
  {year} {2008})\ pp.\ \bibinfo {pages} {95--110}\BibitemShut {NoStop}%
\bibitem [{\citenamefont {ApS}(2016)}]{Mosek2016}%
  \BibitemOpen
  \bibfield  {author} {\bibinfo {author} {\bibfnamefont {M.}~\bibnamefont
  {ApS}},\ }\href {https://www.mosek.com/downloads/8.0.0.60} {\bibinfo {title}
  {The MOSEK optimization toolbox for MATLAB Manual. Version 8.0.0.60.}}
  (\bibinfo {year} {2016})\BibitemShut {NoStop}%
\bibitem [{\citenamefont {Fossier}\ \emph {et~al.}(2009)\citenamefont
  {Fossier}, \citenamefont {Diamanti}, \citenamefont {Debuisschert},
  \citenamefont {Tualle-Brouri},\ and\ \citenamefont {Grangier}}]{Fossier2009}%
  \BibitemOpen
  \bibfield  {author} {\bibinfo {author} {\bibfnamefont {S.}~\bibnamefont
  {Fossier}}, \bibinfo {author} {\bibfnamefont {E.}~\bibnamefont {Diamanti}},
  \bibinfo {author} {\bibfnamefont {T.}~\bibnamefont {Debuisschert}}, \bibinfo
  {author} {\bibfnamefont {R.}~\bibnamefont {Tualle-Brouri}},\ and\ \bibinfo
  {author} {\bibfnamefont {P.}~\bibnamefont {Grangier}},\ }\bibfield  {title}
  {\emph {\bibinfo {title} {Improvement of continuous-variable quantum key
  distribution systems by using optical preamplifiers}},\ }\href
  {https://doi.org/10.1088/0953-4075/42/11/114014} {\bibfield  {journal}
  {\bibinfo  {journal} {Journal of Physics B: Atomic, Molecular and Optical
  Physics}\ }\textbf {\bibinfo {volume} {42}},\ \bibinfo {pages} {114014}
  (\bibinfo {year} {2009})}\BibitemShut {NoStop}%
\bibitem [{\citenamefont {Toh}\ \emph {et~al.}(1999)\citenamefont {Toh},
  \citenamefont {Todd},\ and\ \citenamefont {T\"{u}t\"{u}nc\"{u}}}]{Toh1999}%
  \BibitemOpen
  \bibfield  {author} {\bibinfo {author} {\bibfnamefont {K.~C.}\ \bibnamefont
  {Toh}}, \bibinfo {author} {\bibfnamefont {M.~J.}\ \bibnamefont {Todd}},\ and\
  \bibinfo {author} {\bibfnamefont {R.~H.}\ \bibnamefont
  {T\"{u}t\"{u}nc\"{u}}},\ }\bibfield  {title} {\emph {\bibinfo {title} {SDPT3
  — A Matlab software package for semidefinite programming, Version 1.3}},\
  }\href {https://doi.org/10.1080/10556789908805762} {\bibfield  {journal}
  {\bibinfo  {journal} {Optimization Methods and Software}\ }\textbf {\bibinfo
  {volume} {11}},\ \bibinfo {pages} {545} (\bibinfo {year} {1999})}\BibitemShut
  {NoStop}%
\bibitem [{\citenamefont {T{\"u}t{\"u}nc{\"u}}\ \emph
  {et~al.}(2003)\citenamefont {T{\"u}t{\"u}nc{\"u}}, \citenamefont {Toh},\ and\
  \citenamefont {Todd}}]{Tutuncu2003}%
  \BibitemOpen
  \bibfield  {author} {\bibinfo {author} {\bibfnamefont {R.~H.}\ \bibnamefont
  {T{\"u}t{\"u}nc{\"u}}}, \bibinfo {author} {\bibfnamefont {K.~C.}\
  \bibnamefont {Toh}},\ and\ \bibinfo {author} {\bibfnamefont {M.~J.}\
  \bibnamefont {Todd}},\ }\bibfield  {title} {\emph {\bibinfo {title} {Solving
  semidefinite-quadratic-linear programs using SDPT3}},\ }\href
  {https://doi.org/10.1007/s10107-002-0347-5} {\bibfield  {journal} {\bibinfo
  {journal} {Mathematical Programming}\ }\textbf {\bibinfo {volume} {95}},\
  \bibinfo {pages} {189} (\bibinfo {year} {2003})}\BibitemShut {NoStop}%
\bibitem [{\citenamefont {George}\ \emph {et~al.}(2021)\citenamefont {George},
  \citenamefont {Lin},\ and\ \citenamefont {L\"utkenhaus}}]{George2021}%
  \BibitemOpen
  \bibfield  {author} {\bibinfo {author} {\bibfnamefont {I.}~\bibnamefont
  {George}}, \bibinfo {author} {\bibfnamefont {J.}~\bibnamefont {Lin}},\ and\
  \bibinfo {author} {\bibfnamefont {N.}~\bibnamefont {L\"utkenhaus}},\
  }\bibfield  {title} {\emph {\bibinfo {title} {Numerical calculations of the
  finite key rate for general quantum key distribution protocols}},\ }\href
  {https://doi.org/10.1103/PhysRevResearch.3.013274} {\bibfield  {journal}
  {\bibinfo  {journal} {Phys. Rev. Research}\ }\textbf {\bibinfo {volume}
  {3}},\ \bibinfo {pages} {013274} (\bibinfo {year} {2021})}\BibitemShut
  {NoStop}%
\bibitem [{\citenamefont {Ohya}\ and\ \citenamefont {Petz}(1993)}]{Ohya2004}%
  \BibitemOpen
  \bibfield  {author} {\bibinfo {author} {\bibfnamefont {M.}~\bibnamefont
  {Ohya}}\ and\ \bibinfo {author} {\bibfnamefont {D.}~\bibnamefont {Petz}},\
  }\href@noop {} {\emph {\bibinfo {title} {Quantum Entropy and Its Use}}},\
  Theoretical and Mathematical Physics\ (\bibinfo  {publisher} {Springer-Verlag
  Berlin Heidelberg},\ \bibinfo {address} {New York, USA},\ \bibinfo {year}
  {1993})\BibitemShut {NoStop}%
\bibitem [{\citenamefont {Gerry}\ and\ \citenamefont
  {Knight}(2004)}]{Gerry2004}%
  \BibitemOpen
  \bibfield  {author} {\bibinfo {author} {\bibfnamefont {C.}~\bibnamefont
  {Gerry}}\ and\ \bibinfo {author} {\bibfnamefont {P.}~\bibnamefont {Knight}},\
  }\href {https://doi.org/10.1017/CBO9780511791239} {\emph {\bibinfo {title}
  {Introductory Quantum Optics}}}\ (\bibinfo  {publisher} {Cambridge University
  Press},\ \bibinfo {address} {Cambridge, UK},\ \bibinfo {year}
  {2004})\BibitemShut {NoStop}%
\end{thebibliography}%

\appendix
\onecolumngrid
\section{Uniform Continuity Bound \label{ctybd}}
Here we prove an extension of Lemma 2 in \cite{Winter2016} to subnormalized states. Our development closely parallels that result. Although we are only interested in showing the conditional entropy is uniformly close to decreasing under projection with correction $\Delta$, we will effectively have to derive uniform continuity to determine $\Delta$; so for completeness we give the overall uniform continuity bound as well. Note that the correction term in Eq. \eqref{ucdupcq} is smaller than Eq. \eqref{ctybdcq}.

\begin{thm}[Uniform Continuity and UCDUP of Conditional Entropy]
\label{uc}
Let $\mathcal{H}_A$ and $\mathcal{H}_B$ be two Hilbert spaces where the dimension of $\mathcal{H}_A$ is $\abs{A}$ while $\mathcal{H}_B$ can be infinite-dimensional. Let $\rhot_{AB}, \sigmat_{AB}\in\tilde{D}( \mathcal{H}_A \otimes  \mathcal{H}_B)$ be two subnormalized states; we will omit the system subscripts for readability. WLOG, suppose $\Tr(\rhot)\geq\Tr(\sigmat)$. Let $\frac{1}{2}\norm{\rhot-\sigmat}_1\leq\epsilon\leq1$, $\frac{1}{2}\Tr(\rhot-\sigmat)=\delta$ and $\frac{1}{2}\Tr(\rhot+\sigmat)=a$. Let $\ee=\e+\delta$ and $\eee=\e-\delta$. Then, it holds that
\begin{equation}
\label{ctybdgeneral}
\abs{H(A|B)_{\rhot} -  H(A|B)_{\sigmat}}\leq 2 \e \log_2\abs{A} + (a+\e)\max\left\{h\left(\frac{\ee}{a+\e}\right),h\left(\frac{\eee}{a+\e}\right) \right\}. 
\end{equation}

If $\rhot$ and $\sigmat$ are classical-quantum states, that is $\rhot=\sum_{i=1}^\abs{A} \dyad{i}_A \otimes \rhot^i_B$ and $\sigmat=\sum_{i=1}^\abs{A}\dyad{i}_A \otimes \sigmat^i_B$, then
\begin{equation}
\label{ctybdcq}
\abs{H(A|B)_{\rhot} -  H(A|B)_{\sigmat}}\leq  \ee \log_2\abs{A} + (a+\e)\max\left\{h\left(\frac{\ee}{a+\e}\right),h\left(\frac{\eee}{a+\e}\right) \right\},
\end{equation}
and
\begin{equation}
\label{ucdupcq}
H(A|B)_{\sigmat} -  H(A|B)_{\rhot} \leq  \eee \log_2\abs{A} +(a+\e)h\left(\frac{\eee}{a+\e}\right).
\end{equation}
\end{thm}

\begin{proof}

We can assume $\frac{1}{2}\norm{\rhot-\sigmat}_1=\epsilon$ since our bound will be increasing in $\epsilon$. Note that $\delta\leq\epsilon$. As usual, $\rho$ and $\sigma$ denote the normalized $\rhot$ and $\sigmat$. Let $\cdot_+$ denote the positive part of a Hermitian operator. The proof consists of a series of operator inequalities and applications of strong subadditivity. 

We first determine the trace of the positive and negative parts of $\rhot-\sigmat$. To do this, consider the eigenvalues $\lambda_i$ of $\rhot-\sigmat$. By assumption, $ \sum \abs{\lambda_i}=\norm{\rhot-\sigmat}_1 =2\epsilon$ and $ \sum \lambda_i=\Tr(\rhot-\sigmat)=2\delta$. Thus, $\Tr[(\rhot-\sigmat)_+]=\sum_{\lambda_i\geq0} \lambda_i= \epsilon + \delta=\epsilon'$. Similarly, $\Tr[(\rhot-\sigmat)_-]=-\sum_{\lambda_i<0} \lambda_i= \epsilon - \delta=\epsilon''$.

Thus, $\frac{1}{\ee} (\rhot-\sigmat)_+$ and $\frac{1}{\eee} (\rhot-\sigmat)_-$ are normalized states. Denote them by $\Delta$ and $\Delta'$ respectively. After some rearrangement, we can define a third state $\omega$ satisfying
\begin{equation}
\label{omegasum}
\omega=\frac{\Tr\sigmat}{\Tr\sigmat+\ee}\sigma + \frac{\ee}{\Tr\sigmat+\ee} \Delta =\frac{\Tr\rhot}{\Tr\rhot+\eee}\rho + \frac{\eee}{\Tr\rhot+\eee} \Delta'.
\end{equation}
Note that $\Tr\sigmat+\ee=\Tr\rhot+\eee=a+\e$. We will find an upper and lower bound on $H(A|B)_{\omega}$, and combine them to get our final result.

The lower bound simply follows from the concavity of conditional entropy and the definition of $\omega$ in Eq. \eqref{omegasum},
\begin{equation}
\label{omegaineq2}
H(A|B)_\omega\geq \frac{\Tr\sigmat}{a+\e}H(A|B)_{\sigma}  +  \frac{\ee}{a+\e} H(A|B)_{\Delta}.
\end{equation}

For the upper bound, we first rewrite the conditional entropy in terms of the relative entropy as follows \cite{Ohya2004},
\begin{equation}
\label{relcondentr}
-H(A|B)_{\omega_{AB}}=\min_{\xi_B} D(\omega_{AB}||\mathbbm{1}_A\otimes \xi_B).
\end{equation}
Note that the minimum is achieved at $\xi_B=\omega_B=\Tr_A(\omega_{AB})$. Expanding the definition of the relative entropy, we have
\begin{align}
H(A|B)_\omega&=-D(\omega_{AB}||\mathbbm{1}_A\otimes \omega_B)\\
&=H(\omega)+\Tr[\omega (\mathbbm{1}_A \otimes \log_2\omega_B)].
\end{align}
We upper bound the first term using strong subadditivity, 
\begin{align}
H(\omega)&=H\left(\frac{\Tr \rhot}{a+\e} \rho + \frac{\eee}{a+\e}\Delta'\right)\\
&\leq  \frac{\Tr \rhot}{a+\e}H(\rho)+\frac{\eee}{a+\e}H(\Delta') +  h\left(\frac{\eee}{a+\e}\right).
\end{align}
In the second term, we simply insert the definition of $\omega$ and expand. Thus, we have
\begin{align}
H(A|B)_\omega&\leq  \frac{\Tr \rhot}{a+\e}H(\rho)+\frac{\eee}{a+\e}H(\Delta')+ \frac{\Tr \rhot}{a+\e} \Tr[\rho (\mathbbm{1}_A \otimes \log_2\omega_B)]+\frac{\eee}{a+\e} \Tr[\Delta' (\mathbbm{1}_A \otimes \log_2\omega_B)] + h\left(\frac{\eee}{a+\e}\right) \\
&=- \frac{\Tr \rhot}{a+\e} D(\rho||\mathbbm{1}_A\otimes \omega_B)  - \frac{\eee}{a+\e}D(\Delta'||\mathbbm{1}_A\otimes \omega_B) + h\left(\frac{\eee}{a+\e}\right) ,
\end{align}
where we have recombined the terms into relative entropies. We now use the relation in Eq. \eqref{relcondentr} again, to obtain
\begin{equation}
H(A|B)_\omega\leq \frac{\Tr \rhot}{a+\e} H(A|B)_{\rho}  +  \frac{\eee}{a+\e}H(A|B)_{\Delta'} + h\left(\frac{\eee}{a+\e}\right).
\label{omegaineq1}
\end{equation}

The upper and lower bounds on $H(A|B)_\omega$, in Eq. \eqref{omegaineq1} and Eq. \eqref{omegaineq2} respectively, can be combined to obtain
\begin{align}
\frac{\Tr\sigmat}{a+\e}H(A|B)_{\sigma} + \frac{\ee}{a+\e} H(A|B)_{\Delta} &\leq \frac{\Tr \rhot}{a+\e} H(A|B)_{\rho} + \frac{\eee}{a+\e}H(A|B)_{\Delta'}+h\left(\frac{\eee}{a+\e}\right),\\
H(A|B)_{\sigmat} -  H(A|B)_{\rhot} &\leq  \eee H(A|B)_{\Delta'} -  {\ee} H(A|B)_{\Delta}+(a+\e)h\left(\frac{\eee}{a+\e}\right).
\end{align}

By repeating the proof but interchanging the two expressions for $\omega$, we similarly obtain 
\begin{equation}
H(A|B)_{\rhot} -  H(A|B)_{\sigmat} \leq  \ee H(A|B)_{\Delta} -  {\eee} H(A|B)_{\Delta'}+(a+\e)h\left(\frac{\ee}{a+\e}\right).
\end{equation}

Conditional entropies of normalized states are bounded between $\pm\log_2\abs{A}$. Thus, we have
\begin{equation}
\abs{H(A|B)_{\rhot} -  H(A|B)_{\sigmat}}\leq 2 \e \log_2\abs{A} + (a+\e)\max\left\{h\left(\frac{\ee}{a+\e}\right),h\left(\frac{\eee}{a+\e}\right) \right\}. 
\end{equation}

When $\rhot$ and $\sigmat$ are both classical-quantum states, $\Delta$ and $\Delta'$ are also both classical-quantum states. Then, their conditional entropy is between 0 and $\log_2\abs{A}$. This gives the tighter bound of 
\begin{equation}
\abs{H(A|B)_{\rhot} -  H(A|B)_{\sigmat}}\leq  \ee \log_2\abs{A} + (a+\e)\max\left\{h\left(\frac{\ee}{a+\e}\right),h\left(\frac{\eee}{a+\e}\right) \right\}. 
\end{equation}
Similarly, 
\begin{equation}
H(A|B)_{\sigmat} -  H(A|B)_{\rhot} \leq  \eee \log_2\abs{A} +(a+\e)h\left(\frac{\eee}{a+\e}\right).
\end{equation}
\end{proof}

\begin{cor}
Let $\rhot_{AB}$ and $\sigmat_{AB}$ be two bipartite subnormalized classical-quantum states with $\Tr(\rhot)\geq\Tr(\sigmat)$; the dimension of system B can be infinite. Let $\frac{1}{2}\norm{\rhot-\sigmat}_1\leq\epsilon\leq1$. Then, 
\begin{equation}
H(A|B)_{\sigmat} -  H(A|B)_{\rhot} \leq  \e \log_2\abs{A} +(1+\e)h\left(\frac{\e}{1+\e}\right).
\end{equation}
\end{cor}
\begin{proof}
Begin with the third statement of Theorem \ref{uc}. We can upper bound $\eee$ in the first term on the right-hand side by $\e$. Then, since the function $g(a)=(a+\e) h\left(\frac{c}{a+\e}\right)$ is increasing on $a\in[0,1]$, we can upper bound the second term on the right-hand side by evaluating it at $a=1$. We have
\begin{gather}
\frac{\eee}{1+\e} \leq \frac{\e}{1+\e} \leq \frac{1}{2}.
\end{gather}
Since the binary entropy is increasing on $[0,\frac{1}{2}]$,
\begin{equation}
h\left(\frac{\eee}{1+\e}\right)  \leq h\left(\frac{\e}{1+\e}\right).
\end{equation}
Thus we can replace $\eee$ with $\e$ in the second term as well. This leaves us with the desired expression.
\end{proof}

\section{Matrix Operations in Displaced Basis}\label{matrixops}
Recall our basis is $\{\ket{i}_A\otimes\ket{n_{\beta_i}}_B\}$. We calculate the matrix elements of certain operators in this basis and evaluate the action of relevant channels.
\subsection{Operators}
Our constraint operators take a particularly simple form in the displaced basis. The matrix elements are
\begin{align}
\bra{i}\bra{m_{\beta_i}} \left(\dyad{k} \otimes \obsnk \right)  \ket{j} \ket{n_{\beta_j}}&=\delta_{ik}\delta_{jk} \bra{m_{\beta_k}}  \obsnk \ket{n_{\beta_k}}\\
&=\delta_{ik}\delta_{jk} \delta_{mn} n.
\end{align}
Similarly, for $\obsnsqk$, they are  $\delta_{ik}\delta_{jk} \delta_{mn} n^2$.
The key map POVMs are more complicated. Recall the POVM elements (Eq. \eqref{dmcvqkdmap}) are
\begin{equation}
P^k=\mathbbm{1}_A \otimes R^k_B,
\end{equation}
where $R^k_B$ are the region operators for the non-discarded signals.
The matrix elements are 
\begin{align}
P^k_{ijmn}&=\bra{i}\bra{m_{\beta_i}}\left(\mathbbm{1}_A \otimes R_B^k \right) \ket{j} \ket{n_{\beta_j}}\\
&=\delta_{ij} \bra{m_{\beta_i}} R_B^k \ket{n_{\beta_j}}\\
&=\bra{m_{\beta_i}} R_B^k \ket{n_{\beta_i}} \label{idealelement} \\
&=\frac{1}{\pi} \int_{\Delta_a}^\infty \int_{\frac{(2k-1)\pi}{4}+\Delta_p}^{\frac{(2k+1)\pi}{4}-\Delta_p} r e^{-\abs{\kappa}^2} \frac{\kappa^m \kappa^{*n}}{\sqrt{m!n!}} d\theta dr,
\end{align}
where $\kappa=r e^{i\theta}-\beta_i$. This integral is computed in \textsc{Matlab}. 

\subsection{Channels}
Our basis for the bipartite Hilbert space is not of the form $\ket{i}_A \otimes \ket{j}_B$, where $\ket{i}_A$ and $\ket{j}_B$ are bases for $\mathcal{H}_A$ and $\mathcal{H}_B$ respectively. Matrix multiplication proceeds as normal, since we simply have some orthonormal basis. However, operations that care about subsystems, namely the partial trace and its adjoint, have a different matrix representation than the typical presentation. We have 
\begin{equation}
\rho_{AB}= \sum_{i, j, m, n} c_{ijmn} \dyad{i}{j} \otimes \dyad{m_{\beta_i}}{n_{\beta_j}},
\end{equation}
where the coefficients $c$ are the matrix elements of $\rho$. We denote this matrix by $M_\rho$;
\begin{equation}
M_\rho=\sum_{i, j, m, n} c_{ijmn} \dyad{i}{j} \otimes \dyad{m}{n}.
\end{equation}
The reduced density matrix is
\begin{equation}
\rho_{A}= \sum_{i, j, m, n} c_{ijmn} \dyad{i}{j} \braket{n_{\beta_j}}{m_{\beta_i}}.
\end{equation}
Defining
\begin{equation}
G= \sum_{i, j, m, n}\braket{n_{\beta_j}}{m_{\beta_i}}  \dyad{i}{j} \otimes \dyad{m_{\beta_i}}{n_{\beta_j}},
\end{equation}
we have that 
\begin{equation}
\braket{i}{\rho_A|j}=\rho_{ij} \odot G_{ij}
\end{equation}
where the subscripts on the bipartite operators indicate the respective block matrix, and $\odot$ is the element-wise dot product. Note that each $G_{ij}$ can be thought of as a basis change unitary in $\mathcal{H}_B$. An explicit formula for the elements of $G$ is
\begin{align}
\braket{n_{\beta_j}}{m_{\beta_i}} &= \bra{n}D^\dagger(\beta_j) D(\beta_i) \ket{m}\\
&=\exp(i \Im(-\beta_j\beta_i^*)) \bra{n}D(\beta_i-\beta_j) \ket{m}\\
\begin{split}
&=\exp(i \Im(-\beta_j\beta_i^*)-\frac{\abs{\beta_i-\beta_j}^2}{2}) \sqrt{m!n!}  \  \sum_{k=0}^{\min(m,n)} \frac{1}{k!(m-k)!(n-k)!} (\beta_i-\beta_j)^{n-k} (\beta_j^*-\beta_i^*)^{m-k}.
\end{split}
\end{align}
We compute and store this matrix once at the beginning of the optimization algorithm, and use it each time to calculate the partial trace.

The adjoint of the partial trace also has a matrix representation involving $G$. The adjoint of the partial trace is $\xi(\sigma_A)=\sigma_A \otimes \mathbbm{1}_B$. Letting
\begin{equation}
\sigma_A = \sum_{ij} c_{ij} \dyad{i}{j},
\end{equation}
we seek $d_{ijmn}$ such that 
\begin{equation}
\sigma_A \otimes \mathbbm{1}_B= \sum_{ijmn} d_{ijmn}  \dyad{i}{j} \otimes \dyad{m_{\beta_i}}{n_{\beta_j}}.
\end{equation}
This implies 
\begin{equation}
\sum_{mn} d_{ijmn} \dyad{m_{\beta_i}}{n_{\beta_j}} = c_{ij} \mathbbm{1}_B \quad \forall i,j.
\end{equation}
Taking the bra-ket on both sides, we obtain $d_{ijmn}=c_{ij} \braket{m_{\beta_i}}{n_{\beta_j}}$. We recognize the factor on the right-hand side as $G^*$. Thus, we have that 
\begin{equation}
\xi(\sigma_A)=\sum_{ijmn} c_{ij} \dyad{i}{j} \otimes G^*_{ij}. 
\end{equation}

\section{Numeric Framework Modification}\label{epsilonchange}
In principle, the numerical framework presented in \cite{Winick2018} is tight. We observe the following issue in practice. The near-optimal $\rho$ computed in the first step often has constraint violations, due to the inherent imprecision of convex solvers. At all distances and values of excess noise, and for our particular implementation using \textsc{Matlab} and CVX with the \textsc{Mosek} solver, these violations are typically $10^{-7}-10^{-6}$. At distances approaching $200$ km, the simulated expectation values $\expn$ and $\expnsq$ are both small; approximately $10^{-6}$ for the nonzero values of excess noise we consider. Since the constraint violation is the same order of magnitude as the expectation, the first step solution is effectively an optimal state for double the excess noise. Given the poor scaling of the protocol with excess noise, this implies the approximate key rate from this first step will be much lower than its theoretical value. Since the first step upper bounds the second step, this implies a poor second step result. One way to ameliorate this is to solve with a smaller $N$ so that the solver returns a better first step solution. Thus, purely due to numerical precision issues, solving with $N=30$ instead of $N=40$ can improve the key rate at long distances. Even though the correction term $\Delta(W)$ is slightly larger, this is more than offset by the improved quality of the first step solution. This suggests that due to numerical issues, one should choose the finite dimension carefully, even though analytically a larger dimension is always better.

The reason the first step upper bounds the second step is due to the expansion of the feasible set. Referring to the notation in Appendix D of \cite{Winick2018}, the large constraint violations lead to a large value of $\epsilon'$, which controls how much the set is expanded for the second step. However, the choice $\epsilon'=\max(\epsilon_{rep},\epsilon_{sol})$ (Equation (165) of \cite{Winick2018}) is pessimistic. One only needs to choose $\epsilon'=\epsilon_{rep}$. As noted in Equation (162) of \cite{Winick2018}, this is sufficient to provide a reliable lower bound when accounting for finite numerical precision. Further, note that $f(\rho)$ is lower bounded by a tangent hyperplane at any point in its domain. Thus, it is not necessary to expand the feasible set further to include the point returned by the first step. This change gives improved results in practice, while still being reliable and tight.

In previous work using the numerical framework in \cite{Winick2018}, it has been assumed that $\epsilon_{rep}\leq\epsilon_{sol}$. Hence the issue of how to suitably choose $\epsilon_{rep}$ has not been considered. As noted in \cite{Winick2018}, rigorously determining $\epsilon_{rep}$ for a particular implementation can be an involved process. For our \textsc{Matlab} implementation, which has precision better than $10^{-15}$, we conservatively use $10^{-10}$ for $\epsilon_{rep}'$ and all elements of $\vec{\epsilon}_{rep}$. Finally note that in our numerical evaluation of the unbalanced phase-encoded BB84 protocol, we continue to use the original, larger set expansion as in \cite{Winick2018}. This is to ensure a fair comparison to the flag-state squasher numerical results.

\section{Trusted Detector Noise Operators}
\label{tnqo}
We determine expressions for the relevant operators in the trusted detector noise scenario, focusing on the case where both homodyne detectors in the overall heterodyne setup have the same efficiency $\eta_d$ and electronic noise $\nu_{el}$. The only change in the protocol for the trusted noise scenario is Bob's POVM. The changed POVM enters the optimization in two different ways: through Bob's new observables in the constraints and new region operators in the objective function definition.

Recall Bob's POVM in the ideal case is a projection onto coherent states $\{\frac{1}{\pi} \dyad{\zeta} \}_{\zeta\in\mathbbm{C}}$. In the trusted noise scenario, it is instead a projection onto scaled and displaced thermal states $\left\lbrace \frac{1}{\eta_d \pi} \disp{\frac{\zeta}{\sqrt{\eta_d}}} \rho_{th} (\bar{n}) \dispd{\frac{\zeta}{\sqrt{\eta_d}}} \right\rbrace _{\zeta\in\mathbbm{C}}$ where the mean photon number of the thermal state is $\bar{n}=\frac{1-\eta_d+\nu_{el}}{\eta_d}$. We will simplify our notation by omitting the dependence of $\rho_{th}$ on $\bar{n}$. We denote the POVM elements by $G_\zeta$. Recall our notation $\noisy{\cdot}$ for the noisy version of an operator. From \cite{Lin2020}, if an operator in the ideal detector model is defined as
\begin{equation}
\label{idealf}
X=\int_{\zeta\in\mathbbm{C}} f_X(\zeta) \ \frac{1}{\pi} \dyad{\zeta} \ d^2 \zeta,
\end{equation}
then its noisy counterpart is
\begin{equation}
\label{noisyf}
\noisy{X}=\int_{\zeta\in\mathbbm{C}} f_X(\zeta) \ G_\zeta \ d^2\zeta,
\end{equation}
where $d^2\zeta=d \Re(\zeta) \ d \Im(\zeta)$.
\subsection{Objective Function}
Recall the region operators are defined as $R^j_B=\frac{1}{\pi}\int_{A^j}  \dyad{\zeta} d^2 \zeta$, where $A^j$ are the regions in phase-space in Fig. \ref{keymap}. By definition, the noisy region operators are then
\begin{equation}
\noisy{R^j_B}=\int_{A^j}  G_\zeta \ d^2 \zeta,
\end{equation}
and the noisy POVM is $\noisy{P^k}=\mathbbm{1}_A \otimes \noisy{R^k_B}$. Referring to Eq. \eqref{idealelement}, the matrix elements are
\begin{align}
\noisy{P^k}_{ijmn}&=\bra{m_{\beta_i}} \noisy{R_B^k} \ket{n_{\beta_i}},\\
&= \int_{A^k}  \bra{m_{\beta_i}} G_\zeta  \ket{n_{\beta_i}}d^2 \zeta,\\
&=\frac{1}{\eta_d \pi} \int_{A^k} \bra{m} D\left(\frac{\zeta}{\sqrt{\eta_d}}-\beta_i\right) \rho_{th} D^\dagger\left(\frac{\zeta}{\sqrt{\eta_d}}-\beta_i\right)  \ket{n} d^2 \zeta.
\end{align}
We use the expression in Equation (B1) in \cite{Lin2020} for the matrix elements of the displaced thermal operator in the Fock basis, convert to polar coordinates, and compute the integral in \textsc{Matlab}.

\subsection{Observables}
To express the observables in the form of Eq. \eqref{idealf}, we write them in antinormal ordering, and replace the ladder operators $\lowering, \raising$ with $\zeta,\zeta^*$,
\begin{align}
\hat{n}= \lowering \raising -1 &\implies f_{\hat{n}} (\zeta)=\abs{\zeta}^2 -1, \\
\hat{n}^2= {\lowering}^2 (\raising)^2- 3\hat{n} - 2 &\implies f_{\hat{n}^2} (\zeta)=\abs{\zeta}^4 - 3\abs{\zeta}^2+1.
\end{align}
To find $f$ for the displaced observables, we simply perform a change of variables,
\begin{align}
\hat{n}_\beta &= \disp{\beta} \left( \frac{1}{\pi} \int (\abs{\zeta}^2 -1) \dyad{\zeta} d^2 \zeta\right) \dispd{\beta} \\
&= \frac{1}{\pi} \int (\abs{\zeta}^2 -1) \dyad{\zeta+\beta} d^2\zeta \\
&= \frac{1}{\pi} \int (\abs{\zeta-\beta}^2 -1) \dyad{\zeta} d^2\zeta,
\end{align}
and similarly for $\hat{n}_\beta^2$. Thus,
\begin{align}
\label{dispf}
f_{\hat{n}_\beta} (\zeta)=\abs{\zeta-\beta}^2 -1
\end{align}
and
\begin{align}
\label{dispsqf}
f_{\hat{n}^2_\beta} (\zeta)=\abs{\zeta-\beta}^4 - 3\abs{\zeta-\beta}^2+1.
\end{align}

We can now calculate the noisy observables using Eq. \eqref{noisyf}. We make use of the following identity,
\begin{equation}
\braket{\gamma}{\rho_{th}(\bar{n}) |\gamma}=\frac{e^{-\abs{\gamma}^2/(1+\bar{n})}}{1+\bar{n}}.
\end{equation}
By definition,
\begin{align}
\noisy{\hat{n}_\beta}&= \int (\abs{\zeta-\beta}^2-1) G_\zeta d^2 \zeta\\
&=\frac{\eta_d}{\pi} \int \left(\abs{\zeta-\frac{\beta}{\sqrt{\eta_d}}}^2-\frac{1}{\eta_d}\right) D(\zeta) \rho_{th} D^\dagger(\zeta) \ d^2 \zeta.
\end{align}
Then,
\begin{align}
\bra{\alpha}\noisy{\hat{n}_\beta} \ket{\alpha}&= \frac{\eta_d}{\pi}\int \left(\abs{\zeta-\beta'}^2-\frac{1}{\eta_d}\right)  \bra{\alpha-\zeta} \rho_{th} \ket{\alpha-\zeta} d^2 \zeta\\
&= \frac{\eta_d}{\pi(1+\bar{n})}\int \left(\abs{\zeta+\alpha-\beta'}^2-\frac{1}{\eta_d}\right)  e^{\frac{-\abs{\zeta}^2}{1+\bar{n}}} d^2 \zeta.
\end{align}
where $\beta'=\beta/\sqrt{\eta_d}$. Converting to polar coordinates, the integral is
\begin{align}
\bra{\alpha}\noisy{\hat{n}_\beta} \ket{\alpha} &=\frac{\eta_d}{\pi(1+\bar{n})}\int \left(r^2 + \gamma^* r e^{i\theta} +\gamma r e^{-i\theta}+ \abs{\gamma}^2-\frac{1}{\eta_d} \right) \ e^{\frac{-r^2}{1+\bar{n}}}\ r  \ dr \ d\theta\\
&= \eta_d \abs{\gamma}^2 +\nu_{el},
\end{align}
where $\gamma=\alpha-\beta'$. By the uniqueness of the Husimi Q-function,
\begin{equation}
\noisy{\hat{n}_\beta}=\eta_d \hat{n}_{\frac{\beta}{\sqrt{\eta_d}}} + \nu_{el} \mathbbm{1}.
\end{equation}

Similarly,
\begin{align}
\noisy{\hat{n}_\beta^2} &=  \frac{1}{\pi \eta_d} \int \left(\abs{\zeta-\beta}^4-3\abs{\zeta-\beta}^2+1\right)\disp{\frac{\zeta}{\sqrt{\eta_d}}} \rho_{th} \dispd{\frac{\zeta}{\sqrt{\eta_d}}}  d^2 \zeta\\
&=  \frac{1}{\pi} \int \left(\eta_d^2\abs{\zeta-\beta'}^4-3\eta_d\abs{\zeta-\beta'}^2+1\right)\disp{\zeta} \rho_{th} \dispd{\zeta}  d^2 \zeta.
\end{align}
The Q-function is then
\begin{align}
\bra{\alpha} \noisy{\hat{n}_\beta^2} \ket{\alpha} &=\frac{1}{\pi(1+\bar{n})} \int \left(\eta_d^2\abs{\zeta+\alpha-\beta'}^4-3\eta_d\abs{\zeta+\alpha-\beta'}^2+1\right) e^{\frac{-\abs{\zeta}^2}{1+\bar{n}}} d^2 \zeta\\
\begin{split}
&=\frac{1}{\pi(1+\bar{n})} \int \left(\eta_d^2(r^4+4r^2\abs{\gamma}^2+\abs{\gamma}^4)-3\eta_d(r^2+\abs{\gamma}^2)  +1\right) e^{\frac{-r^2}{1+\bar{n}}} \ r \ dr \ d\theta
\end{split}\\
\begin{split}
&=2\left(1+2\nu_{el}+\nu_{el}^2+2\eta_d (1+\nu_{el})\abs{\gamma}^2-3\frac{1+\nu_{el}}{2}+\eta_d^2\abs{\gamma}^4\frac{1}{2}-\frac{3}{2}\eta_d\abs{\gamma}^2+\frac{1}{2}\right)
\end{split}\\
&=\eta_d^2\abs{\gamma}^4+\eta_d (4\nu_{el}+1)\abs{\gamma}^2 +2\nu_{el}^2+\nu_{el}.
\end{align}
Again, by the uniqueness of the Husimi Q-function,
\begin{equation}
\noisy{\hat{n}_\beta^2}=\eta_d^2  \hat{n}_{\frac{\beta}{\sqrt{\eta_d}}}^2+\eta_d (4\nu_{el}+1-\eta_d) \hat{n}_{\frac{\beta}{\sqrt{\eta_d}}} +(2\nu_{el}^2+\nu_{el})\mathbbm{1}.
\end{equation}

\section{Simulated Expectations and Error-Correction Cost}\label{expectationapp}
We discuss how the coarse-grained expectations can be determined from a heterodyne measurement, and what the expectation values are for the simulation. We focus on the trusted detector noise scenario, as the ideal detector results can be recovered as a special case. 

Bob's measurement results determine a probability density $p(\zeta)=\Tr(\rho  G_\zeta)$ over the complex plane. In general, given an observable $\Gamma=\int_{\zeta\in\mathbbm{C}} f_\Gamma(\zeta) \ G_\zeta \ d^2\zeta$, the expectation value is then
\begin{equation}
\label{expectationintegral}
\Tr(\rho \Gamma)= \int_{\zeta\in\mathbbm{C}} f_\Gamma(\zeta) \ p(\zeta) \ d^2\zeta.
\end{equation}
For the $i^{th}$ conditional state, Bob's coarse-grained observables are $\obsn$ and $\obsnsq$. The corresponding functions $f_{\obsn}$ and $f_{\obsnsq}$ are given in Eqs. {\eqref{dispf}} and {\eqref{dispsqf}}. Using his measurement result $p(\zeta)$, Bob can thus compute the integral {\eqref{expectationintegral}} to determine the desired expectations for each conditional state.

Note that for the typical quadratures $X$ and $P$, we have that 
\begin{align}
X = \frac{1}{\sqrt{2}} (\raising+\lowering) \implies f_{X} (\zeta)= \sqrt{2} \Re(\zeta), \\
P = \frac{i}{\sqrt{2}} (\raising-\lowering) \implies f_{P} (\zeta)= \sqrt{2} \Im(\zeta).
\end{align}
Thus, by expanding $f_{\obsn}(\zeta)$ and $f_{\obsnsq}(\zeta)$ as polynomials in $\Re(\zeta)$ and $\Im(\zeta)$, we can also relate the expectations of $\obsn$ and $\obsnsq$ to the moments and cross-terms of the measurement data of quadratures $X$ and $P$.

We now consider the expectations under the simulated channel model. After passing through a Gaussian channel with loss $\eta$ and excess noise $\xi$, a coherent signal state becomes a displaced thermal state
\begin{equation}
\dyad{\alpha_i} \rightarrow \disp{\beta_i} \rho_{th}(\tfrac{\delta}{2}) \dispd{\beta_i}
\end{equation}
where $\delta=\eta \xi$ and $\beta_i=\sqrt{\eta} \alpha_i$. The expectation values for each conditional state are straightforward to calculate,
\begin{equation}
\Tr(\obsn \disp{\beta_i} \rho_{th}(\tfrac{\delta}{2}) \dispd{\beta_i})=\Tr(\hat{n} \ \rho_{th}(\tfrac{\delta}{2}) ) =\frac{\delta}{2}
\end{equation}
and
\begin{equation}
\Tr(\obsnsq \disp{\beta_i} \rho_{th}(\tfrac{\delta}{2}) \dispd{\beta_i})=\frac{\delta(1+\delta)}{2}.
\end{equation}

For the reduced state constraint, we use the well-known formula for the overlap of two coherent states \cite{Gerry2004}
\begin{equation}
\braket{\alpha_j}{\alpha_i}=\exp(i \Im(\alpha_i \alpha_j^*)-\frac{1}{2} \abs{\alpha_i-\alpha_j}^2 ).
\end{equation}

The error-correction cost is determined by the simulated joint probability distribution. Given Alice prepares $\ket{\alpha_i}$, the probability Bob gets the key map outcome $j$, with $i\in\{0,1,2,3\}$ and $j\in\{0,1,2,3,\perp\}$, is given by the following integral
\begin{align}
p(j|i)&=\Tr(R^j \disp{\beta_i} \rho_{th}(\tfrac{\delta}{2}) \dispd{\beta_i} )\\
 &= \int_{A^j} \Tr(G_\zeta \disp{\beta_i} \rho_{th}(\tfrac{\delta}{2}) \dispd{\beta_i}) d^2 \zeta.\label{probintegral}
\end{align}
The integrand, which is the overlap of two displaced thermal states, is given by \cite{Lin2020}
\begin{equation}
\Tr(G_\zeta \disp{\beta_i} \rho_{th}(\tfrac{\delta}{2}) \dispd{\beta_i}) = \frac{1}{\pi(1+\frac{1}{2}\eta_d \delta+\nu_{el})}\exp(\frac{-\abs{\zeta-\sqrt{\eta_d}\beta_i}^2}{1+\frac{1}{2}\eta_d \delta+\nu_{el}}).
\end{equation}
The integral in Eq. \eqref{probintegral} is converted to polar coordinates and computed in \textsc{Matlab}. As the signal states are distributed uniformly, $p_A(i)=\frac{1}{4}$. Then, $p_{AB}(i,j)=\frac{1}{4}p(j|i)$. Since discarded signals do not incur an error-correction cost, we remove the outcome $j=\perp$ and renormalize $p$ accordingly. Denoting this sifted probability distribution by $q$, the error-correction cost is
\begin{equation}
\delta^{leak}_{EC}= 2-\beta_{EC} [H(q_A)+H(q_B)-H(q_{AB})].
\end{equation}

\end{document}